\documentclass[a4paper]{article}

\usepackage{microtype}
\usepackage{amsmath}
\usepackage{bbm,bm}
\usepackage{amssymb}
\usepackage{amsthm}
\usepackage{graphicx}
\usepackage{mathtools}
\usepackage{authblk}

\newcommand\blfootnote[1]{%
  \begingroup
  \renewcommand\thefootnote{}\footnote{#1}%
  \addtocounter{footnote}{-1}%
  \endgroup
}

\title{Understanding the Topology and the Geometry of the Space of Persistence Diagrams via Optimal Partial Transport}
\author{Vincent Divol}
\author{Th\'eo Lacombe}
\affil{\texttt{firstname.lastname@inria.fr} \\ Datashape, Inria Saclay}

\date{}

\newtheorem{definition}{Definition}[section]
\newtheorem{theorem}{Theorem}[section]
\theoremstyle{plain}
\newtheorem{proposition}[theorem]{Proposition} 
\newtheorem{lemma}{Lemma}[section]
\newtheorem{corollary}{Corollary}[section]
\newtheorem{remark}{Remark}[section]
\newtheorem*{theorem*}{Theorem}

\newcommand{\R}{\mathbb{R}}
\newcommand{\N}{\mathbb{N}}

\newcommand{\Z}{\mathbb{Z}}
\newcommand{\X}{\mathbb{X}}
\newcommand{\WW}{\mathcal{W}}
\newcommand{\XX}{\mathcal{X}}
\newcommand{\DD}{\mathcal{D}}

\renewcommand{\P}{P}
\newcommand{\EE}{\mathcal{E}}
\newcommand{\E}{\mathbb{E}}

\newcommand{\MM}{\mathcal{M}}

\newcommand{\FF}{\mathcal{F}}

\renewcommand{\SS}{\mathcal{S}}
\newcommand{\BB}{\mathcal{B}}

\newcommand{\dd}{\mathrm{d}}
\newcommand{\dgm}{\mathrm{Dgm}}

\newcommand{\Pers}{\mathrm{Pers}}  
\newcommand{\diam}{\mathrm{diam}}
\newcommand{\upperdiag}{\Omega} 
\newcommand{\groundspace}{\overline{\upperdiag}} 
\newcommand{\thediag}{{\partial \Omega}} 
\newcommand{\mtot}{m_{\mathrm{tot}}}
\newcommand{\supp}{\mathrm{spt}}

\newcommand{\defeq}{\vcentcolon=}

\newcommand{\Adm}{\mathrm{Adm}}
\newcommand{\Opt}{\mathrm{Opt}}

\newcommand{\emptydgm}{0}  
\newcommand{\cvvague}{\xrightarrow{v}}
\newcommand{\cvweak}{\xrightarrow{w}}
\newcommand{\id}{\mathrm{id}}
\renewcommand{\epsilon}{\varepsilon}
\newcommand{\ones}{\mathbf{1}}

\newcommand{\Dp}{\mathrm{OT}_p}
\newcommand{\PD}{\DD}
\newcommand{\Dinf}{\mathrm{OT}_\infty}
\newcommand{\Cp}{\mathcal{C}^0_{b,p}}

\newcommand{\repeatresult}[2]{
	\vspace{.2cm}
	\noindent
	\textbf{#1}
	\emph{#2}
	\vspace{.2cm}
}

\DeclareMathOperator*{\argmin}{arg\,min}

\DeclareMathOperator*{\minimize}{minimize}

\begin{document}

\maketitle

\begin{abstract}
Despite the obvious similarities between the metrics used in topological data analysis and those of optimal transport, an optimal-transport based formalism to study persistence diagrams and similar topological descriptors has yet to come. In this article, by considering the space of persistence diagrams as a space of discrete measures, and by observing that its metrics can be expressed as optimal partial transport problems, we introduce a generalization of persistence diagrams, namely Radon measures supported on the upper half plane. Such measures naturally appear in topological data analysis when considering continuous representations of persistence diagrams (e.g.\ persistence surfaces) but also as limits for laws of large numbers on persistence diagrams or as expectations of probability distributions on the persistence diagrams space. We explore topological properties of this new space, which will also hold for the closed subspace of persistence diagrams. New results include a characterization of convergence with respect to Wasserstein metrics, a geometric description of barycenters (Fr\'echet means) for any distribution of diagrams, and an exhaustive description of continuous linear representations of persistence diagrams. We also showcase the strength of this framework to study random persistence diagrams by providing several statistical results made meaningful thanks to this new formalism.
 \end{abstract}

\section{Introduction}

\subsection{Framework and motivations}
\blfootnote{This updated version fixes a minor mistake found in the previous version: the space of Radon measure on $E_\Omega$ was incorrectly endowed with the vague topology, instead of a stronger topology, that we call the VM topology. This modification only impacts the statements of   Proposition \ref{prop:basic_prop} and Proposition \ref{prop:basic_infty}; and the conclusions of all our results remain unchanged.}Topological Data Analysis (TDA) is an emerging field in data analysis that has found applications in computer vision \cite{tda:li2014computervision}, material science \cite{tda:hiraoka2017materialscience,tda:kramar2013materialscience}, shape analysis \cite{tda:carriere2015shape3d,tda:turner2014shapePHT}, to name a few. The aim of TDA is to provide interpretable descriptors of the underlying topology of a given object. One of the most used (and theoretically studied) descriptors in TDA is the \emph{persistence diagram}. This descriptor consists in a locally finite multiset of points in the upper half plane $\upperdiag \defeq \{ (t_1, t_2) \in \R^2, t_2 > t_1\}$, each point in the diagram corresponding informally to the presence of a topological feature (connected component, loop, hole, etc.) appearing at some scale in the \emph{filtration} of an object $\X$. A complete description of the \emph{persistent homology} machinery is not necessary for this work and the interested reader can refer to \cite{tda:edelsbrunner2010computational} for an introduction. The space of persistence diagrams, denoted by $\PD$ in the following, is usually equipped with partial matching extended metrics $d_p$ (i.e.~it can happen that $d_p(\mu,\nu) = +\infty$ for some $\mu,\nu \in \PD$), sometimes called Wasserstein distances \cite[Chapter VIII.2]{tda:edelsbrunner2010computational}: for $p \in [1,+\infty)$ and $a,b$ in $\PD$, define
\begin{equation}
	d_p(a, b) \defeq \left(\inf_{\pi \in \Gamma(a, b)} \sum_{x \in a \cup \thediag} d(x,\pi(x))^p \right)^{\frac{1}{p}},
	\label{eq:dgm_partial_matching_metric}
\end{equation}
where $d(\cdot,\cdot)$ denotes the $q$-norm on $\R^2$ for some $1 \leq q \leq \infty$, $\Gamma(a,b)$ is the set of partial matchings between $a$ and $b$, i.e.\ bijections between $a \cup \thediag$ and $b \cup \thediag$, and $\thediag \defeq \{ (t,t),\ t \in \R \}$ is the boundary of $\upperdiag$, namely the diagonal (see Figure \ref{fig:pers_diag}). When $p \to \infty$, we recover the so-called \emph{bottleneck distance}:
\begin{equation}
	d_\infty(a,b) \defeq \inf_{\pi \in \Gamma(a,b)} \sup_{x \in a \cup \thediag} d(x,\pi(x)).
	\label{eq:dgm_bottleneck_metric}
\end{equation}

\begin{figure}[h]
	\center
	\includegraphics[width = 0.4 \linewidth]{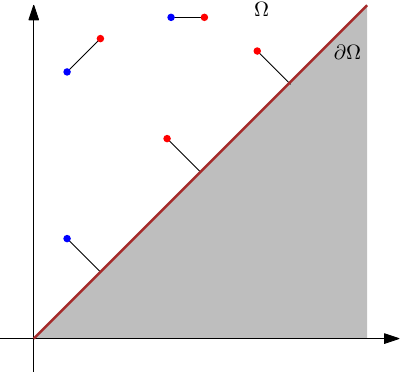}
	\caption{An example of optimal partial matching between two diagrams. The bottleneck distance between these two diagrams is the length of the longest edge in this matching, while their Wasserstein distance $d_p$ is the $p$-th root of the sum of all edge lengths to the power $p$.}
	\label{fig:pers_diag}
\end{figure}

An equivalent viewpoint, developed in \cite[Chapter 3]{tda:chazal2016structure}, is to define a persistence diagram as a measure of the form $a = \sum_{x \in X} n_x \delta_{x}$, where $X \subset \upperdiag$ is locally finite and $n_x \in \N$ for all $x \in X$, so that $a$ is a locally finite measure supported on $\upperdiag$ with integer mass on each point of its support. This measure-based perspective suggests to consider more general Radon measures\footnote{A Radon measure supported on $\Omega$ is a (Borel) measure that gives a finite mass to any compact subset $K \subset \Omega$. See Appendix \ref{subsec:background_measures} for a short reminder about measure theory.} supported on the upper half-plane $\Omega$. Besides this theoretical motivation, considering such measures allows us to address statistical and learning problems that appear in different applications of TDA: 
\begin{enumerate}
	\item[(A1)] \textbf{Continuity of representations.} When given a sample of persistence diagrams $a_1,\dots,a_N$, a common way to perform machine learning is to first map the diagrams into a vector space thanks to a \emph{representation} (or \emph{feature map}) $\Phi : \PD \to \BB$, where $\BB$ is a Banach space. In order to ensure the meaningfulness of the machine learning procedure, the stability of the representations with respect to the $d_p$ distances is usually required. One of our contribution to the matter is to formulate an equivalence between $d_p$-convergence and convergence in terms of measures (Theorem \ref{thm:conv_dp}). This result allows us to characterize a large class of continuous representations (Prop.~\ref{cor:feature_maps_continuous}) that includes some standard tools used in TDA such as the Betti curve \cite{tda:umeda2017time}, the persistence surface \cite{tda:adams2017persistenceImages} and the persistence silhouette \cite{tda:chazal2014stochastic}. \label{it:representations}
	
	\item[(A2)] \textbf{Law of large numbers for diagrams generated by random point clouds.} A popular problem that generates random persistence diagrams is given by filtrations built on top of large random point clouds: if $\X_n$ is a $n$-sample of i.i.d.~points on, say, the cube $[0,1]^d$, recents articles \cite{tda:goel2018asymptotic,tda:divolpolonik} have investigated the asymptotic behavior of the persistence diagram $a_n$ of the \v Cech (or Rips) filtration built on top of the rescaled point cloud $n^{1/d}\X_n$. In particular, it has been shown in \cite{tda:goel2018asymptotic} that the sequence of measures $n^{-1}a_n$ converges vaguely to some limit measure $\mu$ supported on $\upperdiag$ (that is \emph{not} a persistence diagram), and in \cite{tda:divolpolonik} that the moments of $n^{-1}a_n$ also converge to the moments of $\mu$. An interesting problem is to build a metric which generalizes $d_p$ and for which the convergence of $n^{-1}a_n$ to $\mu$ holds. \label{it:thermo}	
	
	\item[(A3)] \textbf{Stability of the expected diagrams.} Of particular interest in the literature are \emph{linear} representations, that is of the form $\Phi(a) \defeq a(f)$, the integral of a function $f : \upperdiag  \to \BB$ against a persistence diagram $a$ (seen as a measure). Given $N$ i.i.d.~diagrams $a_1, \dots, a_N$ following some law $P$, and a linear representation $\Phi$, a natural object to consider is the sample mean $N^{-1}(\Phi(a_1)+\cdots+\Phi(a_N)) = \Phi(N^{-1}(a_1+\cdots +a_N))$. By the law of large numbers, this quantity converges to $\E_P[\bm{a}](f)$, where $\E_P[\bm{a}]$ is the expected persistence diagram of the process, introduced in \cite{tda:divol2018density}. Understanding how the object $\E_P[\bm{a}]$ depends on the underlying process $P$ generating $\bm{a}$ invites one to define a notion of distance between $\E_P[\bm{a}]$ and $\E_{P'}[\bm{a}]$ for $P,P'$ two distributions on the space of persistence diagrams, and relate this distance to a similarity measure between $P$ and $P'$. In the same way that the expected value of an integer-valued random variable may not be an integer, the objects $\E_P[\bm{a}]$ are not persistence diagrams in general, but Radon measures on $\upperdiag$. Therefore, extending the distances $d_p$ to Radon measures in a consistent way will allow us to assess the closeness between those quantities.
\end{enumerate}

\begin{remark} Note that, throughout this article, we consider persistence diagrams with possibly infinitely many points (although locally finite). This is motivated from a statistical perspective. Indeed, the space of finite persistence diagram is lacking completeness, as highlighted in \cite[Definition 2]{tda:mukherjee2011probabilitymeasure}, so that, for instance, the expectation of a probability distribution of diagrams with finite numbers of points may have an infinite mass. An alternative approach to recover a complete space is to study the space of persistence diagrams with total mass less than or equal to some fixed $m \geq 0$. However, this might be unsatisfactory as the number of points of a persistence diagram is known to be an unstable quantity with respect to perturbations of the input data. Note that infinite persistence diagrams may also help to model the topology of standard objects: for instance, the (random) persistence diagram built on the sub-level sets of a Brownian motion has infinitely many points (see \cite[Section 6]{tda:divol2018density}). However, numerical applications generally involve finite sets of finite diagrams, which are studied in Section \ref{subsec:uniformly_bounded_mass}.
\end{remark}

\begin{remark} In general, persistence diagrams may contain points with coordinates of the form $(t,+\infty)$, called the \emph{essential} points of the diagram. The distance between two persistence diagrams is then defined as the sum of two independent terms: the cost $d_p$ that handle points with finite coordinates and the cost of a simple one-dimensional optimal matching between the first coordinates of the essential points (set to be $+\infty$ if the cardinalities of the essential parts differ). We focus on persistence diagrams with only points with finite coordinates for the sake of simplicity, but all the results stated in this work may easily be adapted to the more general case including points with infinite coordinates.
\end{remark}

\subsection{Outline and main contributions}

Examples (A2) and (A3) motivate the introduction of metrics on the space $\MM$ of Radon measures supported on $\Omega$, which generalize the distances $d_p$ on $\PD$: these are presented in Section \ref{sec:background_OT}. For finite $p \geq 1$ (the case $p=\infty$ is studied in Section \ref{subsec:bottleneck}), we define the persistence of $\mu \in \MM$ as
\begin{equation}\label{eq:pers}
\Pers_p(\mu) \defeq \int_{\upperdiag} d(x, \thediag)^p \dd \mu(x),
\end{equation}
where $d(x,\thediag)\defeq \inf_{y\in \thediag} d(x,y)$ is the distance from a point $x\in \upperdiag$ to (its orthogonal projection onto) the diagonal $\thediag$, and we define 
\begin{equation}\label{eq:MMp}
\MM^p \defeq \{ \mu \in \MM,\ \Pers_p(\mu) < \infty \}.
\end{equation}

We equip $\MM^p$ with metrics $\Dp$ (see Definition \ref{def:OT_for_Radon_measures}), originally introduced in a work of Figalli and Gigli \cite{ot:figalli2010newTransportationDistance}. We show in Proposition \ref{prop:Dp_equals_dp} that $\Dp$ and $d_p$ coincide on $\PD^p\defeq \PD \cap \MM^p$, making $\Dp$ a good candidate to address the questions raised in (A2) and (A3). To emphasize that we equip the space of Radon measures with a specific metric designed for our purpose, we will refer to elements of the metric space $(\MM^p, \Dp)$ as \emph{persistence measures} in the following. As $\PD^p$ is closed in $\MM^p$ (Corollary \ref{cor:Dp_close_dp}), most properties of $\MM^p$ hold for $\PD^p$ too (e.g.~being Polish, Proposition \ref{prop:MM_p_Polish}).

A sequence of Radon measures $(\mu_n)_n$ is said to \emph{converge vaguely} to a measure $\mu$, denoted by $\mu_n \cvvague \mu$, if for any continuous compactly supported function $f : \upperdiag \to \R$, $\mu_n(f) \to \mu(f)$ (where  the notation $\mu(f)\defeq \int f\dd \mu$ stands for the integration of the function $f$ against $\mu$). We prove the following equivalence between convergence for the metric $\Dp$ and the vague convergence:

\repeatresult{Theorem \ref{thm:conv_dp}.}{
Let $1 \leq p < \infty$. Let $\mu, \mu_1, \mu_2,\dots$ be measures in $\MM^p$. Then, 
\begin{equation}
\Dp(\mu_n,\mu) \to 0 \Leftrightarrow \begin{cases} \mu_n \cvvague \mu,  \\ \Pers_p(\mu_n) \to \Pers_p(\mu). \end{cases} 
\end{equation}}

This equivalence gives a positive answer to the issues raised by (A2), as detailed in Section \ref{sec:applications}. Note also that this characterization in particular holds for persistence diagrams in $\PD^p$, and can thus be helpful to show the convergence or the tightness of a sequence of diagrams. This theorem is analogous to the characterization of convergence of probability measures with respect to Wasserstein distances (see \cite[Theorem 6.9]{ot:villani2008optimal}). A proof for Radon measures supported on a common bounded set can be found in \cite[Proposition 2.7]{ot:figalli2010newTransportationDistance}. Our contribution consists in extending this result to non-bounded sets, in particular to the upper half plane $\upperdiag$. 

Section \ref{subsec:uniformly_bounded_mass} is dedicated to sets of measures with finite masses, appearing naturally in numerical applications. We show in particular that computing the $\Dp$ metric between two measures of finite mass can be turned into the known problem of computing a Wasserstein distance (see Section \ref{sec:background_OT}) between two measures with the \emph{same} mass (Prop.~\ref{prop:Dp_equals_Wp_finite}), a result having practical implications for the computation of $\Dp$ distances between persistence measures (and diagrams).

Section \ref{subsec:bottleneck} studies the case $p=\infty$, which is somewhat ill-behaved from a statistical analysis point of view (for instance, the space of persistence diagrams endowed with the bottleneck metric is not separable, as observed in \cite[Theorem 5]{tda:bubenik2018topological}), but is also of crucial interest in TDA as it is motivated by algebraic considerations \cite{tda:oudot2015persistence} and satisfies stronger stability results \cite{tda:cohen2007stability} than its $p < \infty$ counterparts \cite{tda:cohen2010lipschitzStablePersistence}. In particular, we give in Propositions \ref{prop:Dinf_implies_vaguecv} and \ref{prop:nice_conv_dinf} a characterization of bottleneck convergence (in the vein of Theorem \ref{thm:conv_dp}) for persistence diagrams satisfying some finiteness assumptions (namely, for each $r > 0$, the number of points with persistence greater than $r$ must be finite).

Section \ref{sec:barycenter} studies Fr\'echet means (i.e.~barycenters, see Definition \ref{def:variance_and_bary}) for probability distributions of persistence measures. In the specific case of persistence diagrams, the study of Fr\'echet means was initiated in \cite{tda:mukherjee2011probabilitymeasure,tda:turner2014frechet}, where authors prove their existence for certain types of distributions \cite[Theorem 28]{tda:mukherjee2011probabilitymeasure}. Using the framework of persistence measures, we show that this existence result is actually true for any distribution of persistence diagrams (and measures) with finite moment. Namely, we prove the following results:

\repeatresult{Theorem \ref{thm:existenceBary}.}{Assume that $1<p<\infty$ and that $d(\cdot,\cdot)$ denotes the $q$-norm for $1<q<\infty$. For any probability distribution $\P$ supported on $\MM^p$ with finite $p$-th moment, the set of $p$-Fr\'echet means of $\P$ is a non-empty compact convex subset of $\MM^p$.}

%

\repeatresult{Theorem \ref{thm:barycenterInD}.}{Assume that $1<p<\infty$ and that $d(\cdot,\cdot)$ denotes the $q$-norm for $1<q<\infty$. For any probability distribution $\P$ supported on $\PD^p$ with finite $p$-th moment, the set of $p$-Fr\'echet means of $\P$, which is a subset of $\MM^p$, contains an element in $\DD^p$. Furthermore, if $P$ is supported on a finite set of finite persistence diagrams, then the set of the $p$-Fr\'echet means of $P$ is a convex set whose extreme points are in $\PD^p$.}

Section \ref{sec:applications} applies the formalism we developed to address the questions raised in (A1)---(A3). In Section \ref{subsec:continuity_of_representations}, we prove a strong characterization of continuous linear representations of persistence measures (and diagrams), which answers to the issue raised by (A1) for the class of linear representations (see Figure \ref{fig:representations}).

\repeatresult{Proposition \ref{cor:feature_maps_continuous}}
{Let $p \in [1, +\infty)$, $d\geq 1$, and $f: \upperdiag \to \BB$ for some Banach space $\BB$ (e.g.~$\R^d$). The representation $\Phi : \MM^p \to \BB$ defined by $\Phi(\mu) = \int_\upperdiag f(x) \dd \mu(x)$ is continuous with respect to $\Dp$ if and only if $f$ is of the form $f(x)= g(x) d(x, \thediag)^p$, where $g : \upperdiag \to \BB$ is a continuous bounded map.}

This new result can be compared to the recent work \cite[Theorem 13]{tda:hofer2019learning}, which gives a similar result in the case $\BB = \R$ on the space of finite persistence diagrams $\PD_f$), or the works \cite[Proposition 8]{tda:hiraoka2017kernelWeight} and \cite[Theorem 3]{tda:divolpolonik}, which show that linear representations can have more regularity (e.g.~Lipschitz or H\"older) under additional assumptions. 

\begin{figure}
	\center
	\includegraphics[width=\textwidth]{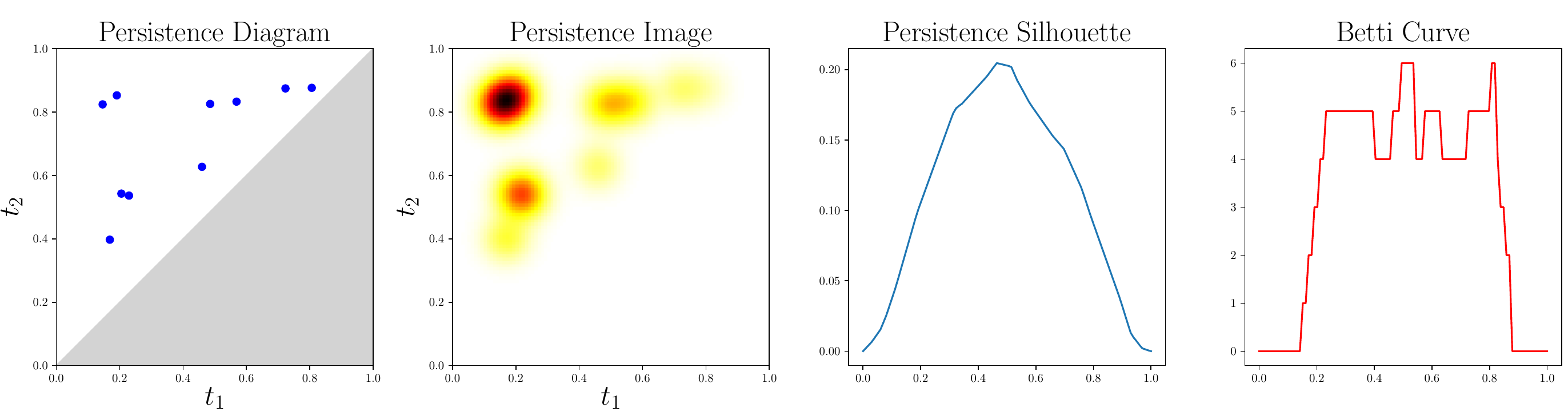}
\caption{Some common linear representations of persistence diagrams. From left to right: A persistence diagram. Its persistence surface \cite{tda:adams2017persistenceImages}, which is a persistence measure. The corresponding persistence silhouette \cite{tda:chazal2014stochastic}. The corresponding Betti Curve \cite{tda:umeda2017time}. See Section \ref{subsec:continuity_of_representations} for details.}
\label{fig:representations}
\end{figure}

Section \ref{subsec:thermo} states a very concise law of large for persistence diagrams. Namely, building on the previous works \cite{tda:hiraoka2016limitTheoremForPD,tda:divolpolonik} along with Theorem \ref{thm:conv_dp}, we prove the following:

\repeatresult{Proposition \ref{prop:application_randomPD}}
{Let $\X_n = \{X_1, \dots, X_n\}$ be a sample of $n$ points on the $d$-dimensional cube $[0,1]^d$, sampled from a density bounded from below and from above by positive constants, and let $\mu_n = \frac{1}{n} \dgm(n^{1/d} \X_n)$, where $\dgm(n^{1/d}\X_n)$ is either the Rips or \v Cech complex built on the point cloud $n^{1/d}\X_n$. Then, there exists a measure $\mu \in \MM^p$ such that $\Dp(\mu_n,\mu) \to 0$.}

Finally, Section \ref{subsec:stability_expectedPD} considers the problem (A3), that is the stability of the expected persistence diagrams. In particular, we prove a stability result between an input point cloud $\X_n$ in a random setting and its expected (\v Cech) diagrams $\E(\dgm(\X_n))$:

\repeatresult{Proposition \ref{prop:stability_expectedPD}}{Let $\xi,\xi'$ be two probability measures supported on $\R^d$. Let $\X_n$ (resp. $\X'_n$) be a $n$-sample of law $\xi$ (resp. $\xi'$). Then, for any $k > d$, and any $p \geq k+1$,
\begin{equation}
\Dp^p(\E[\dgm(\X_n)],\E[\dgm(\X'_n)]) \leq C_{k,d} \cdot n \cdot  W_{p-k}^{p-k}(\xi,\xi')
\end{equation}
where $C_{k,d} \defeq C \diam(\X)^{k-d}\frac{k}{k-d}$ for some constant $C$ depending only on $\X$. \\
In particular, letting $p \rightarrow \infty$, we obtain a bottleneck stability result:
\begin{equation}
	\Dinf(\E[\dgm(\X_n)],\E[\dgm(\X'_n)]) \leq W_{\infty}(\xi,\xi').
\end{equation}}

\section{Elements of optimal partial transport}\label{sec:background_OT}
In this section, $(\XX, d)$ denotes a Polish metric space. 

\subsection{Optimal transport between probability measures and Wasserstein distances}\label{subsec:wasserstein}
In its standard formulation, optimal transport is a widely developed theory providing tools to study and compare probability measures supported on $\XX$ \cite{ot:villani2003topicsInOT,ot:villani2008optimal,otam}, that is---up to a renormalization factor---non-negative measures of the same mass. 
Given two probability measures $\mu, \nu$ supported on $(\XX, d)$, the $p$-\emph{Wasserstein} distance ($p \geq 1$) induced by the metric $d$ between $\mu$ and $\nu$ is defined as
\begin{equation}
	W_{p,d}(\mu, \nu) \defeq \left( \inf_{\pi \in \Pi(\mu,\nu)} \iint_{\XX \times \XX} d(x,y)^p \dd \pi(x,y) \right)^{\frac{1}{p}},
	\label{eq:Wp_std_def}
\end{equation}
where $\Pi(\mu,\nu)$ denotes the set of \emph{transport plans} between $\mu$ and $\nu$, that is the set of measures on $\XX \times \XX$ which have respective marginals $\mu,$ and $\nu$. When there is no ambiguity on the distance $d$ used, we simply write $W_p$ instead of $W_{p,d}$. In order to have $W_p$ finite, $\mu$ and $\nu$ are required to have a finite $p$-th moment, that is there exists $x_0 \in \XX$ such that $\int_\XX d(x, x_0)^p \dd \mu(x)$ (resp.~$\dd \nu$) is finite. The set of such probability measures, endowed with the metric $W_p$, is referred to as $\WW^p(\XX)$.

Wasserstein distances and $d_p$ metrics defined in Eq.~\eqref{eq:dgm_partial_matching_metric} share the key idea of defining a distance by minimizing a cost over some matchings. However, the set of transport plans $\Pi(\mu,\nu)$ between two measures is non-empty if and only if the two measures have the same mass, while persistence diagrams with different masses can be compared, making a crucial difference betwen the $W_p$ and $d_p$ metrics.

\subsection{Extension to Radon measures supported on a bounded space} 
Extending optimal transport to measures of different masses, generally referred to as optimal \emph{partial} transport, has been addressed by different authors \cite{ot:figalli2010optimal,ot:chizat2015unbalanced,ot:kondratyev2016new}. As it handles the case of measures with infinite masses, the work of Figalli and Gigli \cite{ot:figalli2010newTransportationDistance}, is of particular interest for us. The athors propose to extend Wasserstein distances to Radon measures supported on a \emph{bounded} open proper subset $\XX$ of $\R^d$, whose boundary is denoted by $\partial \XX$ (and $\overline{\XX} \defeq \XX \sqcup \partial \XX$). 
\begin{definition}{\cite[Problem 1.1]{ot:figalli2010newTransportationDistance}} \label{def:OT_for_Radon_measures}
	Let $p \in [1, +\infty)$. Let $\mu, \nu$ be two Radon measures supported on $\XX$ satisfying
	\[ \int_\XX d(x,\partial \XX)^p \dd \mu(x) < +\infty, \quad \int_\XX d(x,\partial \XX)^p \dd \nu(x) < +\infty.\]
	The set of admissible transport plans (or couplings) $\Adm(\mu,\nu)$ is defined as the set of Radon measures $\pi$ on $\overline{\XX} \times \overline{\XX}$ satisfying for all Borel sets $A,B \subset \XX$,
\[ \pi(A\times \overline{\XX}) = \mu(A) \quad \text{ and } \quad \pi(\overline{\XX}  \times B) = \nu(B). \]
	 The cost of $\pi \in \Adm(\mu,\nu)$ is defined as
\begin{equation}
C_p(\pi) \defeq \iint_{\overline{\XX} \times \overline{\XX}} d(x,y)^p \dd \pi(x,y).
\label{eq:costDefinition}
\end{equation}
The Optimal Transport (with boundary) distance $\Dp(\mu,\nu)$ is then defined as
\begin{equation}
\Dp(\mu,\nu) \defeq \left( \inf_{\pi \in \Adm(\mu,\nu)} C_p(\pi) \right)^{1/p}.
\label{eq:optimalCostPbm}
\end{equation}
Plans $\pi \in \Adm(\mu,\nu)$ realizing the infimum in \eqref{eq:optimalCostPbm} are called \emph{optimal}. The set of optimal transport plans between $\mu$ and $\nu$ for the cost $(x,y) \mapsto d(x,y)^p$ is denoted by $\Opt_p(\mu,\nu)$.
\end{definition}

We introduce the following definition, which shows how to build an element of $\Adm(\mu,\nu)$ given a map $f : \overline{\XX} \to \overline{\XX}$ satisfying some balance condition (see Figure \ref{fig:transport_map}).
\begin{definition}\label{def:induced_transport_map}
Let $\mu, \nu \in \MM$. Consider $f : \overline{\XX} \to \overline{\XX}$ a measurable function satisfying for all Borel set $B \subset \XX$
\begin{equation}
 \mu(f^{-1}(B) \cap \XX) + \nu(B \cap f(\partial \XX)) = \nu(B). 
\end{equation}
Define for all Borel sets $A, B \subset \overline{\XX}$, 
\begin{equation}
 \pi(A \times B) = \mu(f^{-1}(B) \cap \XX \cap A) + \nu(\XX \cap B \cap f(A \cap \partial \XX)). 
\end{equation}
$\pi$ is called the transport plan \emph{induced} by the \emph{transport map} $f$.
\end{definition}
	One can easily check that we have indeed $\pi(A \times \overline{\XX}) = \mu(A)$ and $\pi(\overline{\XX} \times B) = \nu(B)$ for any Borel sets $A,B \subset \XX$, so that $\pi \in \Adm(\mu,\nu)$ (see Figure \ref{fig:transport_map}).
	\begin{figure}
		\center
		\includegraphics[width = 0.4 \linewidth]{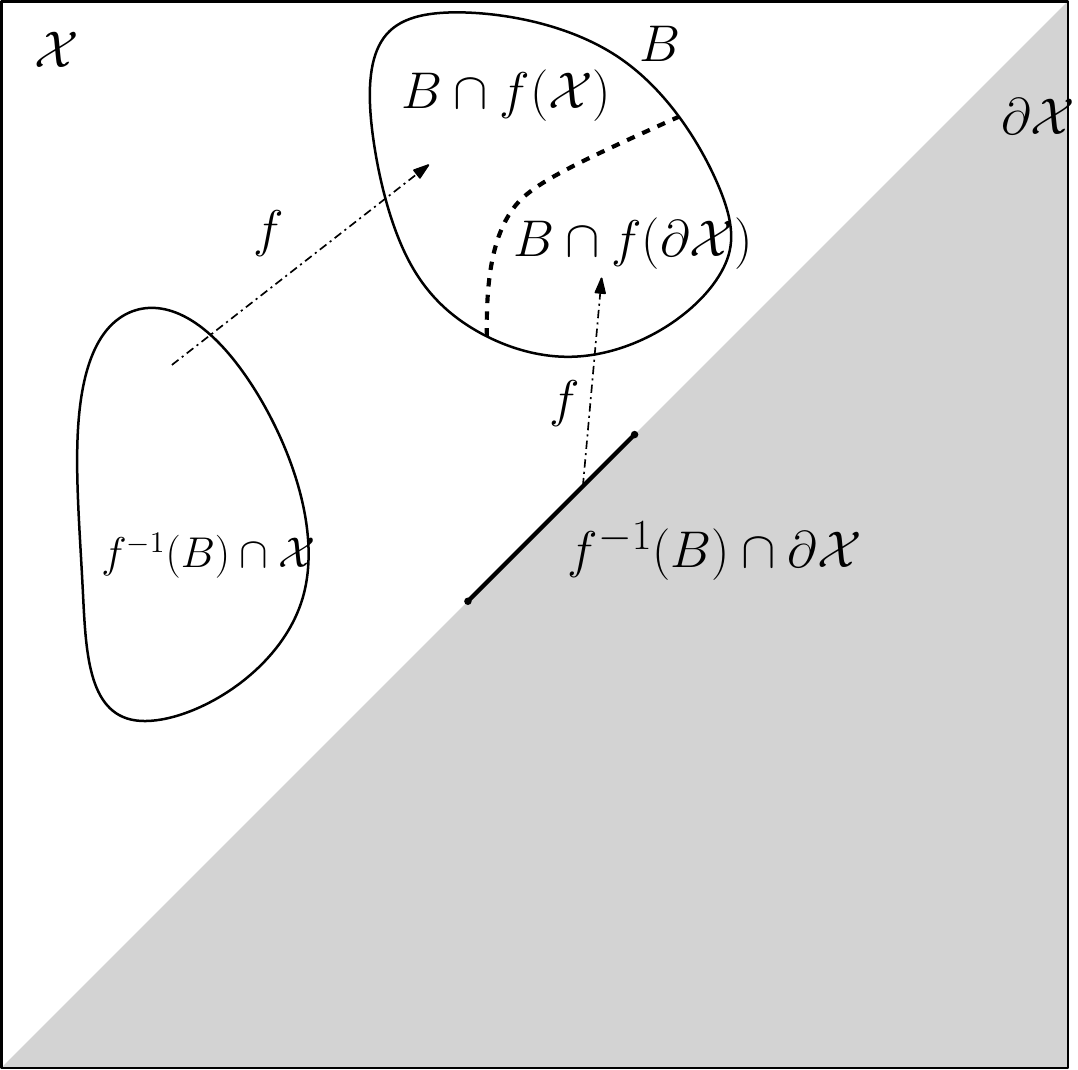}
		\includegraphics[width = 0.4 \linewidth]{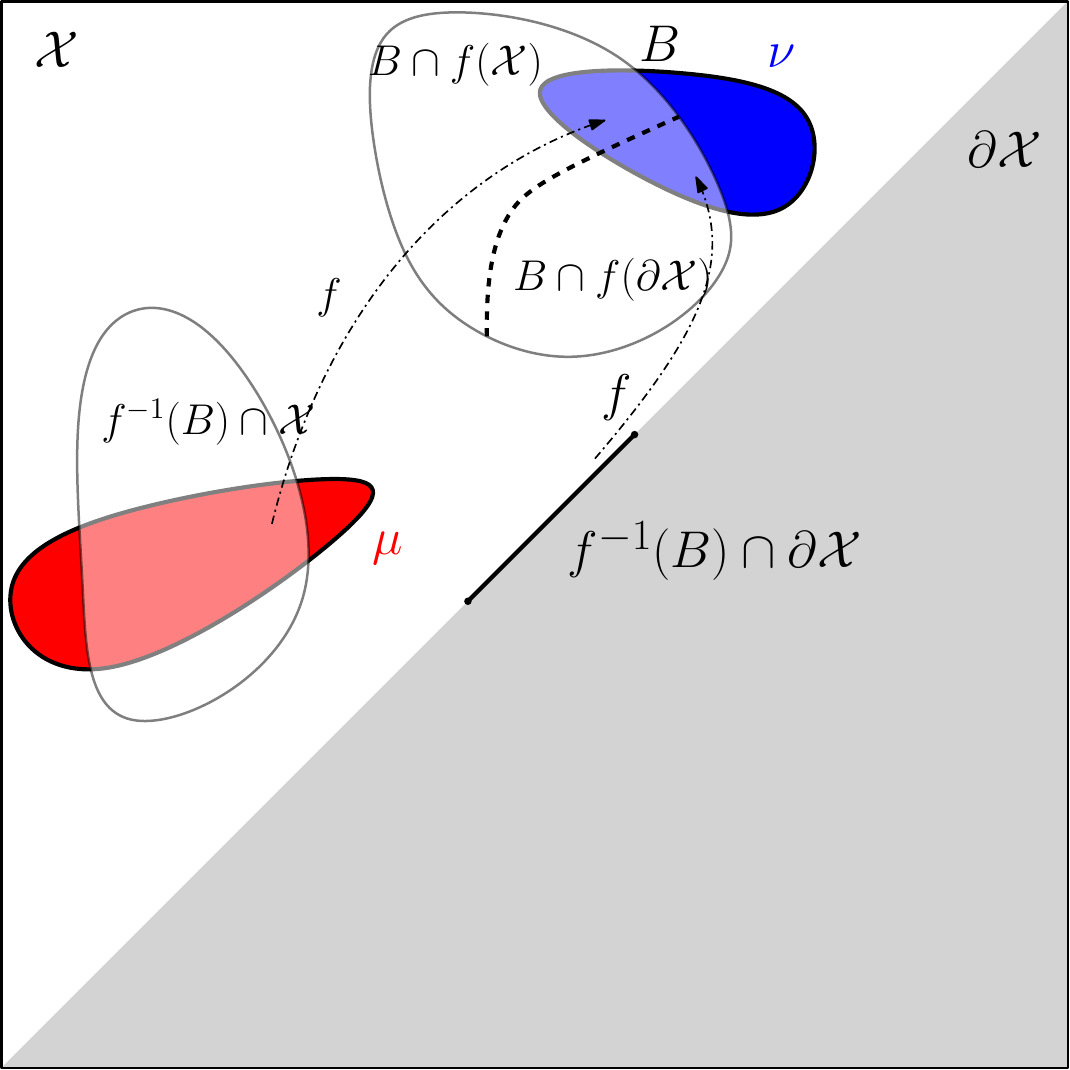}
		\caption{A transport map $f$ must satisfy that the mass $\nu(B)$ (light blue) is the sum of the mass $\mu(f^{-1}(B) \cap \XX)$ given by $\mu$ that is transported by $f$ onto $B$ (light red) and the mass $\nu(B \cap f(\partial \XX))$ coming from $\partial \XX$ and transported by $f$ onto $B$.}\label{fig:transport_map}
	\end{figure}
\begin{remark}
Since we have no constraints on $\pi(\partial \XX \times \partial \XX)$, one may always assume that a plan $\pi$ satisfies $\pi(\partial \XX \times \partial \XX) = 0$, so that measures $\pi \in \Adm(\mu,\nu)$ are supported on 
\begin{equation}
	E_{\XX}\defeq (\overline{\XX} \times \overline{\XX}) \backslash (\partial \XX \times \partial \XX).
	\label{eq:E_upperdiag}
\end{equation}
\end{remark}



\section{Structure of the persistence measures and diagrams spaces}
\label{sec:structure_dgm_space}

This section is dedicated to general properties of $\MM^p$. Results in Section \ref{subsec:generalities} are inspired from the ones of Figalli and Gigli in \cite{ot:figalli2010newTransportationDistance}, which are stated for a bounded subset $\XX$ of $\R^d$. Our goal is to state properties over the space $\upperdiag$, which is of course \emph{not bounded}. Adapting the results of \cite{ot:figalli2010newTransportationDistance} to our purpose is sometimes straightforward, in which case the proofs are delayed to Appendix \ref{sec:proofs}, and sometimes more involving, in which case the proofs are exposed in the main part of this article. Following Sections \ref{subsec:uniformly_bounded_mass} and \ref{subsec:bottleneck} are respectively dedicated to finite measures (involved in applications) and the case $p=\infty$ (of major interest in topological data analysis).  

\begin{remark}\label{rem:polish}
The results exposed in this section would remain true in a more general setting, namely for any locally compact Polish metric space $\XX$ that is partitioned into $\XX = A \sqcup B$, where $A$ is open and $B$ is closed (here, $A = \upperdiag$ and $B = \thediag$).
\end{remark}

\subsection{General properties of $\MM^p$}
\label{subsec:generalities}
It is assumed for now that $1 \leq p < \infty$. The case $p = \infty$ is studied in Section \ref{subsec:bottleneck}. Consider the space $\MM^p$ defined in \eqref{eq:MMp}. First, we observe that the quantities introduced in Definition \ref{def:OT_for_Radon_measures}, in particular the metric $\Dp$, are still well-defined when $\XX = \upperdiag$ is not bounded.

Consider the functional 
\begin{equation}
  I:  \pi \in \MM(E_\Omega)\mapsto (\pi,\pi_1,\pi_2)\in \MM(E_\Omega)\times \MM(\Omega)\times \MM(\Omega),
\end{equation}
where $\pi_1$ (resp. $\pi_2$) is the first (resp. second) marginal of $\pi$. Consider the initial topology on $\MM(E_\Omega)$ associated with the functional $I$, where the final space $\MM(E_\Omega)\times \MM(\Omega)\times \MM(\Omega)$ is endowed with the product of the vague topologies. This initial topology, that we call the VM topology (for vague and marginal), is stronger than the vague topology. A sequence  converge to some $\pi\in \MM(E_\Omega)$ for the VM topology if and only if the sequence vaguely converges, and if the two sequences of marginals vaguely converge to the two marginals of $\pi$. We describe properties of the VM topology in Appendix \ref{subsec:background_measures}.

\begin{proposition}\label{prop:basic_prop}
	Let $\mu, \nu \in \MM$. The set of transport plans $\Adm(\mu,\nu)$ is sequentially compact for the VM topology on $E_\upperdiag \defeq \groundspace \times \groundspace \backslash \thediag \times \thediag$. Moreover, if $\mu,\nu \in \MM^p$, for this topology,
	\begin{itemize}
	\item $\pi \in \Adm(\mu,\nu) \mapsto C_p(\pi)$ is lower semi-continuous. \label{it:C_p_lsc}
	\item $\Opt_p(\mu,\nu)$ is a non-empty sequentially compact set. \label{it:opt_compact}
	\item 	\label{it:Dp_lsc} $\Dp$ is lower semi-continuous, in the sense that for sequences $(\mu_n)_n, (\nu_n)_n$ in $\MM^p$ satisfying $\mu_n \cvvague \mu$  and $\nu_n \cvvague \nu$, we have 
	\[ \Dp(\mu,\nu) \leq \liminf_{n\to \infty} \Dp(\mu_n,\nu_n).\]
	\end{itemize}
Moreover, $\Dp$ is a metric on $\MM^p$.
	\end{proposition}
These properties are mentioned in \cite[pages 4-5]{ot:figalli2010newTransportationDistance} in the bounded case, and corresponding proofs adapt straightforwardly to our framework. For the sake of completeness, we provide a detailed proof in Appendix \ref{sec:proofs}.

\begin{remark} If a (Borel) measure $\mu$ satisfies $\Pers_p(\mu) < \infty$, then for any Borel set $A \subset \upperdiag$ satisfying $d(A,\thediag) \defeq \inf_{x \in A} d(x, \thediag) > 0$, we have: 
\begin{equation}
	\mu(A) d(A,\thediag)^p \leq \int_A d(x, \thediag)^p \dd \mu(x) \leq \int_\upperdiag d(x, \thediag)^p \dd \mu(x) = \Pers_p(\mu) < \infty, 
	\label{eq:pers_compact} 
\end{equation}
so that $\mu(A) < \infty$. In particular, $\mu$ is automatically a Radon measure.
\end{remark}

The following lemma gives a simple way to approximate a persistence measure (resp.~diagram) with ones of finite masses. 
\begin{lemma}\label{lemma:mu_r_to_mu}
	Let $\mu \in \MM^p$. Fix $r > 0$, and let $A_r \defeq \{x \in \upperdiag, d(x,\thediag) \leq r\}$. Let $\mu^{(r)}$ be the restriction of $\mu$ to $\upperdiag \backslash A_r$. Then $\Dp(\mu^{(r)},\mu) \to 0$ when $r \to 0$. Similarly, if $a \in \PD^p$, we have $d_p(a^{(r)},a) \to 0$.
\end{lemma}

\begin{proof}
Let $\pi \in \Adm(\mu, \mu^{(r)})$ be the transport plan induced by the identity map on $\groundspace \backslash A_r$, and the projection onto $\thediag$ on $A_r$. As $\pi$ is sub-optimal, one has:
	\[ \Dp^p (\mu,\mu^{(r)} ) \leq C_p(\pi) = \int_{A_r} d(x,\thediag)^p \dd \mu(x) = \Pers_p(\mu) - \Pers_p(\mu^{(r)}). \]
Thus, by the monotone convergence theorem applied to $\mu$ with the functions $f_r : x \mapsto d(x, \thediag)^p \cdot \ones_{\upperdiag \backslash A_r}(x)$, $\Dp(\mu,\mu^{(r)}) \to 0$ as $r \to 0$. Similar arguments show that $d_p(a^{(r)}, a) \to 0$ as $r \to 0$. 
\end{proof}

The following proposition is central in our work: it shows that the metrics $\Dp$ are extensions of the metrics $d_p$.
\begin{proposition}\label{prop:Dp_equals_dp}
For $a,b \in \PD^p$,  $\Dp(a,b) = d_p(a,b)$.
\end{proposition}

\begin{proof}
	Let $a,b\in \PD^p$ be two persistence diagrams. The case where $a,b$ have a finite number of points is already treated in \cite[Proposition 1]{tda:lacombe2018large}. 

In the general case, let $r > 0$. Due to \eqref{eq:pers_compact}, the diagrams $a^{(r)}$ and $b^{(r)}$ defined in Lemma \ref{lemma:mu_r_to_mu} have a finite mass (thus finite number of points). Therefore, $d_p(a^{(r)}, b^{(r)}) = \Dp(a^{(r)}, b^{(r)})$. By Lemma \ref{lemma:mu_r_to_mu}, the former converges to $d_p(a,b)$ while the latter converges to $\Dp(a,b)$, giving the conclusion. 
\end{proof}

As a consequence of this proposition, we will use $\Dp$ to denote the distance between two elements of $\PD^p$ from now on.

\begin{proposition}
The space $(\MM^p,\Dp)$ is a Polish space.
	\label{prop:MM_p_Polish}
\end{proposition}
As for Proposition \ref{prop:basic_prop}, this proposition appears in \cite[Proposition 2.7]{ot:figalli2010newTransportationDistance} in the bounded case, and its proof is straightforwardly adapted to our framework. For the sake of completeness, we provide a detailed proof in Appendix \ref{sec:proofs}.
 
We now state one of our main result: a characterization of convergence in $(\MM^p, \Dp)$.

\begin{theorem}\label{thm:conv_dp}
Let $\mu, \mu_1, \mu_2,\dots$ be measures in $\MM^p$. Then, 
\begin{equation}
\Dp(\mu_n,\mu) \to 0 \Leftrightarrow \begin{cases} \mu_n \cvvague \mu,  \\ \Pers_p(\mu_n) \to \Pers_p(\mu). \end{cases}
\end{equation} 
\end{theorem}
This result is analog to the characterization of convergence of probability measures in the Wasserstein space (see \cite[Theorem 6.9]{ot:villani2008optimal}) and can be found in \cite[Proposition 2.7]{ot:figalli2010newTransportationDistance} in the case where the ground space is bounded. While the proof of the direct implication can be easily adapted from \cite{ot:figalli2010newTransportationDistance} (it can be found in Appendix \ref{sec:proofs}), a new proof is needed for the converse implication.

\begin{proof}[Proof of the converse implication]
For a given compact set $K \subset \upperdiag$, we denote its complementary set in $\upperdiag$ by $K^c$, its interior set by $\mathring{K}$, and its boundary by $\partial K$. Let $\mu,\mu_1,\mu_2\dots$ be elements of $\MM^p$ and assume that $\mu_n\cvvague \mu$ and $\Pers_p(\mu_n)\to \Pers_p(\mu)$. Since \[\Dp(\mu_n,\mu) \leq \Dp(\mu_n,0) + \Dp(\mu,0) = \Pers_p(\mu_n)^{1/p} + \Pers_p(\mu)^{1/p},\] the sequence $(\Dp(\mu_n,\mu))_n$ is bounded. Thus, if we show that $(\Dp(\mu_n,\mu))_n$ admits $0$ as an unique accumulation point, then the convergence holds. Up to extracting a subsequence, we may assume that $(\Dp(\mu_n,\mu))_n$ converges to some limit. Let $(\pi_n)_n \in \Opt(\mu_n, \mu)^\N$ be corresponding optimal transport plans. 
Let $K$ be a compact subset of $\Omega$. Recall (Prop.~\ref{prop:rel_cpt_vague} in Appendix \ref{subsec:background_measures}) that relative compactness for the vague convergence of a sequence $(\mu_n)_n$ is equivalent to $\sup_n \{\mu_n(K)\} < \infty$ for every compact $K \subset \upperdiag$. Hence, according to Prop.~\ref{prop:compact_VM}, the sequence $(\pi_n)_n$ is  relatively compact for the VM topology.

Let thus $\pi$ be the limit of any converging subsequence of $(\pi_n)_n$, whose indexes are still denoted by $n$. As $\mu_n \cvvague \mu$, $\pi$ is necessarily in $\Opt_p(\mu,\mu)$ (see \cite[Prop.~2.3]{ot:figalli2010newTransportationDistance}), i.e.\ $\pi$ is supported on $\{(x,x),x\in \Omega\}$. The vague convergence of $(\mu_n)_n$ and the convergence of $(\Pers_p(\mu_n))_n$ to $\Pers_p(\mu)$ imply that for a given compact set $K \subset \upperdiag$, we have
\begin{align*}
	&\limsup_{n \to \infty} \int_{K^c} d(x,\thediag)^p \dd \mu_n(x) \\
	&= \limsup_{n \to \infty} \left( \Pers_p(\mu_n) - \int_K d(x, \thediag)^p \dd \mu_n(x) \right) \\
	&=  \Pers_p(\mu) - \liminf_n \int_{\mathring{K}} d(x,\thediag)^p \dd \mu_n(x) - \liminf_n \int_{\partial K} d(x, \thediag)^p \dd \mu_n(x) \\
	&\leq \Pers_p(\mu) - \int_{\mathring{ K}} d(x, \thediag)^p \dd \mu(x) \text{ by the Portmanteau theorem} \\
	&= \int_{\overline{K^c}} d(x, \thediag)^p \dd \mu(x),
\end{align*}
where the Portmanteau theorem is recalled in Appendix \ref{subsec:background_measures}.
As $\Pers_p(\mu)$ is finite, for $\epsilon >0$, there exists  some compact set $K\subset \upperdiag$ with
\begin{equation} 
\limsup_n \int_{K^c} d(x,\thediag)^p \dd \mu_n(x) <\epsilon \quad \text{ and} \quad \int_{K^c} d(x,\thediag)^p \dd \mu(x) <\epsilon.
\label{eq:tightness}
 \end{equation}
Let $s:\upperdiag\to \thediag$ be the projection on $\thediag$ for the metric $d$. Such a projection is not unique for $q=1$ or for the more general Polish spaces $\XX$ of Remark \ref{rem:polish}, but we can always select a measurable projection $s$ \cite{cascales2003measurable}. We consider the following transport plan $\tilde{\pi}_n$ (consider informally that what went from $K$ to $K^c$ and from $K^c$ to $K$ is now transported onto the diagonal, while everything else is unchanged):
\begin{equation}
\begin{cases}
\tilde{\pi}_n = \pi_n &\hspace{-.1cm}\text{on } K^2 \sqcup (K^c)^2 , \\
\tilde{\pi}_n = 0 &\hspace{-.1cm}\text{on } K\times K^c \sqcup K^c \times K, \\
\tilde{\pi}_n(A\times B) = \pi_n(A \times B) + \pi_n(A\times (s^{-1}(B) \cap K^c)) &\hspace{-.1cm}\text{for } A\subset K,\ B\subset \thediag, \\
\tilde{\pi}_n(A\times B) = \pi_n (A \times B) + \pi_n(A\times (s^{-1}(B)\cap  K ))  &\hspace{-.1cm}\text{for } A\subset K^c,\ B\subset \thediag, \\
\tilde{\pi}_n(A \times B) = \pi_n(A \times B) + \pi_n((s^{-1}(A)\cap  K^c)\times B)  &\hspace{-.1cm}\text{for } A\subset  \thediag,\  B\subset K, \\
\tilde{\pi}_n(A \times B) = \pi_n( A \times B) + \pi_n((s^{-1}(A)\cap  K)\times B) &\hspace{-.1cm}\text{for } A\subset  \thediag,\ B\subset K^c.
\end{cases}
\label{eq:def_pi}
\end{equation}
Note that $\tilde{\pi}_n \in \Adm(\mu_n,\mu)$: for instance, for $A\subset K$ a Borel set, 
\begin{align*}
\tilde{\pi}_n(A \times \groundspace) &= \tilde \pi_n(A\times K) +  \tilde \pi_n(A\times K^c) + \tilde  \pi_n(A\times \thediag)\\
&=   \pi_n(A\times K) +  0 + \pi_n(A\times \thediag) + \pi_n(A\times (s^{-1}(\thediag)\cap K^c))\\
& = \pi_n(A\times \groundspace)=\mu_n(A),
\end{align*}
 and it is shown likewise that the other constraints are satisfied. As $\tilde{\pi}_n$ is suboptimal, $\Dp^p(\mu_n,\mu) \leq \int_{\groundspace^2} d(x,y)^p \dd\tilde{\pi}_n(x,y)$. The latter integral is equal to a sum of different terms, and we will show that each of them converges to $0$. Assume without loss of generality that the compact set $K$ belongs to an increasing sequence of compact sets whose union is $\upperdiag$, with $\pi(\partial(K \times K)) = 0$ for all compacts of the sequence.
\begin{itemize}
	\item We have $\iint_{K^2} d(x,y)^p \dd\tilde{\pi}_n(x,y) = \iint_{K^2} d(x,y)^p \dd \pi_n(x,y)$. The $\limsup$ of the integral is less than or equal to $\iint_{K^2} d(x,y)^p \dd\pi(x,y)$ by the Portmanteau theorem (applied to the sequence $(d(x,y)^p \dd \pi_n(x,y))_n$), and, recalling that $\pi$ is supported on the diagonal of $E_\upperdiag$, this integral is equal to $0$.
	
	\item For optimality reasons, any optimal transport plan must be supported on $\{ d(x,y)^p \leq d(x,\thediag)^p + d(y,\thediag)^p\}$ (this fact is detailed in \cite[Prop.~2.3]{ot:figalli2010newTransportationDistance}). It follows that
	\begin{align*}
	\iint_{(K^c)^2} d(x,y)^p \dd \tilde{\pi}_n(x,y) &= \iint_{(K^c)^2} d(x,y)^p \dd \pi_n(x,y)\\
	&\leq \int_{K^c} d(x,\thediag)^p \dd\mu_n(x) + \int_{K^c} d(y,\thediag)^p \dd \mu(y).
	\end{align*}
	Taking the $\limsup$ in $n$, and then letting $K$ goes to $\upperdiag$, this quantity converges to $0$ by \eqref{eq:tightness}.
	
	\item We have 
\begin{align*}
&\iint_{K\times \thediag} d(x,\thediag)^p \dd \tilde{\pi}_n(x,y) \\
&= \iint_{K\times \thediag} d(x,\thediag)^p \dd \pi_n(x,y) + \iint_{K\times K^c} d(x,\thediag)^p\dd\pi_n(x,y)\\
&= \iint_{K\times \groundspace} d(x,\thediag)^p \dd \pi_n(x,y) - \iint_{K^2} d(x,\thediag)^p\dd\pi_n(x,y)\\
&= \int_K d(x,\thediag)^p \dd \mu_n(x) - \iint_{K^2} d(x,\thediag)^p \dd \pi_n(x,y)
\end{align*}
	By the Portmanteau theorem applied to the sequence $(d(x,\thediag)^p \dd \mu_n(x))_n$, the $\limsup$ of the first term is less than or equal to $\int_K d(x,\thediag)^p \dd \mu(x)$. Recall that we assume that $\pi(\partial(K\times K))=0$. By applying the second characterization of Portmanteau theorem (see Prop. \ref{prop:portemanteau}) on the second term to the sequence $(d(x,y)^p \dd \pi_n(x,y))_n$, and using that $\pi$ is supported on the diagonal of $E_\upperdiag$, we obtain that the limsup of the second term is less than or equal to $- \iint_{K^2} d(x,\thediag)^p \dd \pi(x,y)= -\int_K d(x,\thediag)^p \dd\mu(x)$. Therefore, the $\limsup$ of the integral is equal to 0.
	
	\item The three remaining terms (corresponding to the three last lines of the definition \eqref{eq:def_pi}) are treated likewise this last case.
\end{itemize}
Finally, we have proven that $(\Dp(\mu_n,\mu))_n$ is bounded and that for any converging subsequence $(\mu_{n_k})_k$, $\Dp(\mu_{n_k},\mu)$ converges to $0$. It follows that $\Dp(\mu_n,\mu) \to 0$. 
\end{proof}

\begin{remark}
The assumption $\Pers_p(\mu_n) \to \Pers_p(\mu)$ is crucial to obtain $\Dp$-convergence assuming vague convergence. For example, the sequence defined by $\mu_n \defeq \delta_{(n, n+1)}$ converges vaguely to $\mu = 0$ and $(\Pers_p(\mu_n))_n$ does converge (it is constant), while $\Dp(\mu_n, 0) \nrightarrow 0$. This does not contradict Theorem \ref{thm:conv_dp} since $\Pers_p(\mu) = 0 \neq \lim_n \Pers_p(\mu_n)$. 
\end{remark}

Theorem \ref{thm:conv_dp} implies some useful results. First, it entails that the topology of the metric $\Dp$ is stronger than the vague topology. As a consequence, the following corollary holds, using Proposition \ref{prop:diagram_close} ($\PD^p$ is closed in $\MM^p$ for the vague topology).
\begin{corollary}\label{cor:Dp_close_dp}
$\PD^p$ is closed in $\MM^p$ for the metric $\Dp$.
\end{corollary}
We recover in particular that the space $(\PD^p,\Dp)$ is a Polish space (Proposition \ref{prop:MM_p_Polish}), a result already proved in \cite[Theorems 7 and 12]{tda:mukherjee2011probabilitymeasure} with a different approach.

Secondly, we show that the vague convergence of $\mu_n$ to $\mu$ along with the convergence of $\Pers_p(\mu_n) \to \Pers_p(\mu)$ is equivalent to the weak convergence of a weighted measure (see Appendix \ref{subsec:background_measures} for a definition of weak convergence, denoted by $\cvweak$ in the following).
	For $\mu \in \MM^p$, let us introduce the Borel measure with finite mass $\mu^{(p)}$ defined, for a Borel subset $A \subset \Omega$, as:
	\begin{equation}\label{eq:def_mu^p}
		\mu^{(p)}(A) = \int_A d(x,\thediag)^p \dd \mu(x).
	\end{equation}
\begin{corollary}\label{cor:Dp_to_weak_cv} 
	For a sequence $(\mu_n)_n$ and a persistence measure $\mu \in \MM^p$, we have
	\[ \Dp(\mu_n,\mu) \to 0 \text{ if and only if } \mu^{(p)}_n \cvweak \mu^{(p)}. \]
\end{corollary}

\begin{proof}
Consider $\mu, \mu_1, \mu_2, \dots \in \MM^p$ and assume that $\Dp(\mu_n, \mu) \to 0$. By Theorem \ref{thm:conv_dp}, this is equivalent to $\mu_n \cvvague \mu$ and $\mu^{(p)}_n(\Omega) = \Pers_p(\mu_n) \to \Pers_p(\mu) = \mu^{(p)}(\Omega)$. Since for any continuous function $f$ compactly supported, the map $x \mapsto d(x, \thediag)^p f(x)$ is also continuous and compactly supported, $\mu_n \cvvague \mu$ implies $\mu^{(p)}_n \cvvague \mu^{(p)}$. Likewise, the map $x \mapsto d(x, \thediag)^{-p} f(x)$ is continuous and compactly supported, so that $\mu^{(p)}_n \cvvague \mu^{(p)}$ also implies $\mu_n \cvvague \mu$. Hence, $\mu_n \cvvague \mu$ is equivalent to  $\mu^{(p)}_n \cvvague \mu^{(p)}$. By Proposition \ref{prop:weak_vague}, the vague convergence $\mu^{(p)}_n \cvvague \mu^{(p)}$ along with the convergence of the masses is equivalent to $\mu^{(p)}_n \cvweak \mu^{(p)}$. 
\end{proof} 

We end this section with a characterization of relatively compact sets in $(\MM^p, \Dp)$.
\begin{proposition} A set $F$ is relatively compact in $(\MM^p,\Dp)$ if and only if the set $\{\mu^{(p)},\mu \in F \}$ is tight and $\sup_{\mu \in F} \Pers_p(\mu) < \infty$.
\label{prop:prokhorov}
\end{proposition}
\begin{proof} From Corollary \ref{cor:Dp_to_weak_cv}, the relative compactness of a set $F \subset \MM^p$ for the metric $\Dp$ is equivalent to the relative compactness of the set $\{ \mu^{(p)} ,\ \mu \in F \}$ for the weak convergence. Recall that all $\mu^{(p)}$ have a finite mass, as $\mu^{(p)}(\Omega) = \Pers_p(\mu) < \infty$. Therefore, one can use Prokhorov's theorem (Proposition \ref{prop:rel_cpt_weak}) to conclude. 
\end{proof}

\begin{remark}
This characterization is equivalent to the one described in \cite[Theorem 21]{tda:mukherjee2011probabilitymeasure} for persistence diagrams. The notions introduced by the authors of off-diagonally birth-death boundedness, and uniformness are rephrased using the notion of tightness, standard in measure theory.
\end{remark}

We end this section with a remark on the existence of transport maps, assuming that one of the two measures has a density with respect to the Lebesgue measure on $\upperdiag$. We denote by $f_{\#}\mu$ the pushforward of a measure $\mu$ by a map $f$, defined by $f_{\#}\mu(A)=\mu(f^{-1}(A))$ for $A$ a Borel set.
\begin{remark}
	Following \cite[Corollary 2.5]{ot:figalli2010newTransportationDistance}, one can prove that if $\mu \in \MM^2$ has a density with respect to the Lebesgue measure on $\upperdiag$, then for any measure $\nu \in \MM^2$, there exists an unique optimal transport plan $\pi$ between $\mu$ and $\nu$ for the $\mathrm{OT}_2$ metric. The restriction of this transport plan to $\upperdiag\times \groundspace$ is equal to $(\id,T)_{\#}\mu$ where $T:\upperdiag \to \groundspace$ is the gradient of some convex function, whereas the transport plan restricted to $\thediag \times \upperdiag$ is given by $(s,\id)_{\#}(\nu-T_{\#}\mu)$, where $s:\upperdiag \to \thediag$ is the projection on the diagonal.
	A proof of this fact in the context of persistence measures would require to introduce various notions that are out of the scope covered by this paper. We refer the interested reader to \cite[Prop.~2.3]{ot:figalli2010newTransportationDistance} and \cite[Theorem 6.2.4]{ot:ambrosio2008gradient} for details.
\end{remark}

\subsection{Persistence measures in the finite setting}
\label{subsec:uniformly_bounded_mass}
In practice, many statistical results regarding persistence diagrams are stated for sets of diagrams with uniformly bounded number of points \cite{tda:kwitt2015statisticalTDAkernelPerspective,tda:carriere2017sliced}, and the specific properties of $\Dp$ in this setting are therefore of interest. 
Introduce for $m \geq 0$ the subset $\MM^p_{\leq m}$ of $\MM^p$ defined as $\MM^p_{\leq m} \defeq \{ \mu \in \MM^p,\ \mu(\upperdiag) \leq m \}$, and the set $\MM_f^p$ of finite persistence measures, $\MM^p_f \defeq \bigcup_{m \geq 0} \MM^p_{\leq m}$. Define similarly the set $\PD_{\leq m}$ (resp.~$\PD_f$). Note that the assumption $\Pers_p(a) < \infty$ is always satisfied for a finite diagram $a$ (which is not true for general Radon measures), so that the exponent $p$ is not needed when defining $\PD_{\leq m}$ and $\DD_f$. 
\begin{proposition}
$\MM^p_f$ (resp. $\PD_f$) is dense in $\MM^p$ (resp. $\PD^p$) for the metric $\Dp$.
\label{prop:finite-to-global}
\end{proposition}

\begin{proof}
This is a straightforward consequence of Lemma \ref{lemma:mu_r_to_mu}. Indeed, if $\mu\in \MM^p$ and $r>0$, then \eqref{eq:pers_compact} implies that $\mu^{(r)}$ is of finite mass.
\end{proof}

Let $\tilde{\upperdiag} = \upperdiag \sqcup \{\thediag\}$ be the quotient of $\groundspace$ by the closed subset $\thediag$---i.e.~we encode the diagonal by just one point (still denoted by $\thediag$). The distance $d$ on $\groundspace^2$ induces naturally a function $\tilde{d}$ on $\tilde{\upperdiag}^2$, defined for $x,y \in \upperdiag$ by $\tilde{d}(x,y) = d(x,y)$, $\tilde{d}(x,\thediag) = \tilde{d}(\thediag,x) = d(x,s(x))$ and $\tilde{d}(\thediag,\thediag) = 0$. However, $\tilde{d}$ is not a distance since one can have $\tilde{d}(x,y) > \tilde{d}(x,\thediag) + \tilde{d}(y,\thediag)$. Define 
\begin{equation}
	\rho(x,y) \defeq \min \{ \tilde{d}(x,y), \tilde{d}(x,\thediag) + \tilde{d}(y,\thediag)\}.
\label{eq:def_rho}
\end{equation}
It is straightforward to check that $\rho$ is a distance on $\tilde{\upperdiag}$ and that $(\tilde{\upperdiag},\rho)$ is a Polish space. One can then define the Wasserstein distance $W_{p,\rho}$ with respect to $\rho$ for finite measures on $\tilde{\upperdiag}$ which have the same masses, that is the infimum of $\tilde{C}_p(\tilde{\pi}) \defeq \iint_{\tilde{\upperdiag}^2} \rho(x,y)^p \dd \tilde{\pi}(x,y)$, for $\tilde{\pi}$ a transport plan with corresponding marginals (see Section \ref{subsec:wasserstein}). The following theorem states that {the problem of computing the $\Dp$ metric between two persistence measures with finite masses can be turn into the one of computing the Wasserstein distances between two measures supported on $\tilde{\upperdiag}$ with the same mass. For the sake of simplicity, we assume in the following that $d(x,y) = \|x-y\|_q$ with $q > 1$, ensuring that the quantity $\argmin_{y \in \thediag} d(x,y)$ is reduced to the orthogonal projection $s(x)$ of $x$ onto the diagonal $\thediag$. The following result could be seamlessly adapted to the case $q=1$.

\begin{proposition}\label{prop:Dp_equals_Wp_finite}
Let $\mu,\nu \in \MM_f^p$ and $r\geq \mu(\Omega) + \nu(\Omega)$. Define $\tilde{\mu} = \mu + (r-\mu(\Omega))\delta_{\thediag}$ and $\tilde{\nu} = \nu + (r-\nu(\Omega))\delta_{\thediag}$. Then $\Dp(\mu,\nu) = W_{p,\rho}(\tilde{\mu},\tilde{\nu})$.
\end{proposition}

Before proving Proposition \ref{prop:Dp_equals_Wp_finite}, we need the two following lemmas:
\begin{lemma}\label{lemma:kappa_iota}
Let $\mu,\nu \in \MM^p_f$ and $r\geq \max(\mu(\upperdiag),\nu(\upperdiag))$. Let $\tilde{\mu} \defeq \mu + (r-\mu(\upperdiag))\delta_{\thediag}$, $\tilde{\nu} \defeq \nu + (r-\nu(\upperdiag))\delta_{\thediag}$ and $s: \upperdiag \to \thediag$ be the orthogonal projection on the diagonal. 
\begin{enumerate}
\item Define $T(\mu,\nu)$ the set of plans $\pi \in \Adm(\mu,\nu)$ satisfying $\pi(\{(x,y)\in \upperdiag \times \thediag,\ y\neq s(x) \} ) = \pi(\{(x,y)\in \thediag \times \upperdiag,\ x \neq s(y) \} ) = 0$ along with $\pi(\thediag \times \thediag) = 0$. Then, $\Opt_p(\mu,\nu) \subset T(\mu,\nu)$. \label{it:opt_subset_t}
\item Let $\pi \in T(\mu,\nu)$ be such that $\mu(\Omega)  +  \pi(\thediag\times \upperdiag)  \leq  r$. Define  $\iota(\pi) \in \Pi(\tilde{\mu},\tilde{\nu})$ by, for Borel sets $A, B \subset \upperdiag$,
\begin{equation}
\begin{cases}
	\iota(\pi)(A \times B) = \pi(A \times B),\\
	\iota(\pi) (A \times \{\thediag\}) = \pi(A \times \thediag), \\
	\iota(\pi)(\{ \thediag \}\times B) = \pi(\thediag \times B), \\
	\iota(\pi)(\{\thediag \} \times \{\thediag \}) = r - \mu(\Omega) - \pi(\thediag \times \upperdiag) \geq 0.
\end{cases}
\label{eq:pi_to_pi_tilde}
\end{equation}
Then, $C_p(\pi) = \iint_{\tilde{\upperdiag} \times \tilde{\upperdiag}} \tilde d(x,y)^p \dd \iota(\pi)(x,y)$. \label{it:iota}
\item Let $\tilde{\pi} \in  \Pi(\tilde{\mu},\tilde{\nu})$. Define $\kappa(\tilde{\pi}) \in T(\mu,\nu)$  by,
\[
	\begin{cases}
		\kappa(\tilde{\pi})(A \times B) = \tilde{\pi}(A \times B) & \text{ for } A, B \subset \upperdiag, \\
		\kappa(\tilde{\pi})(A \times B) = \tilde{\pi}((A \cap s^{-1} (B)) \times \{\thediag\}) & \text{ for } A \subset \upperdiag, B \subset \thediag, \\
		\kappa(\tilde{\pi})(A \times B) = \tilde{\pi}(\{\thediag\} \times (B \cap s^{-1} (A))) & \text{ for  } A \subset \thediag, B \subset \upperdiag,  \\
		\kappa(\tilde{\pi})(\thediag, \thediag) = 0.
	\end{cases}
\]
Then, $\iint_{\tilde{\upperdiag} \times \tilde{\upperdiag}} \tilde d(x,y)^p \dd \tilde{\pi}(x,y) = C_p(\kappa(\tilde{\pi}))$. \label{it:kappa}
\end{enumerate}
\end{lemma}

\begin{proof}~
\begin{enumerate}
\item Consider $\pi \in \Adm(\mu,\nu)$, and define $\pi'$ that coincides with $\pi$ on $\upperdiag \times \upperdiag$, and is such that we enforce mass transported on the diagonal to be transported on its orthogonal projection: more precisely, for all Borel set $A \subset \upperdiag$, $B \subset \thediag$, $\pi'(A \times B) = \pi((s^{-1}(B)\cap A) \times B)$ and $\pi'(B \times A) = \pi(B \times (s^{-1}(B)\cap A) )$. Note that $\pi' \in T(\mu,\nu)$. Since $s(x)$ is the unique minimizer of $y \mapsto d(x,y)^p$, it follows that $C_p(\pi') \leq C_p(\pi)$, with equality if and only if $\pi \in T(\mu,\nu)$, and thus $\Opt_p(\mu,\nu) \subset T(\mu,\nu)$. 
\item Write $\tilde{\pi} = \iota(\pi)$. The mass $\tilde{\pi}(\{\thediag \} \times \{\thediag \})$ is nonnegative by definition. One has for all Borel sets $A \subset \upperdiag$,
\begin{align*}
\tilde{\pi}(A \times \tilde{\Omega}) &= \tilde{\pi}(A \times \upperdiag) + \tilde{\pi}(A \times \{\thediag\})\\
&= \pi(A \times \upperdiag) + \pi(A \times \thediag) = \pi(A \times \groundspace) = \mu(A) = \tilde{\mu}(A).
\end{align*} 
Similarly, $\tilde{\pi}(\tilde{\upperdiag} \times B) = \tilde{\nu}(B)$ for all $B \subset \upperdiag$. Observe also that
\[ \tilde{\pi}(\{\thediag\} \times \tilde{\upperdiag}) = \tilde{\pi} (\{\thediag\} \times \{\thediag\}) + \tilde{\pi} (\{\thediag\} \times \upperdiag) = r - \mu(\upperdiag) = \tilde{\mu}(\{\thediag\}). \]
Similarly, $\tilde{\pi}(\tilde{\upperdiag} \times \{\thediag\} ) = \tilde{\nu}(\{\thediag\})$. It gives that $\iota(\pi) \in \Pi(\tilde{\mu},\tilde{\nu})$, so that $\iota$ is well defined. Observe that 
\begin{align*}
	\iint_{\tilde{\upperdiag} \times \tilde{\upperdiag}} \tilde d(x,y)^p \dd \tilde{\pi}(x,y) &= \iint_{\upperdiag \times \upperdiag} d(x,y)^p \dd \pi(x,y) \\
&\quad+ \int_\upperdiag d(x,\thediag)^p \dd \pi(x,\thediag) \\
&\quad+\int_\upperdiag d(\thediag,y)^p \dd \pi(\thediag,y) + 0 \\
	&= C_p(\pi)\text{ as }\pi\in T(\mu,\nu).
\end{align*}
\item Write $\pi = \kappa(\tilde{\pi})$. For $A\subset  \Omega$ a Borel set,
\begin{align*}
\pi(A \times \groundspace)  &= \pi(A\times \upperdiag) + \pi(A\times \thediag) \\
&=\tilde{\pi}(A\times  \upperdiag)+ \tilde{\pi}(A\times \{\thediag\}) = \tilde{\pi}(A \times \tilde{\Omega}) = \mu(A).
\end{align*} 
Similarly, $\pi(\groundspace \times B) = \nu(B)$ for all $B \subset \upperdiag$. Therefore, $\pi\in \Adm(\mu,\nu)$, and by construction, if a point $x\in \Omega$ is transported on $\thediag$, it is transported on $s(x)$, so that $\pi \in T(\mu,\nu)$. Observe that $\mu(\Omega) +  \pi(\thediag \times  \Omega) \leq  \tilde{\pi}(\tilde{\Omega}\times \tilde{\Omega})=r$, so that $\iota(\pi)$ is well defined. Also, $\iota(\pi) = \tilde{\pi}$, so that, according to point \ref{it:iota}, $C_p(\pi) = \iint_{\tilde{\upperdiag} \times \tilde{\upperdiag}} \tilde d(x,y)^p \dd \tilde{\pi}(x,y)$.
\end{enumerate} 
\end{proof}

We show that the inequality $\Dp(\mu,\nu) \leq  W_{p,\rho}(\tilde{\mu},\tilde{\nu})$ holds as long as $r\geq \max(\mu(\Omega),\nu(\Omega))$.
\begin{lemma}\label{lem:Dp_smaller_than_W_p} Let $\mu,\nu \in \MM^p_f$ and $r\geq \max(\mu(\upperdiag),\nu(\upperdiag))$. Let $\tilde{\mu} \defeq \mu + (r-\mu(\upperdiag))\delta_{\thediag}$, $\tilde{\nu} \defeq \nu + (r-\nu(\upperdiag))\delta_{\thediag}$. Then, $\Dp(\mu,\nu) \leq W_{p,\rho}(\tilde{\mu},\tilde{\nu})$.
\end{lemma}
\begin{proof}
Let $\tilde{\pi} \in \Pi(\tilde{\mu}, \tilde{\nu})$. Define the set $H \defeq \{(x,y)\in \tilde{\upperdiag}^2,\ \rho(x,y) = d(x,y) \}$,
and let $H^c$ be its complementary set in $\tilde{\upperdiag}^2$, i.e.~the set where $\rho(x,y) = d(x,\thediag) + d(\thediag,y)$.
Define $\tilde{\pi}' \in \MM(\tilde{\upperdiag}^2)$ by, for Borel sets $A,B \subset \upperdiag$:
\[\begin{cases}
\tilde{\pi}'(A\times B) = \tilde{\pi}((A\times B)\cap H) \\
\tilde{\pi}'(A \times \{\thediag \}) = \tilde{\pi}((A \times \tilde{\upperdiag})\cap H^c) + \tilde{\pi}(A \times  \{\thediag \})\\
\tilde{\pi}'(\{\thediag \}\times B) = \tilde{\pi}((\tilde{\upperdiag} \times B)\cap H^c) + \tilde{\pi}(\{\thediag\}\times B).
\end{cases}\]
We easily check that $\tilde{\pi}' \in \Pi(\tilde{\mu},\tilde{\nu})$. Also, using $(a+b)^p\geq a^p+b^p$ for positive $a,b$, we have
\begin{align*}
\iint_{\tilde{\upperdiag} \times \tilde{\upperdiag}} \rho(x,y)^p \dd \tilde{\pi}(x,y) &= \iint_{H} \tilde d(x,y)^p \dd \tilde{\pi}(x,y) \\
&\quad+ \iint_{H^c} (\tilde d(x,\thediag)+\tilde d(\thediag, y))^p \dd \tilde{\pi}(x,y) \\
&\geq \iint_{H}\tilde d(x,y)^p \dd \tilde{\pi}'(x,y) \\
&\quad+ \iint_{H^c} \left(\tilde d(x,\thediag)^p + \tilde d(y,\thediag)^p \right) \dd \tilde{\pi}(x,y) \\
&= \iint_{\tilde{\upperdiag} \times \tilde{\upperdiag}} d(x,y)^p \dd \tilde{\pi}'(x,y) \\
&\geq \inf_{\tilde{\pi}' \in \Pi(\tilde{\mu}, \tilde{\nu})} \iint_{\tilde{\upperdiag} \times \tilde{\upperdiag}}\tilde d(x,y)^p \dd \tilde{\pi}'(x,y).
\end{align*} 
We conclude by taking the infimum on $\tilde{\pi}$ that 
\[ W_{p,\rho}(\tilde{\mu}, \tilde{\nu}) \geq \inf_{\tilde{\pi}' \in \Pi(\tilde{\mu}, \tilde{\nu})} \iint_{\tilde{\upperdiag} \times \tilde{\upperdiag}} \tilde d(x,y)^p \dd \tilde{\pi}'(x,y).\] 
Since $\rho(x,y) \leq \tilde d(x,y)$, it follows that 
\begin{equation} W_{p,\rho}^p(\tilde{\mu},\tilde{\nu}) = \inf_{\tilde{\pi} \in \Pi(\tilde{\mu}, \tilde{\nu})} \iint_{\tilde{\upperdiag}^2} \tilde d(x,y)^p \dd \tilde{\pi}(x,y).
\label{eq:from_rho_to_d}
\end{equation}
Since $\tilde d$ is continuous, the infimum in the right hand side of \eqref{eq:from_rho_to_d} is reached \cite[Theorem 4.1]{ot:villani2008optimal}. Consider thus $\tilde{\pi} \in \Pi(\tilde{\mu}, \tilde{\nu})$ which realizes the infimum. We can write, using Lemma \ref{lemma:kappa_iota},
\begin{align*}
	W_{p,\rho}^p(\tilde{\mu}, \tilde{\nu}) &=\iint_{\tilde{\upperdiag}^2}\tilde d(x,y)^p \dd \tilde{\pi}(x,y) = \iint_{\groundspace \times \groundspace} d(x,y)^p \dd \kappa(\tilde{\pi})(x,y) \\
								   &\geq \inf_{\pi \in T(\mu, \nu)} \iint_{\groundspace \times \groundspace} d(x,y)^p \dd \pi(x,y) = \Dp^p(\mu, \nu),
\end{align*}
which concludes the proof. 
\end{proof}

\begin{proof}[Proof of Proposition \ref{prop:Dp_equals_Wp_finite}]
Let $\pi \in T(\mu,\nu)$. As $\mu(\Omega) + \pi(\thediag \times \upperdiag) \leq \mu(\Omega) +  \nu(\Omega) \leq r$, one can define $\tilde{\pi} = \iota(\pi)$. Since $\rho(x,y) \leq \tilde d(x,y)$, we have $\tilde{C}_p(\tilde{\pi}) \leq \iint \tilde d(x,y)^p \dd \tilde{\pi}(x,y) = C_p(\pi)$ (Lemma \ref{lemma:kappa_iota}). Taking infimum gives $W_{p,\rho}(\tilde{\mu}, \tilde{\nu}) \leq \Dp(\mu,\nu).$ The other inequality holds according to Lemma \ref{lem:Dp_smaller_than_W_p}. 
\end{proof}

\begin{remark}
The starting idea of this theorem---informally,``adding the mass of one diagram to the other and vice-versa"---is known in TDA as a \emph{bipartite graph matching} \cite[Ch.~VIII.4]{tda:edelsbrunner2010computational} and used in practical computations \cite{tda:kerber2017geometryHelps}. Here, Proposition \ref{prop:Dp_equals_Wp_finite} states that solving this bipartite graph matching problem can be formalized as computing a Wasserstein distance on the metric space $(\tilde{\Omega}, \rho)$ and as such, makes sense (and remains true) for more general measures.
\end{remark}

\begin{remark} Proposition \ref{prop:Dp_equals_Wp_finite} is useful for numerical purposes since it allows us in applications, when dealing with a finite set of finite measures (in particular diagrams), to directly use the various tools developed in computational optimal transport \cite{ot:CuturiPeyre2017COT} to compute Wasserstein distances. This alternative to the combinatorial algorithms considered in the literature \cite{tda:kerber2017geometryHelps,tda:turner2014frechet} is studied in details in \cite{tda:lacombe2018large}. This result is also helpful to prove the existence of $p$-Fr\'echet means of sets of persistence measures (see Section \ref{sec:barycenter}).
\end{remark}

\subsection{The $\Dinf$ distance}
\label{subsec:bottleneck}

In classical optimal transport, the $\infty$-Wasserstein distance is known to have a much more erratic behavior than its $p<\infty$ counterparts \cite[Section 5.5.1]{otam}. However, in the context of persistence diagrams, the $d_\infty$ distance defined in Eq.~\eqref{eq:dgm_bottleneck_metric} appears naturally as an interleaving distance between persistence modules and satisfies strong stability results: it is thus worthy of interest. It also happens that, when restricted to diagrams having some specific finiteness properties, most irregular behaviors are suppressed and a convenient characterization of convergence exists. 
\begin{definition}
Let $\supp(\mu)$ denote the support of a measure $\mu$ and define $\Pers_\infty(\mu) \defeq \sup \{ d(x, \thediag) ,\ x \in \supp(\mu) \}$. Let
\begin{equation}
	\MM^\infty \defeq \{ \mu \in \MM ,\ \Pers_\infty(\mu) < \infty \} \quad \text{ and } \quad \PD^\infty \defeq \PD \cap \MM^\infty.
\end{equation}
For $\mu,\nu \in \MM^\infty$ and $\pi \in \Adm(\mu,\nu)$, let $C_\infty(\pi) \defeq \sup \{ d(x,y),\ (x,y) \in \supp(\pi)\}$ and let
\begin{equation} \label{eq:def_dinfty}
\Dinf(\mu,\nu) \defeq  \inf_{\pi \in \Adm(\mu,\nu)} C_\infty(\pi).
\end{equation}
The set of transport plans minimizing \eqref{eq:def_dinfty} is denoted by $\Opt_\infty(\mu,\nu)$.
\end{definition}

Recall that $E_\upperdiag=(\groundspace\times \groundspace)\backslash(\thediag\times \thediag)$.
\begin{proposition}\label{prop:basic_infty} Let $\mu, \nu \in \MM^\infty$. For the VM topology on $E_\upperdiag$,
\begin{itemize}
\item the map $\pi \in \Adm(\mu,\nu) \mapsto C_\infty(\pi)$ is lower semi-continuous.
\item The set $\Opt_\infty(\mu,\nu)$ is a non-empty sequentially compact set.
\item $\Dinf$ is lower semi-continuous.
\end{itemize}
Moreover, $\Dinf$ is a metric on $\MM^\infty$.
\end{proposition}
The proofs of these results are found in Appendix \ref{sec:proofs}. 

As in the case $p < \infty$, $\Dinf$ and $d_\infty$ coincide on $\PD^ \infty$.
\begin{proposition} \label{prop:dinf_equal_Dinf}
For $a,b \in \PD^\infty$, $\Dinf(a,b) = d_\infty(a,b)$.
\end{proposition}

\begin{proof}
Consider two diagrams $a,b \in \PD^\infty$, written as $a = \sum_{i \in I} \delta_{x_i}$ and $b = \sum_{j \in J} \delta_{y_j}$, where $I,J \subset \N^*$ are (possibly infinite) sets of indices.
The marginals constraints imply that a plan $\pi \in \Adm(\mu,\nu)$ is supported on $(\{x_i\}_i \cup \thediag)\times (\{y_j\}_j \cup \thediag)$. If some of the mass $\pi(\{x_i\},\thediag)$ (resp.~$\pi(\thediag,\{y_j\})$) is sent on a point other than the projection of $x_i$ (resp.~$y_j$) on the diagonal $\thediag$, then the cost of such a plan can always be (strictly if $q>1$) reduced. Introduce the matrix $C$ indexed on $(-J \cup I) \times (-I \cup J)$ defined by
	\begin{equation}
		\begin{cases}
			C_{i,j} = d(x_i, y_j) & \text{ for } i,j > 0, \\
			C_{i,j} = d(\thediag, y_j) & \text{ for } i < 0 , j > 0, \\
			C_{i,j} = d(x_i, \thediag) & \text{ for } i > 0, j < 0, \\
			C_{i,j} = 0 & \text { for } i,j < 0.
			\end{cases}
		\label{eq:cost_matrix}
	\end{equation}
In this context, an element of $\Opt(a,b)$ can be written a matrix $P$ indexed on $(-J \cup I) \times (-I \cup J)$, and marginal constraints state that $P$ must belong to the set of doubly stochastic matrices $\SS$. Therefore, $\Dinf(a,b) = \inf_{P \in \SS} \sup \{ C_{i,j} ,\  (i,j) \in \supp(P) \}$, where $\SS$ is the set of doubly stochastic matrices indexed on $(-J \cup I) \times (-I \cup J)$, and $\supp(P)$ denotes the support of $P$, that is the set $\{(i,j),\ P_{i,j} > 0\}$. 

Let $P\in  \SS$. For any $k \in \N$, and any set of distinct indices $\{i_1, \dots ,i_k \} \subset -J\cup I$, we have
\[
	k = \sum_{k' = 1}^k \underbrace{\sum_{j \in -I\cup J} P_{i_{k'}, j}}_{=1} = \sum_{j \in -I\cup J} \underbrace{\sum_{k' = 1}^k P_{i_{k'}, j}}_{\leq 1}.
\]
Thus, the cardinality of $\{j ,\ \exists k' \text{ such that } (i_{k'},j) \in \supp(P) \}$ must be larger than $k$. Said differently, the marginals constraints impose that any set of $k$ points in $a$ must be matched to \emph{at least} $k$ points in $b$ (points are counted with eventual repetitions here). Under such conditions, the Hall's marriage theorem (see \cite[p.~51]{ot:hall1967combinatorial}) guarantees the existence of a permutation matrix $P'$ with $\supp(P') \subset \supp(P)$. As a consequence,
\begin{align*}
		 \sup \{ C_{i,j} ,\ (i,j) \in \supp(P) \} &\geq \sup \{ C_{i,j} ,\ (i,j) \in \supp(P') \} \\
		 &\geq \inf_{P' \in \SS'} \sup \{ C_{i,j} ,\ (i,j) \in \supp(P') \} = d_\infty(a,b),
\end{align*}
where $\SS'$ denotes the set of permutations matrix indexed on $(-J \cup I) \times (-I \cup J)$. Taking the infimum on $P \in \SS$ on the left-hand side and using that $\SS' \subset \SS$ finally gives that $\Dinf (a,b) = d_\infty(a,b)$. 
\end{proof}

\begin{proposition}
The space $(\MM^\infty, \Dinf)$ is complete.
\label{prop:d_infty_polish}
\end{proposition}

\begin{proof}
Let $(\mu_n)_n$ be a Cauchy sequence for $\Dinf$. Fix a compact $K\subset \upperdiag$, and pick $\epsilon = d(K,\thediag) / 2$. There exists $n_0$ such that for $n > n_0$, $\Dinf(\mu_n, \mu_{n_0}) < \epsilon$. Let $K_\epsilon \defeq \{ x \in \upperdiag,\ d(x,K) \leq \epsilon \}$. By considering $\pi_n \in \Opt_\infty(\mu_n, \mu_{n_0})$, and since $\Dinf(\mu_n, \mu_{n_0}) < \epsilon$, we have that
	\begin{equation}
		\mu_n(K) = \pi_n(K \times \groundspace) = \pi_n(K \times K_\epsilon) \leq \mu_{n_0}(K_\epsilon).
	\label{eq:mu_n_K_is_bounded}
	\end{equation}
Therefore, $(\mu_n(K))_n$ is uniformly bounded, and Proposition \ref{prop:rel_cpt_vague} implies that $(\mu_n)_n$ is relatively compact. Finally, the exact same computations as in the proof of the completeness for $p<\infty$ (see Appendix \ref{sec:proofs}) show that $(\mu_n)_n$ converges for the $\Dinf$ metric. 
\end{proof}

\begin{remark} Contrary to the case $p < \infty$, the space $\PD^\infty$ (and therefore $\MM^\infty$) is not separable. Indeed, for $I \subset \N$, define the diagram $a_I \defeq \sum_{i \in I} \delta_{(i, i+1)} \in \PD^\infty$. The family $\{ a_I,\ I \subset \N \}$ is uncountable, and for two distinct $I, I'$, $\Dinf(a_I, a_{I'}) = \frac{\sqrt{2}}{2}$. This result is similar to \cite[Theorem 4.20]{tda:bubenik2018topological}.
\end{remark}

We now show that the direct implication in Theorem \ref{thm:conv_dp} still holds in the case $p=\infty$.

\begin{proposition} 	\label{prop:Dinf_implies_vaguecv}
Let $\mu, \mu_1, \mu_2,\dots$ be measures in $\MM^\infty
$. If $\Dinf(\mu_n,\mu) \to 0$, then $(\mu_n)_n$ converges vaguely to $\mu$ and $\Pers_\infty(\mu_n)$ converges to $\Pers_\infty(\mu)$.
\end{proposition}

\begin{proof} First, the convergence of $\Pers_\infty(\mu_n)$ towards $\Pers_\infty(\mu)$ is a consequence of the reverse triangle inequality:
\[ |\Pers_\infty(\mu_n)-\Pers_\infty(\mu)|=|\Dinf(\mu_n,0)-\Dinf(\mu,0)|\leq \Dinf(\mu_n,\mu),\]
which converges to $0$ as $n$ goes to $\infty$. 

We now prove the vague convergence. Let $f\in C_c(\upperdiag)$, whose support is included in some compact set $K$. For any $\epsilon >0$, there exists a $L$-Lipschitz function $f_\epsilon$, whose support is included in $K$, with $\| f - f_\epsilon \|_\infty \leq \epsilon$. Observe that $\sup_k \mu_k(K) < \infty$ using the same arguments than for \eqref{eq:mu_n_K_is_bounded}. Let $\pi_n \in \Opt_\infty(\mu_n,\mu)$. We have 
\begin{align*}
|\mu_n(f)-\mu(f)| &\leq |\mu_n(f-f_\epsilon)| + |\mu(f-f_\epsilon)|  + |\mu_n(f_\epsilon)-\mu(f_\epsilon)| \\
&\leq (\mu_n(K) + \mu(K)) \epsilon + |\mu_n(f_\epsilon)-\mu(f_\epsilon)| \\
&\leq (\sup_k \mu_k(K) + \mu(K))\epsilon + |\mu_n(f_\epsilon)-\mu(f_\epsilon)|.
\end{align*}
Also, 
\begin{align*}
&|\mu_n(f_\epsilon)-\mu(f_\epsilon)| = \left|\iint_{\groundspace^2}(f_\epsilon(x)-f_\epsilon(y))\dd \pi_n(x,y) \right| \\
&\qquad \qquad\leq \iint_{\groundspace^2} |f_\epsilon(x)-f_\epsilon(y)| \dd\pi_n(x,y) \\
&\qquad \qquad\leq  L \iint\limits_{(K \times \groundspace) \cup (\groundspace \times K)} d(x,y)\dd\pi_n(x,y) \text{ as $f_\epsilon$ is $L$-Lipschitz continuous} \\
&\qquad \qquad\leq L C_\infty(\pi_n) (\pi_n(K \times  \groundspace) + \pi_n(\groundspace\times K)) \\
&\qquad \qquad\leq L  \Dinf(\mu_n,\mu) \left(\sup_k \mu_k(K) + \mu(K)\right).
\end{align*}
This last quantity converge to $0$ as $n$ goes to $\infty$ for fixed $\epsilon$. Therefore, taking the $\limsup$ in $n$ and then letting $\epsilon$ go to $0$, we obtain that $\mu_n(f) \to \mu(f)$. 

\end{proof}

\begin{remark} As for the case $1 \leq p < \infty$, Proposition \ref{prop:Dinf_implies_vaguecv} implies that $\Dinf$ metricizes the vague convergence, and thus using Propositions \ref{prop:dinf_equal_Dinf} and \ref{prop:diagram_close}, we have that $(\PD^\infty, d_\infty)$ is closed in $(\MM^\infty, \Dinf)$ and is---in particular---complete.
\end{remark}

Contrary to the $p< \infty$ case, a converse of Proposition \ref{prop:Dinf_implies_vaguecv} does not hold, even on the subspace of persistence diagrams (see Figure \ref{fig:ex_bottleneck}). To recover a space with a structure more similar to $\PD^p$, it is useful to look at a smaller set. 
Introduce $\PD^\infty_0$ the set of persistence diagrams such that for all $r >0$, there is a finite number of points of the diagram of persistence larger than $r$ and recall that $\PD_f$ denotes the set of persistence diagrams with finite number of points.

\begin{figure}
	\center
	\includegraphics[width=\textwidth]{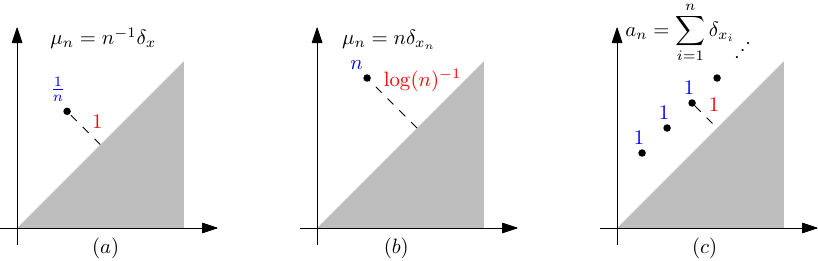}
	\caption{Illustration of differences between $\Dp$, $\Dinf$, and vague convergences. Blue color represents the mass on a point while red color designates distances. $(a)$ A case where $\Dp(\mu_n, 0) \to 0$ for any $p < \infty$ while $\Dinf(\mu_n, 0) = 1$. $(b)$ A case where $\Dinf(\mu_n, 0) \to 0$ while for all $p < \infty$, $\Dp(\mu_n, \mu) \to \infty$. $(c)$ A sequence of persistence diagrams $a_n \in \PD^\infty$, where $(a_n)_n$ converges vaguely to $a = \sum_{i} \delta_{x_i}$ and $\Pers_\infty(a_n)=\Pers_\infty(a)$, but $(a_n)$ does not converge to $a$ for $\Dinf$.}
	\label{fig:ex_bottleneck}
\end{figure}

\begin{proposition}
	The closure of $\PD_f$ for the distance $\Dinf$ is $\PD_0^\infty$.
	\label{prop:closure_of_finiteDgm}
\end{proposition}

\begin{proof}
Consider $a \in \PD_0^\infty$. By definition, for all $n \in \N$, $a$ has a finite number of points with persistence larger than $\frac{1}{n}$, so that the restriction $a_n$ of $a$ to points with persistence larger than $\frac{1}{n}$ belongs to $\PD_f$. As $\Dinf(a, a_n) \leq \frac{1}{n} \rightarrow 0$, $\PD_0^\infty$ is contained in the closure of $\PD_f$.

Conversely, consider a diagram $a \in \PD^\infty \backslash \PD_0^\infty$. There is a constant $r > 0$ such that $a$ has infinitely many points with persistence larger than $r$. For any finite diagram $a' \in \PD_f$, we have $\Dinf(a', a) \geq r$, so that $a$ is not the limit for the $\Dinf$ metric of any sequence in $\PD_f$. 
\end{proof}

\begin{remark}
	The space $\PD_0^\infty$ is exactly the set introduced in \cite[Theorem 3.5]{tda:blumberg2014robust} as the completion of $\PD_f$ for the bottleneck metric $d_\infty$. Here, we recover that $\PD_0^\infty$ is complete as a closed subset of the complete space $\PD^\infty$.
\end{remark}

Define for $r>0$ and $a\in \PD$, $a^{(r)}$ the persistence diagram restricted to $\{x\in \upperdiag,\ d(x, \thediag) > r\}$ (as in Lemma \ref{lemma:mu_r_to_mu}). The following characterization of convergence holds in $\PD_0^\infty$. 
\begin{proposition} \label{prop:nice_conv_dinf}
	Let $a, a_1, a_2,\dots$ be persistence diagrams in $\PD_0^\infty$. Then,
	\[\Dinf(a_n, a) \to 0 \Leftrightarrow \begin{cases} a_n \cvvague a, \\ (a_n^{(r)})_n \text{ is tight for all positive $r$}. \end{cases} \]
\end{proposition}

\begin{proof}
Let us prove first the direct implication. Proposition \ref{prop:Dinf_implies_vaguecv} states that the convergence with respect to $\Dinf$ implies the vague convergence. Fix $r>0$. By definition, $a^{(r)}$ is made of a finite number of points, all included in some open bounded set $U \subset \upperdiag$. As $a_n^{(r)}(U^c)$ is a sequence of integers, the bottleneck convergence implies that for $n$ large enough, $a_n^{(r)}(U^c)$ is equal to 0. Thus, $(a_n^{(r)})_n$ is tight.

Let us prove the converse. Consider $a \in \PD_0^\infty$ and a sequence $(a_n)_n$ that converges vaguely to $a$, with $(a_n^{(r)})$ tight for all $r>0$. Fix $r >0$ and let $x_1,\dots,x_K$ be an enumeration of the points in $a^{(r)}$, the point $x_k$ being present with multiplicity $m_k \in \N$. Denote by $B(x,\epsilon)$ (resp. $\overline{B}(x,\epsilon)$) the open (resp. closed) ball of radius $\epsilon$ centered at $x$. By the Portmanteau theorem, for $\epsilon$ small enough,
\[
	\begin{cases}
		\displaystyle \liminf_{n\to \infty} a_n(B(x_k,\epsilon)) \geq a(B(x_k,\epsilon)) = m_k \\
		\displaystyle \limsup_{n\to \infty} a_n(\overline{B}(x_k,\epsilon)) \leq a(\overline{B}(x_k,\epsilon)) = m_k,
	\end{cases}
\]
so that, for $n$ large enough, there are exactly $m_k$ points of $a_n$ in $B(x_k,\epsilon)$ (since $(a_n(B(x_k, \epsilon)))_n$ is a converging sequence of integers). The tightness of $(a_n^{(r)})_n$ implies the existence of some compact $K \subset \upperdiag$ such that for $n$ large enough, $a_n^{(r)}(K^c) = 0$ (as the measures take their values in $\N$). Applying Portmanteau's theorem to the closed set $K' \defeq K \backslash \bigcup_{i=1}^K B(x_i, \epsilon)$ gives 
\[ \limsup_{n \to \infty} a_n^{(r)}(K') \leq a^{(r)}(K') = 0. \]
This implies that for $n$ large enough, there are no other points in $a_n$ with persistence larger than $r$ and thus $\Dinf(a^{(r)},a_n)$ is less than or equal to $r + \epsilon$. Finally,
\[ \limsup_{n\to \infty} \Dinf(a_n,a) \leq \limsup_{n\to \infty} \Dinf(a_n,a^{(r)}) +r \leq 2r + \epsilon.\]
Letting $\epsilon \to 0$ then $r \to 0$, the bottleneck convergence holds. 
\end{proof}

\section{$p$-Fr\'echet means for distributions supported on $\MM^p$}
\label{sec:barycenter}
In this section, we state the existence of $p$-Fr\'echet means for probability distributions supported on $\MM^p$. We start with the finite case (i.e.~averaging finitely many persistence measures) and then extend the result to any probability distribution with finite $p$-th moment. We then study the specific case of distribution supported on $\PD^p$ (i.e.~averaging persistence diagrams), and show that in the finite setting, the set of $p$-Fr\'echet means is a convex set whose extreme points are in $\PD^d$ (i.e.~are actual persistence diagrams). 

\begin{remark}
In this section, we will assume that $1 <p <\infty$ and $1 < q < \infty$ (recall $d( \cdot, \cdot) = \| \cdot - \cdot \|_q$). These assumptions will ensure that $(i)$ the projection of $x \in \upperdiag$ onto $\thediag$ is uniquely defined and $(ii)$ the $p$-Fr\'echet mean of $k$ points $x_1 \dots x_k$ in $\groundspace$, i.e.~minimizer of $x \mapsto \sum_{i=1}^k \|x - x_i\|_q^p$, is also uniquely defined; two facts used in our proofs. 
\end{remark}

Recall that $(\MM^p, \Dp)$ is a Polish space, and let $W_{p,\Dp}$ denote the Wasserstein distance (see Section \ref{subsec:wasserstein}) between probability measures supported on $(\MM^p, \Dp)$. We denote by $\WW^p(\MM^p)$ the space of probability measures $\P$ supported on $\MM^p$, equipped with the $W_{p,\Dp}$ metric, which are at a finite distance from $\delta_\emptydgm$---the Dirac mass supported on the empty diagram---i.e.~\[W_{p,\Dp}^p(\P,\delta_{\emptydgm}) = \int_{\nu \in \MM^p} \Dp^p(\nu, 0) \dd \P(\nu) = \int_{\nu \in \MM^p} \Pers_p(\nu) \dd \P(\nu)< \infty.\]

\begin{definition}\label{def:variance_and_bary} 
Consider $\P \in \WW^p(\MM^p)$. A measure $\mu^* \in \MM^p$ is a \emph{$p$-Fr\'echet mean} of $\P$ if it minimizes $\EE : \mu \in \MM^p \mapsto \int_{\nu \in \MM^p} \Dp^p (\mu,\nu) \dd \P(\nu)$.
\end{definition}

\subsection{$p$-Fr\'echet means in the finite case}
\label{subsec:bary_finite}
Let $\P$ be of the form $\sum_{i=1}^N \lambda_i \delta_{\mu_i}$ with $N \in \N$, $\mu_i$ a persistence measure of finite mass $m_i$, and $(\lambda_i)_i$ non-negative weights that sum to $1$. Define $\mtot \defeq \sum_{i=1}^N m_i$. To prove the existence of $p$-Fr\'echet means for such a $\P$, we show that, in this case, $p$-Fr\'echet means correspond to $p$-Fr\'echet means for the Wasserstein distance of some distribution on $\MM^p_{\mtot}(\tilde{\Omega})$, the sets of measures on $\tilde{\Omega}$ that all have the same mass $\mtot$ (see Section \ref{subsec:uniformly_bounded_mass}), a problem well studied in the literature \cite{ot:agueh2011barycenters,ot:carlier2010matchingForTeam,ot:carlier2015numerical}.

We start with a lemma which affirms that if a measure $\mu$ has too much mass (larger than $\mtot$), then it cannot be a $p$-Fr\'echet mean of $\mu_1 \dots \mu_N$.

\begin{lemma}\label{lem:barycenter_small_mass} We have $\inf\{\EE(\mu),\ \mu\in \MM^p\} = \inf\{\EE(\mu), \mu\in \MM^p_{\leq \mtot}\}$.
\end{lemma}

\begin{proof}
The idea of the proof is to show that if a measure $\mu$ has some mass that is mapped to the diagonal in each transport plan between $\mu$ and $\mu_i$, then we can build a measure $\mu'$ by ``removing'' this mass, and then observe that such a measure $\mu'$ has a smaller energy.

Let thus $\mu \in \MM^p$. Let $\pi_i \in \Opt_p(\mu_i,\mu)$ for $i=1,\dots,N$. The measure $A\subset \Omega \mapsto \pi_i(\thediag \times A)$ is absolutely continuous with respect to $\mu$. Therefore, it has a density $f_i$ with respect to $\mu$. Define for $A \subset \upperdiag$ a Borel set,
\[ \mu'(A) \defeq \mu(A) - \int_A \min_j f_j(x) \dd \mu(x),\]
and, for $i = 1,\dots,N$, a measure $\pi'_i$, equal to $\pi_i$ on $\Omega \times \groundspace$ and which satisfies for $A \subset \upperdiag$ a Borel set,
\[\pi'_i(\thediag \times A) = \pi'_i(s(A) \times A) \defeq \pi_i(\thediag \times A) - \int_A \min_j f_j(x) \dd \mu(x),\]
where $s$ is the orthogonal projection on $\thediag$. As $\pi_i(\thediag\times A) = \int_A f_i (x) \dd\mu(x)$, $\pi'_i(A)$ is nonnegative, and as $\pi_i(\thediag\times A) \leq \mu(A)$, it follows that $\mu'(A)$ is nonnegative. To prove that $\pi'_i \in \Adm(\mu_i,\mu')$, it is enough to check that for $A \subset \upperdiag$, $\pi'_i(\groundspace \times A)=\mu'(A)$:
\begin{align*}
\pi'_i(\groundspace \times A) &= \pi_i(\upperdiag \times A) + \pi_i(\thediag \times A) - \int_A \min_j f_j(x) \dd \mu(x) \\
&= \mu(A) - \int_A \min_j f_j(x) \dd \mu(x) = \mu'(A).
\end{align*} 
Also, 
\begin{align*}
\mu'(\Omega) &= \int_\Omega(1-\min_j f_j) \dd\mu(x) \leq  \sum_{j=1}^N \int_\Omega (1-f_j) \dd\mu(x) \\
&= \sum_{j=1}^N (\mu(\Omega)-\pi_j(\thediag \times \Omega)) = \sum_{j=1}^N (\pi_j(\groundspace\times \Omega)-\pi_j(\thediag \times \Omega)) \\
&=  \sum_{j=1}^N \pi_j(\Omega \times \Omega) \leq \sum_{j=1}^N \pi_j(\Omega \times \groundspace) = \sum_{j=1}^N m_j = \mtot.
\end{align*}
and thus $\mu'(\Omega) \leq \mtot$. To conclude, observe that
\begin{align*}
\EE(\mu') &\leq \sum_{i=1}^N \lambda_i C_p(\pi'_i) = \sum_{i=1}^N \lambda_i \left(\iint_{\Omega \times \groundspace} d(x,y)^p \dd \pi_i(x,y)\right.\\
&\qquad \left. + \iint_{\thediag \times \Omega} d(x,y)^p \dd \pi_i(x,y) - \int_\Omega d(x,\thediag)^p \min_j f_j(x) \dd\mu(x) \right) \\
&\leq \sum_{i=1}^N \lambda_i C_p(\pi) = \EE(\mu).
\end{align*} 

\end{proof}

Recall that $W_{p,\rho}$ denotes the Wasserstein distance between measures with same mass supported on the metric space $(\tilde{\upperdiag}, \rho)$ (see Sections \ref{subsec:wasserstein} and  \ref{subsec:uniformly_bounded_mass}).
\begin{proposition}\label{prop:equiv_minimum_functionals}
Let $\Psi: \mu  \in \MM^p_{\leq \mtot}   \mapsto \tilde{\mu} \in \MM_{\mtot}^p(\tilde{\upperdiag}) $, where $\tilde{\mu} \defeq \mu + (\mtot - \mu(\upperdiag)) \delta_{\thediag}$. The functionals
\begin{align*}
	&\EE : \mu \in  \MM^p_{\leq \mtot} \mapsto \sum_{i=1}^N \lambda_i \Dp^p(\mu , \mu_i) \text{ and }  \\
	&\FF : \tilde{\mu} \in  \MM_{\mtot}^p(\tilde{\Omega}) \mapsto \sum_{i=1}^N \lambda_i W^p_{p,\rho}(\tilde{\mu} , \Psi (\mu_i) ),
\end{align*}
have the same infimum values and $\argmin \EE = \Psi^{-1} (\argmin  \FF)$.
\end{proposition}

\begin{proof}
Let $G$ be the set of $\mu\in \MM^p$ such that, for all $i$, there exists $\pi_i\in  \Opt_p(\mu_i, \mu)$ with $\pi_i(\Omega,\thediag) = 0$. By point \ref{it:iota} of Lemma \ref{lemma:kappa_iota}, for $\mu\in  G$  and  $\pi_i  \in \Opt_p(\mu_i,\mu)$ with $\pi_i(\Omega,\thediag) = 0$, $\iota(\pi_i)$ is well defined and satisfies 
\[ \Dp^p(\mu_i,\mu) = C_p(\pi_i) = \iint_{\tilde{\upperdiag} \times \tilde{\upperdiag}} \tilde d(x,y)^p \dd \iota(\pi_i)(x,y) \geq \tilde{C}_p(\iota(\pi_i)) \geq W_{p,\rho}^p(\tilde \mu_i,\tilde{\mu}), \]
so that $\FF(\Psi(\mu)) \leq \EE(\mu)$. As, by Lemma \ref{lem:Dp_smaller_than_W_p}, $\EE \leq  \FF \circ \Psi$, we therefore have $\EE(\mu) = \FF(\Psi(\mu))$ for $\mu \in G$.

We now show that if $\mu \notin G$, then there exists $\mu' \in \MM^p$ with $\EE(\mu')<\EE(\mu)$. Let $\mu \notin G$ and $\pi_i \in  \Opt_p(\mu_i,\mu)$. Assume that for some $i$, we have $\pi_i(\upperdiag, \thediag) > 0$, and introduce $\nu \in \MM^p$ defined as $\nu(A) = \pi_i(A,\thediag)$ for $A \subset \upperdiag$. 
Define
\[
	T : \upperdiag \ni x \mapsto \argmin_{y \in \upperdiag} \left\{ \lambda_i d(x,y)^p + \sum_{j \neq i} \lambda_j d(y,\thediag)^p \right\} \in \upperdiag.
\]
Note that since $p > 1$, this map is well defined (the minimizer is unique due to strict convexity) and continuous thus measurable. Consider the measure $\mu' = \mu + (T_\# \nu)$, where $T_\# \nu$ is the push-forward of $\nu$ by the application $T$. Consider the transport plan $\pi_i'$ deduced from $\pi_i$ where $\nu$ is transported onto $T_\# \nu$ instead of being transported to $\thediag$ (see Figure \ref{fig:illu_transport_barycenter}). More precisely, $\pi_i'$ is the measure on $\groundspace \times \groundspace$ defined by, for Borel sets $A, B \subset \upperdiag$:
\begin{align*}
	&\pi_i'(A \times B) = \pi_i(A \times B) + \nu( A \cap T^{-1}(B)), \\
	&\pi_i'(A\times \thediag) = 0, \quad \pi_i'(\thediag\times B) = \pi_i(\thediag\times B).
\end{align*}
\begin{figure}
	\center
	\includegraphics[width=0.5\textwidth]{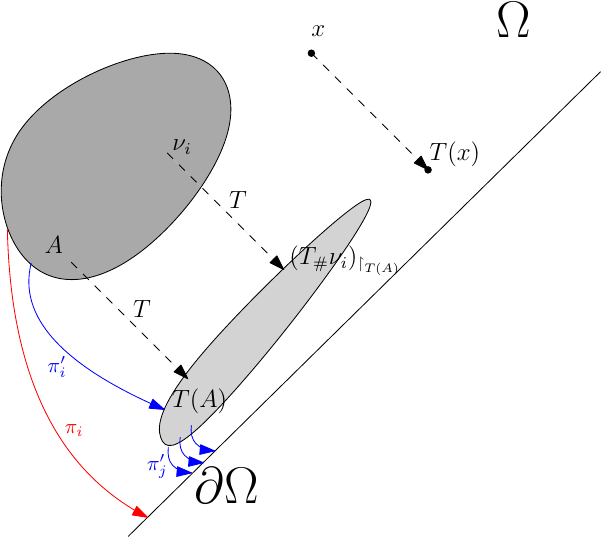}
	\caption{Global picture of the proof. The main idea is to observe that the cost induced by $\pi_i$ (red) is strictly greater than the sum of costs induces by the $\pi_i'$s (blue), which leads to a strictly better energy.}
	\label{fig:illu_transport_barycenter}
\end{figure}
We have $\pi_i' \in \Adm(\mu_i,\mu')$. Indeed, for Borel sets $A, B \subset \upperdiag$:
\[
	\pi_i'(A\times \groundspace) = \pi_i'(A\times \upperdiag) = \pi_i(A\times \upperdiag) + \nu(A) = \pi_i(A\times \groundspace) = \mu_i(A),
\]
and
\begin{align*}
\pi_i'(\groundspace\times B) &= \pi_i'(\upperdiag\times B) + \pi_i'(\thediag\times B) \\
&= \pi_i(\upperdiag\times B) + \nu(T^{-1}(B)) + \pi_i(\thediag\times B)\\
& = \mu(B) + T_\# \nu(B) = \mu'(B).
\end{align*}
	
Using $\pi_i'$ instead of $\pi_i$ changes the transport cost by the quantity \[\int_\upperdiag [d(x, T(x))^p - d(x, \thediag)^p] \dd \nu(x)<0.\]

In a similar way, we define for $j \neq i$ the plan $\pi_j' \in \Adm(\mu_j, \mu')$ by transporting the mass induced by the newly added $(T_\# \nu)$ to the diagonal $\thediag$. Using these modified transport plans increases the total cost by 
\[ \sum_{j\neq i}\lambda_j \int_\upperdiag d(T(x), \thediag)^p \dd \nu(x). \]

One can observe that
\[ \int_\upperdiag \left[\lambda_i\left(d(x,T(x))^p - d(x, \thediag)^p \right) + \sum_{j\neq i}\lambda_j d(T(x), \thediag)^p  \right] \dd \nu(x) < 0 \]
due to the definition of $T$ and $\nu(\upperdiag) > 0$.

Therefore, the total transport cost induced by the $(\pi_i')_{i=1 \dots N}$ is strictly less than $\EE(\mu)$, and thus $\EE(\mu') < \EE(\mu)$.
Finally, we have
\[ \inf_{\mu\in \MM^p_{\leq  m_{tot}}} \hspace{-0.3cm} \EE(\mu) = \inf_{\mu\in  G} \EE(\mu)  = \inf_{\mu\in G}\FF(\Psi(\mu))\geq  \inf_{\mu\in \MM^p_{\leq  m_{tot}}}\hspace{-0.3cm}\FF(\Psi(\mu)) \geq \inf_{\mu\in \MM^p_{\leq  m_{tot}}} \hspace{-0.3cm} \EE(\mu),\]
where the last inequality comes from $\FF \circ \Psi \geq \EE$ (Lemma \ref{lem:Dp_smaller_than_W_p}). Therefore,  $\inf  \EE = \inf \FF \circ \Psi$, which is equal to $\inf \FF$, as $\Psi$ is a bijection. Also, if $\mu$ is a minimizer of $\EE$ (should it exist), then $\mu \in G$ and $\EE(\mu) = \FF(\Psi(\mu))$. Therefore, as the infimum are equal, $\Psi(\mu)$ is a minimizer of $\FF$. Reciprocally, if $\tilde{\mu}$ is a minimizer of $\FF$, then, by Lemma \ref{lem:Dp_smaller_than_W_p}, $\FF(\tilde{\mu}) \geq \EE(\Psi^{-1}(\tilde{\mu}))$, and, as the infimum are equal, $\Psi^{-1}(\tilde{\mu})$ is a minimizer of $\EE$. 
\end{proof}

The existence of minimizers $\tilde{\mu}$ of $\FF$, that is ``Wasserstein barycenter'' (i.e.~$p$-Fr\'echet means for the Wasserstein distance) of $\tilde{\P} \defeq \sum_{i=1}^N \lambda_i \delta_{\tilde{\mu_i}}$, is well-known (see \cite[Theorem 8]{ot:agueh2011barycenters}). Proposition \ref{prop:equiv_minimum_functionals} asserts that $\Psi^{-1}(\tilde{\mu})$ is a minimizer of $\EE$ on $\MM^p_{\leq \mtot}$, and thus a $p$-Fr\'echet mean of $\P$ according to Lemma \ref{lem:barycenter_small_mass}. We therefore have proved the existence of $p$-Fr\'echet means in the finite case.

\subsection{Existence and consistency of $p$-Fr\'echet means}
\label{subsec:bary_in_M}

We now extend the results of the previous section to the $p$-Fr\'echet means of general probability measures supported on $\MM^p$. First, we show a \emph{consistency} result, in the vein of \cite[Theorem 3]{ot:leGouic2016existenceConsistencyWB}.

\begin{proposition} Let $\P_n, \P$ be probability measures in $\WW^p(\MM^p)$. Assume that each $\P_n$ has a $p$-Fr\'echet mean $\mu_n$ and that $W_{p,\Dp}(\P_n,\P) \to 0$. Then, the sequence $(\mu_n)_n$ is relatively compact in $(\MM^p,\Dp)$, and any limit of a converging subsequence is a $p$-Fr\'echet mean of $\P$.
\label{prop:consistencyBary}
\end{proposition}


\begin{proof}
In order to prove relative compactness of $(\mu_n)_n$, we use the characterization stated in Proposition \ref{prop:rel_cpt_vague}. Consider a compact set $K \subset \upperdiag$. We have, because of \eqref{eq:pers_compact},
\begin{align*}
	\mu_n(K)^\frac{1}{p} &\leq \frac{1}{d(K,\thediag)} \Dp(\mu_n, \emptydgm) = \frac{1}{d(K,\thediag)} W_{p,\Dp}(\delta_{\mu_n} , \delta_{\emptydgm} ) \\
	&\leq \frac{1}{d(K,\thediag)} \left( W_{p,\Dp}(\delta_{\mu_n}, \P_n) + W_{p,\Dp}(\P_n, \delta_{\emptydgm}) \right)
\end{align*}
Since $\mu_n$ is a $p$-Fr\'echet mean of $\P_n$, it minimizes $\{ W_{p,\Dp}(\delta_\nu, \P_n),\ \nu \in \MM^p\}$, and in particular $W_{p,\Dp}(\delta_{\mu_n}, \P_n) \leq W_{p,\Dp}(\delta_{\emptydgm} , \P_n)$. Furthermore, as by assumption $W_{p,\Dp}(\P_n, \P) \rightarrow 0$, we have that $\sup_n W_{p,\Dp}(\P_n,\ \delta_{\emptydgm}) < \infty$. As a consequence $\sup_n \mu_n(K) < \infty$, and Proposition \ref{prop:rel_cpt_vague} allows us to conclude that the sequence $(\mu_n)_n$ is relatively compact for the vague convergence.

To conclude the proof, we use the following two lemmas, whose proofs are found in Appendix \ref{sec:proof_barycenters}. 
\begin{lemma}\label{lem:rewrite} There exists a subsequence $(\mu_{n_k})_k$ of $(\mu_n)_n$ which vaguely converges towards $\mu$ a $p$-Fr\'echet mean of $\P$ and there exists $\nu \in \MM^p$ such that $\Dp(\mu_{n_k}, \nu) \rightarrow \Dp(\mu, \nu)$ as $k\to \infty$.
\end{lemma}

\begin{lemma}\label{lemma:Skorokhod}
Let $\mu, \mu_1, \mu_2, \dots \in \MM^p$. Then, $\Dp(\mu_n,\mu) \to 0$ if and only if (i) $\mu_n \cvvague \mu$ and (ii) there exists a persistence measure $\nu \in \MM^p$ such that $\Dp(\mu_n, \nu) \to \Dp(\mu, \nu)$.
\end{lemma}

Let $\mu'_k=\mu_{n_k}$ be any subsequence of $\mu_n$. We want to show that there exists a subsequence of $\mu'_k$ which converges with respect to the $\Dp$ metric towards some $p$-Fr\'echet mean of $P$. By Lemma \ref{lem:rewrite}  applied to the sequence $(\mu'_k)_k$, there exists a subsequence $\mu'_{k_l}$ which converges vaguely to some $p$-Fr\'echet mean $\mu$ of $P$, and some $\nu$ with $\Dp(\mu'_{k_l}, \nu) \rightarrow \Dp(\mu, \nu)$ as $l\to \infty$. By Lemma \ref{lemma:Skorokhod}, this implies that $\mu'_{k_l}$ converges to $\mu$ with respect to the $\Dp$ metric, showing the conclusion.

\end{proof}

As the finite case is solved, generalization follows easily using Proposition \ref{prop:consistencyBary}.
\begin{theorem}For any probability distribution $\P$ supported on $\MM^p$ with finite $p$-th moment, the set of $p$-Fr\'echet means of $\P$ is a non-empty compact convex set of $\MM^p$.
\label{thm:existenceBary}
\end{theorem}

\begin{proof}
We first prove the non-emptiness. Let $\P = \sum_{i=1}^N \lambda_i \mu_i$ be a probability measure on $\MM^p$ with finite support $\mu_1, \dots ,\mu_N$. According to Proposition \ref{prop:finite-to-global}, there exists sequences $(\mu_i^{(n)})_n$ in $\MM^p_f$ with $\Dp(\mu_i^{(n)},\mu_i) \to 0$. As a consequence of the result of Section \ref{subsec:bary_finite}, the probability measures $\P^{(n)} \defeq \sum_i \lambda_i \delta_{\mu_i^{(n)}}$ admit $p$-Fr\'echet means. Furthermore, $W_{p,\Dp}^p(\P^{(n)} , \P) \leq \sum_i \lambda_i \Dp^p(\mu_i^{(n)} ,\mu_i)$ so that this quantity converges to $0$ as $n \rightarrow \infty$. It follows from Proposition \ref{prop:consistencyBary} that $\P$ admits a $p$-Fr\'echet mean.

If $\P$ has infinite support, following \cite{ot:leGouic2016existenceConsistencyWB}, it can be approximated (in $W_{p,\Dp}$) by a empirical probability measure $\P_n = \frac{1}{n} \sum_{i=1}^n \delta_{\mu_i}$ where the $\mu_i$ are i.i.d.\ from $\P$. We know that $\P_n$ admits a $p$-Fr\'echet mean since its support is finite, and thus, applying Proposition \ref{prop:consistencyBary} once again, we obtain that $\P$ admits a $p$-Fr\'echet mean. 

Finally, the compactness of the set of $p$-Fr\'echet means follows from Proposition \ref{prop:consistencyBary} applied with $P_n=P$: if $(\mu_n)_n$ is a sequence of $p$-Fr\'echet means, then the sequence is relatively compact in $(\MM^p,\Dp)$, and any converging subsequence is also a $p$-Fr\'echet mean of $P$. Also, the convexity of the set of $p$-Fr\'echet means follows from the convexity of $\Dp^p$ (see Lemma \ref{lem:convex} in Appendix \ref{sec:details_expectations} below): if $\mu_1$, $\mu_2$ are two  $p$-Fr\'echet means with energy $\EE(\mu_1)=\EE(\mu_2)=E_0$ and $0\leq \lambda\leq 1$, then
\begin{align*}
\EE(\lambda\mu_1+(1-\lambda)\mu_2)&= \int_{\nu \in \MM^p} \Dp^p (\lambda\mu_1+(1-\lambda)\mu_2,\nu) \dd \P(\nu) \\
&\leq \int_{\nu \in \MM^p} (\lambda\Dp^p (\mu_1,\nu)+(1-\lambda)\Dp^p(\mu_2,\nu)) \dd \P(\nu) \\
&= \lambda \EE(\mu_1)+(1-\lambda)\EE(\mu_2)= E_0,
\end{align*} 
so that $\lambda\mu_1+(1-\lambda)\mu_2$ is also a $p$-Fr\'echet mean.

\end{proof}

\subsection{$p$-Fr\'echet means in $\PD^p$}
\label{subsec:bary_in_D}
We now prove the existence of $p$-Fr\'echet means for distributions of persistence diagram (i.e.~probability distributions supported on $\PD^p$), extending the results of \cite{tda:mukherjee2011probabilitymeasure}, in which authors prove their existence for specific probability distributions (namely distributions with compact support or specific rates of decay). Theorem \ref{thm:barycenterInD} below asserts two different things: that $\argmin\{\EE(a),\ a \in \PD^p\}$ is non empty, and that $\min\{\EE(a),\ a\in \PD^p\} = \min\{\EE(\mu),\ \mu \in \MM^p\}$, i.e~a persistence measure cannot perform strictly better than an optimal persistence diagram when averaging diagrams. As for $p$-Fr\'echet means in $\MM^p$, we start with the finite case. The following lemma actually gives a geometric description of the set of $p$-Fr\'echet means obtained when averaging a finite number of finite diagrams.
\begin{lemma}\label{lemma:tum}
Consider $a_1, \dots, a_N \in \PD_f$, weights $(\lambda_i)_i$ that sum to $1$, and let $\P \defeq \sum_{i=1}^N \lambda_i \delta_{a_i}$. Then, the set of minimizers of $\mu \mapsto \sum_{i=1}^N \lambda_i \Dp^p(\mu, a_i)$ is a non empty convex subset of $\MM_f^p$ whose extreme points belong to $\PD_f$. In particular, $\P$ admits a $p$-Fr\'echet mean in $\PD_f$.
\end{lemma}
The proof of this lemma is delayed to Appendix \ref{sec:proof_barycenters}. Note that, as a straightforward consequence, if $\P$ has a unique minimizer in $\PD_f$ (which is generically true \cite{tda:turner2013meansAndMedians}), then so it does in $\MM_f^p$.

\begin{theorem}
For any probability distribution $\P$ supported on $\PD^p$ with finite $p$-th moment, the set of $p$-Fr\'echet means of $\P$ contains an element of $\DD^p$. Furthermore, if $P$ is supported on a finite set of finite persistence diagrams, then the set of the $p$-Fr\'echet means of $P$ is a convex set whose extreme points are in $\PD^p$.\label{thm:barycenterInD}
\end{theorem}
\begin{proof} The second assertion of the theorem is stated in Lemma \ref{lemma:tum}. To prove the existence of a $p$-Fr\'echet mean which is a persistence diagram, we argue as in the proof of Theorem \ref{thm:existenceBary}, using additionally the fact that $\PD^p$ is closed in $\MM^p$ (Proposition \ref{prop:diagram_close}). 
\end{proof}

\section{Applications}\label{sec:applications}

\subsection{Characterization of continuous linear representations} \label{subsec:continuity_of_representations}
As mentioned in the introduction, a \emph{linear representation} of persistence measures (in particular persistence diagrams) is a mapping $\Phi : \MM^p \to \BB$ for some Banach space $\BB$ of the form $\mu \mapsto \mu(f)$, where $f : \upperdiag \to \BB$ is some chosen function. Doing so, one can turn a sample of diagrams (or measures) into a sample of vectors, making the use of machine learning tools easier. Of course, a minimal expectation is that $\Phi$ should be continuous. In practice, building a linear representations (see below for a list of examples) generally follows the same pattern: first consider a ``nice'' function $g$, e.g. a gaussian distribution, then introduce a weight with respect to the distance to the diagonal $d( \cdot , \thediag)^p$, and	 prove that $\mu \mapsto \mu(g(\cdot) d(\cdot ,\thediag)^p)$ has some regularity properties (continuity, stability, etc.). Applying Theorem \ref{thm:conv_dp}, we show that this approach always gives a continuous linear representation, and that it is the only way to do so.

For $\BB$ a Banach space (typically $\R^d$), define the class of functions:
\begin{equation}
\Cp = \left\{ f : \upperdiag \to \BB,\ f \text{ continuous and } x \mapsto \frac{f(x)}{d(x,\thediag)^p} \text{ bounded } \right\}
\end{equation}
\begin{proposition}\label{cor:feature_maps_continuous}
Let $\BB$ be a Banach space and $f: \upperdiag \to \BB$ a function. The linear representation $\Phi : \MM^p \to \BB$ defined by $\Phi : \mu \mapsto \mu(f) = \int_\upperdiag f(x) \dd \mu(x)$ is continuous with respect to $\Dp$ if and only if $f \in \Cp$.
\end{proposition}

\begin{proof}
Consider first the case $\BB = \R$. Let $f \in \Cp$ and $\mu, \mu_1, \mu_2 \dots \in \MM^p$ be such that $\Dp(\mu_n, \mu) \to 0$. Recall the definition \eqref{eq:def_mu^p} of $\mu^{(p)}$. Using Corollary \ref{cor:Dp_to_weak_cv}, having $\Dp(\mu_n, \mu) \to 0$ means that $\mu^{(p)}_n \cvweak \mu^{(p)}$, and thus that
\[
\int_\upperdiag \frac{f(x)}{d(x,\thediag)^p} \dd \mu^{(p)}_n(x) \to  \int_\upperdiag \frac{f(x)}{d(x,\thediag)^p} \dd \mu^{(p)}(x),
\]
that is
\[ \Phi(\mu_n) = \int_\upperdiag f(x) \dd \mu_n(x) \to \int_\upperdiag f(x) \dd \mu(x) = \Phi(\mu), \]
i.e.~$\Phi$ is continuous with respect to $\Dp$.

Now, let $\BB$ be any Banach space. \cite[Theorem 2]{nielsen2011weak} states that if a sequence of measures $(\mu_n)_n$ weakly converges to $\mu$, then $\mu_n(f) \to \mu(f)$ for any continuous bounded function $g : \upperdiag \to \BB$. Applying this result to the sequence $(\mu^{(p)}_n$ with $g = f / d(\cdot, \thediag)^p$ yields the desired result.
\medbreak
Conversely, let $f: \upperdiag \to \BB$. Assume first that $f$ is not continuous in some $x \in \upperdiag$. There exist a sequence $(x_n)_n \in \upperdiag^\N$ such that $x_n \to x$ but $f(x_n) \nrightarrow f(x)$. Let $\mu_n = \delta_{x_n}$ and $\mu = \delta_{x}$. We have $\Dp(\mu_n,\mu) \to 0$, but $\mu_n(f) = f(x_n) \nrightarrow f(x_0) = \mu(f)$, so that the linear representation $\mu \mapsto \mu(f)$ cannot be continuous. \\
Then, assume that $f$ is continuous but that $x \mapsto \frac{f(x)}{d(x,\thediag)^p}$ is not bounded. Let thus $(x_n)_n \in \upperdiag^\N$ be a sequence such that $\left\| \frac{f(x_n)}{d(x_n,\thediag)^p} \right\| \to +\infty$. Define the measure $\mu_n \defeq \frac{1}{\|f(x_n)\|} \delta_{x_n}$. Observe that $\Dp(\mu_n, 0) = \frac{d(x_n,\thediag)^p}{\|f(x_n)\|} \to 0$ by hypothesis. However, $\| \mu_n(f) \| = 1$ for all $n$, allowing us to conclude once again that $\mu \mapsto \mu(f)$ cannot be continuous.

\end{proof}

Let us give some examples of such linear representations (which are thus continuous) commonly used in applications of TDA. Note that the following definitions do not rely on the fact that the input must be a persistence diagram and actually make sense for any persistence measure in $\MM^p$. See Figure \ref{fig:representations} for an illustration, in which computations are done with $p=1$ (and $p'=1$ for the weighted Betti curve).
\begin{itemize}
	\item \emph{Persistence surface and its variations.} Let $K:\R^2\to \R$  be a nonnegative Lipschitz continuous bounded function (e.g.~$K(x,y) = \exp\left(-\frac{\|x-y\|^2}{2}\right)$) and define $f:x  \in \upperdiag \mapsto d(x,\thediag)^p \times K(x,\cdot)$, so that $f(x) : \R^2 \to \R$ is a real-valued function. The corresponding representation $\Phi$ takes its values in $(C_b(\R^2),\|\cdot\|_\infty)$, the (Banach) space of continuous bounded functions. This representation is called the persistence surface and has been introduced with slight variations in different works \cite{tda:adams2017persistenceImages,tda:chen2015statistical,tda:kusano2016PWG,tda:reininghaus2015stable}.
	\item \emph{Persistence silhouettes.} Let $\Lambda(x,t) = \max\left(\frac{x_2-x_1}{2} - \left| t-\frac{x_2+x_1}{2}\right|,0\right)$ for $t\in \R$ and $x\in \upperdiag$. Then, defining $f:x\in \upperdiag \mapsto d(x,\thediag)^{p-1}\times \Lambda(x,\cdot)$, one has that $\|f(x)\|_\infty$ is proportional to $d(x,\thediag)^p$, so that the corresponding representation is continuous for $\Dp$. This representation is called the persistence silhouette, and was introduced in \cite{tda:chazal2014stochastic}. In particular, it consists in a weighted sum of the different functions of the persistence landscape \cite{tda:bubenik2017persistence}. The corresponding Banach space is $(C_b(\R),\|\cdot\|_\infty)$.
	\item \emph{Weighted Betti curves.} For $t\in\R$, define $B_t$ the rectangle $(-\infty ,t]\times [t,+\infty)$. Let $p, p'\geq 1$, and define $f:x\in \upperdiag \mapsto (t \mapsto d(x,\thediag)^{p-1/p'} \ones\{ x\in B_t \})$. Then $f(x)\in L_{p'}(\R)$ with $\|f(x)\|_{p'}$ proportional to $d(x,\thediag)^p$. The corresponding function $\Phi$ is the weighted Betti curve, which takes its values in the Banach space $(L_{p'}(\R),\|\cdot\|_{p'})$. In particular, one obtains the continuity of the classical Betti curves from $(\MM^1,\mathrm{OT}_1)$ to $L_1(\R)$.
\end{itemize}

\paragraph{Stability in the case $p=1$.} Continuity is a basic expectation when embedding a set of diagrams (or measures) in some Banach space $\BB$. One could however ask for more, e.g.~some Lipschitz regularity: given a representation $\Phi : \MM^p \to \BB$, one may want to have $\|\Phi(\mu) - \Phi(\nu)\| \leq C \cdot \Dp(\mu,\nu)$ for some constant $C$. This property is generally referred to as ``stability'' in the TDA community and is generally obtained with $p=1$, see for example \cite[Theorem 5]{tda:adams2017persistenceImages}, \cite[Theorem 3.3 \& 3.4]{tda:carriere2017sliced}, \cite[\S 4]{tda:som2018perturbation}, \cite[Theorem 2]{tda:reininghaus2015stable}, etc. 

Here, we still consider the case of linear representations, and show that stability always holds with respect to the distance $\mathrm{OT}_1$. Informally, this is explained by the fact that when $p=1$, the cost function $(x,y) \mapsto d(x,y)^p$ is actually a distance. 
\begin{proposition}
	Define  $\mathcal{L}$ the set of Lipschitz continuous functions $f : \groundspace \to \R$ with Lipschitz constant less than $1$ and that satisfy $f(\thediag) = 0$.	Let $T$ be any set, and consider a family $(f_t)_{t \in T}$ with $f_t \in \mathcal{L}$. Then the linear representation $\Phi : \mu \mapsto (\mu(f_t))_{t \in T}$ is $1$-Lipschitz continuous in the following sense:
	\begin{equation}
		\|\Phi(\mu) - \Phi(\nu)\|_\infty \defeq \sup_{t \in T} |(\mu - \nu)(f_t)| \leq \mathrm{OT}_1(\mu,\nu),
	\end{equation}
	for any measures $\mu, \nu \in \MM^1$.
\end{proposition}

\begin{proof}
Consider $\mu, \nu \in \MM^1$, and $\pi \in \Opt(\mu,\nu)$ an optimal transport plan. Let $t\in T$. We have:
\begin{align*} (\mu - \nu)(f_t) &= \int_\upperdiag f_t(x) \dd \mu(x) - \int_\upperdiag f_t(y) \dd \nu(y) = \iint_{\groundspace \times \groundspace} (f_t(x) - f_t(y)) \dd \pi(x,y) \\
&\leq \iint_{\groundspace \times\groundspace} d(x,y) \dd \pi(x,y) = \mathrm{OT}_1(\mu,\nu),
\end{align*}
and thus, $\|\Phi(\mu) - \Phi(\nu)\|_\infty\leq \mathrm{OT}_1(\mu,\nu)$. 
\end{proof}

In particular, if $f:\groundspace \to \BB$, where $\BB$ is some Banach space, is $1$-Lipschitz with $f(\thediag)=0$, then one can let $T=\BB_1^*$ (the unit ball of the dual of $\BB$) and $f_t(x) \defeq t(f(x))$ for $t\in T$. If $\Phi(\mu)=\mu(f)$, we then obtain that $\|\Phi(\mu)-\Phi(\nu)\|\leq \mathrm{OT}_1(\mu,\nu)$, i.e. that $\Phi:(\MM^1,\mathrm{OT}_1)\to (\BB,\|\cdot\|)$ is $1$-Lipschitz.

\begin{remark}
One actually has a converse of such an inequality, i.e.~it can be shown that 
	\begin{equation}\label{eq:KantoRubin_duality}
		\mathrm{OT}_1(\mu,\nu) = \max_{f \in \mathcal{L}} (\mu-\nu)(f),
	\end{equation}
	This equation is an adapted version of the very well-known Kantorovich-Rubinstein formula, which is itself a particular version in the case $p=1$ of the duality formula in optimal transport, see for example \cite[Theorem 5.10]{ot:villani2008optimal} and \cite[Theorem 1.39]{otam}. A proof of Eq. \eqref{eq:KantoRubin_duality} would require to introduce several optimal transport notions. The interested reader can consult Proposition 2.3 in \cite{ot:figalli2010newTransportationDistance} for details.
\end{remark}

\subsection{Convergence of random persistence measures in the thermodynamic regime}\label{subsec:thermo}
Geometric probability is the study of geometric quantities arising naturally from point processes in $\R^d$. Recently, several works \cite{tda:bobrowski2017maximally,tda:divolpolonik,tda:hiraoka2016limitTheoremForPD,tda:goel2018asymptotic,tda:schweinhart2018weighted} used techniques originating from this field to understand the persistent homology of such point processes. Let $\X_n \defeq \{X_1,\dots,X_n\}$ be a $n$-sample of a distribution having some density on the cube $[0,1]^d$, bounded from below and above by positive constants. Extending the work of Hiraoka, Shirai and Trinh \cite{tda:hiraoka2016limitTheoremForPD}, Divol and Polonik \cite{tda:divolpolonik} show laws of large numbers for the persistence diagrams $\dgm(\X_n)$ of $\X_n$, built with either the \v Cech or Rips filtration. More precisely, \cite[Theorem 5]{tda:divolpolonik} states that there exists a Radon measure $\mu$ on $\upperdiag$ such that, almost surely, the sequence of measures $\mu_n \defeq n^{-1}\dgm(n^{1/d}\X_n)$ converges vaguely to $\mu$ and \cite[Theorem 6]{tda:divolpolonik} implies that $\Pers_p(\mu_n)$ converges to $\Pers_p(\mu)<\infty$ for all $p\geq 1$. Those two facts (vague convergence and convergence of the total persistence), along with Theorem \ref{thm:conv_dp}, gives the following result:
\begin{proposition}\label{prop:application_randomPD}
$\Dp(\mu_n,\mu) \xrightarrow[n\to \infty]{} 0 \text{ almost surely}.$
\end{proposition}

\paragraph{Numerical illustration of the convergence in the one-dimensional case.}
In dimension $d \geq 2$, there is no known closed-form expression for the limit $\mu$. However, in the case $d=1$, authors in \cite[Remark 2 (b)]{tda:divolpolonik} show that if the $X_i$ are $n$ i.i.d~realizations of a random variable $X$ admitting a density $\kappa$ supported on $[0,1]$, bounded from below and above by positive constants, then (the ordinate of) $\mu$---which is supported on $\{0\} \times (0, +\infty)$ as we consider the Rips filtration in homology dimension $0$---admits $\varphi(u) \defeq \E_{X \sim \kappa} [\exp(- u \kappa(X))]$ as density. In particular, if $X$ is uniform, then (the ordinate of) $\mu$ admits $u \mapsto e^{-u}$ as density.

It allows us to realize a simple numerical experiment: we sample $n$ points $(X_1 \dots X_n)$ uniformly on $[0,1]$ and then compute the corresponding persistence diagram using the Rips filtration, whose points are denoted by $(0, t_1), \dots,$ $(0, t_{n-1})$ (note that we removed the point $(0, +\infty)$). We can now introduce the one-dimensional measure $\mu_n = \frac{1}{n} \sum_{k=1}^{n-1} \delta_{t_k}$ and compute $\Dp(\mu_n, \mu)$ in closed form using Proposition \ref{prop:Dp_equals_Wp_finite} and the fact that computing the distance $W_p$ in dimension $d=1$ is particularly easy (see \cite[Chapter 2]{otam}). See Figure \ref{fig:illu_cv} for an illustration.

\begin{figure}
	\center
	\includegraphics[width=0.67\textwidth]{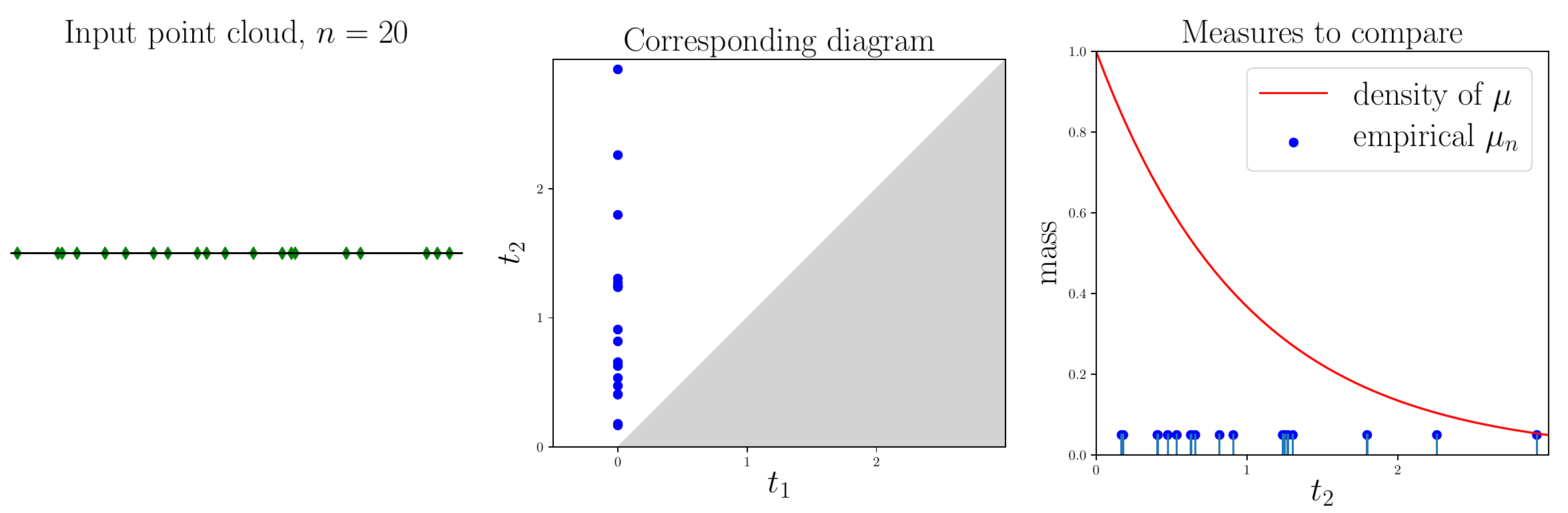}
	\includegraphics[width=0.32\textwidth]{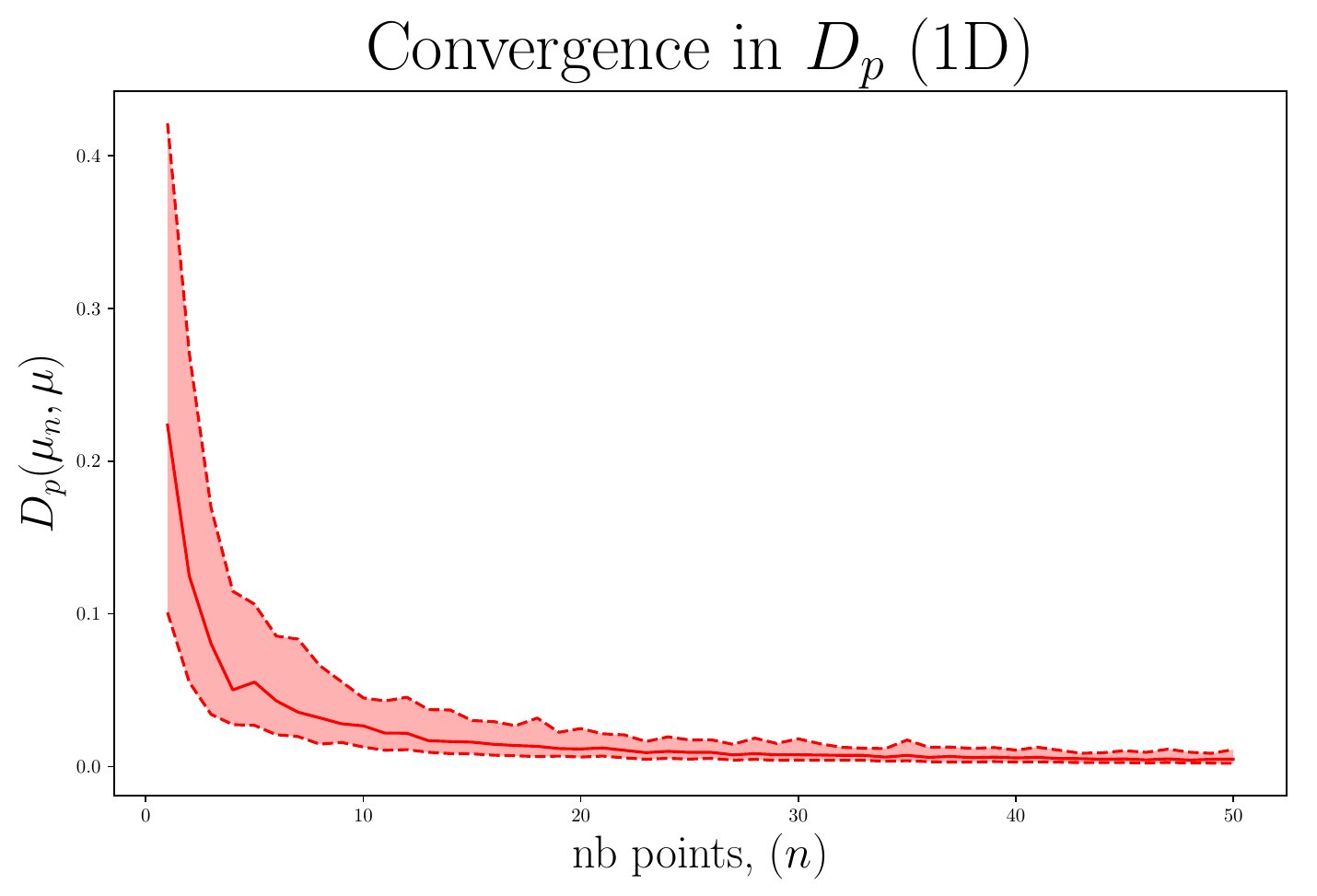}
	\caption{Numerical illustration of $D_p$ convergence in the one dimensional case. From left to right: an example of input point cloud with $n=20$ points ; the corresponding persistence diagram ; the (ordinate of) the two measures we compare: $\mu_n$ is the empirical one. Most right: median $\Dp(\mu_n, \mu)$ for $n = 2 \dots 50$, over $100$ runs along with 90\% and 10\% percentiles. Computations made with $p=2$.}
	\label{fig:illu_cv}
\end{figure}

\subsection{Stability of the expected persistence diagrams}\label{subsec:stability_expectedPD}
Given an i.i.d.~sample of $N$ persistence diagrams $a_1 \dots a_N$, a natural persistence measure to consider is their linear sample mean $\overline{a} \defeq \frac{1}{N} \sum_{1 \leq i \leq N} a_i$. More generally, given $\P \in \WW^p(\MM^p)$, and $\bm{\mu} \sim \P$, one may want to define $\E[\bm{\mu}]$ the linear expectation of $\P$ in the same vein. A well-suited definition of the linear expectation requires technical care (basically, turning the finite sum into a Bochner integral) and is detailed in Appendix \ref{sec:details_expectations}. It however satisfies the natural following characterization---that is sufficient to understand this section: 
\begin{equation}
\forall K \subset \upperdiag\ \text{compact},\ \E[\bm{\mu}](K) = \E[\bm{\mu}(K)].
\end{equation}
The behavior of such measures is studied in \cite{tda:divol2018density}, which shows that they have densities with respect to the Lebesgue measure in a wide variety of settings. A natural question is the stability of the linear expectations of random diagrams with respect to the underlying phenomenon generating them. The following proposition gives a positive answer to this problem, showing that given two close probability distributions $\P$ and $\P'$ supported on $\MM^p$, their linear expectations are close for the metric $\Dp$. 
\begin{proposition}\label{prop:Dp_jensen}
Let $\P,\P' \in \WW^p(\MM^p)$. Let $\bm \mu\sim P$ and $\bm\mu'\sim P'$. Then, 
 we have $\Dp(\E[\bm{\mu}],\E[\bm{\mu'}]) \leq W_{p,\Dp}(P,P')$.
\end{proposition}
The proof is postponed to Appendix \ref{sec:details_expectations}.

Using stability results on the distances $d_p$ ($= \Dp$) between persistence diagrams \cite{tda:cohen2010lipschitzStablePersistence}, one is able to obtain a more precise control between the expectations in some situations. For $\mathbb{Y}$ a sample in some metric space, denote by $\dgm(\mathbb{Y})$ the persistence diagram of $\mathbb{Y}$ built with the \v Cech filtration.
\begin{proposition}\label{prop:stability_expectedPD}
Let $\xi,\xi'$ be two probability measures on $\R^d$. Let $\X_n$ (resp. $\X'_n$) be a $n$-sample of law $\xi$ (resp. $\xi'$). Then, for any $k > d$, and any $p \geq k+1$,
\begin{equation}
\Dp^p(\E[\dgm(\X_n)],\E[\dgm(\X'_n)]) \leq C_{k,d} \cdot n \cdot  W_{p-k}^{p-k}(\xi,\xi')
\end{equation}
where $C_{k,d} \defeq C \diam(\X)^{k-d}\frac{k}{k-d}$ for some constant $C$ depending only on $\X$.

In particular, letting $p \rightarrow \infty$, we obtain a bottleneck stability result:
\begin{equation}
	\Dinf(\E[\dgm(\X_n)],\E[\dgm(\X'_n)]) \leq W_{\infty}(\xi,\xi').
\end{equation}
\end{proposition}
\begin{proof}
Let $P$ be the law of $\dgm(\X_n)$, $P'$ be the law of $\dgm(\X'_n)$ and let $\gamma$ be any coupling between $\X_n$ a $n$-sample of law $\xi$, and $\X'_n$ a $n$-sample  of law  $\xi'$. Then, the law of $(\dgm(\X_n),\dgm(\X'_n))$ is a coupling between $P$ and $P'$. Thus, Proposition \ref{prop:Dp_jensen} yields
\[ \Dp^p(\E[\dgm(\X_n)],\E[\dgm(\X'_n)]) \leq \E_{\gamma}[\Dp^p(\dgm(\X_n),\dgm(\X'_n))]. \]
 It is stated in \cite[Wasserstein Stability Theorem]{tda:cohen2010lipschitzStablePersistence} that \[\Dp^p(\dgm(\X_n),\dgm(\X'_n))  \leq  C_{k,d} H(\X_n,\X'_n)^{p-k},\] where $C_{k,d} \defeq C \diam(\X)^{k-d}\frac{k}{k-d}$ for some constant $C$ depending only on $\X$, and $H$ is the Hausdorff distance between sets. By taking the infimum on transport plans $\gamma$, we obtain 
 \[ \Dp^p(\E[\dgm(\X_n)],\E[\dgm(\X'_n)]) \leq C_{k,d} W_{p-k,H}^{p-k}(\xi^{\otimes n},(\xi')^{\otimes n}), \]
where $W_{H,p}$ is the $p$-Wasserstein distance between probability distributions on compact sets of the manifold $\X$, endowed with the Hausdorff distance. Lemma 15 of \cite{tda:chazal2015subsampling} states that \[W_{p-k,H}^{p-k}(\xi^{\otimes n},(\xi')^{\otimes n}) \leq n \cdot W_{p-k}^{p-k}(\xi,\xi'),\]
 concluding the proof. 
\end{proof}

Note that this proposition illustrates the usefulness of introducing new distances $\Dp$: considering the proximity between linear expectations requires to extend the metrics $d_p$ to Radon measures.

\section{Conclusion and further work.}
In this article, we introduce the space of persistence measures, a generalization of persistence diagrams which naturally appears in different applications (e.g.~when studying persistence diagrams coming from a random process). We provide an analysis of this space that also holds for the subspace of persistence diagrams. In particular, we observe that many notions used for the statistical analysis of persistence diagrams can be expressed naturally using this formalism based on optimal partial transport. We give characterizations of convergence of persistence diagrams and measures with respect to optimal transport metrics in terms of convergence for measures. We then prove existence and consistency of $p$-Fr\'echet means for any probability distribution of persistence diagrams and measures, extending previous work in the TDA community. We illustrate the interest of introducing the persistence measures space and its metrics in statistical applications of TDA: continuity of diagram representations, law of large number for persistence diagrams, stability of diagrams in a random settings.

We believe that closing the gap between optimal transport metrics and those of topological data analysis will help to develop new approaches in TDA and give a better understanding of the statistical behavior of topological descriptors. It allows one to address new problems, in particular those involving continuous counterpart of persistence diagrams, and invites one to consider various theoretical and computational tools developed in optimal transport theory that could be of major interest in topological data analysis: gradient flows \cite{ot:santambrogio2017gradientflow} in the persistence diagrams space, entropic regularization \cite{ot:cuturi2013sinkhorn} of metrics, use of diagram metrics in (deep) learning pipelines \cite{ot:genevay2018learning}, along with the use of well developed libraries \cite{pot}.

\appendix

\clearpage
\section{Elements of measure theory}
\label{subsec:background_measures}
In the following, $\Omega$ denotes a locally compact Polish metric space (i.e.~a Polish space equipped with a distinguished Polish metric).
\begin{definition} 
The space $\MM(\Omega)$ of Radon measures supported on $\Omega$ is the space of Borel measures which give finite mass to every compact set of $\Omega$. The vague topology on $\MM(\Omega)$ is the coarsest topology such that the maps $\mu \mapsto \mu(f)\defeq \int f\dd \mu$ are continuous for every $f \in C_c(\Omega)$, the space of continuous functions with compact support in $\Omega$. 
\end{definition}

Radon measures on a general space are also required to be regular (i.e.~well approximated by above by open sets and by below by compact sets). However, on a locally compact Polish metric space (such as $\Omega$), regularity is implied by the above definition (see \cite[Section 7.1 and Theorem 7.8]{folland2013real} for details).
\begin{definition}
	Denote by $\MM_f(\Omega)$ the space of finite Borel measures on $\Omega$. The weak topology on $\MM_f(\Omega)$ is the coarsest topology such that the maps $\mu \mapsto \mu(f)$ are continuous for every $f \in C_b(\Omega)$, the space of continuous bounded functions in $\Omega$. 
\end{definition}

See \cite[Chapter 8]{bogachev2007measure} for more details on the weak topology on the set of finite Borel measures (which coincide with the set of Baire measures for $\Omega$ a metrizable space).
We denote by $\cvvague$ the vague convergence and $\cvweak$ the weak convergence.

\begin{definition} A set $F \subset \MM(\Omega)$ is said to be tight if, for every $\epsilon >0$, there exists a compact set $K$ with $\mu(\Omega\backslash K)\leq \epsilon$ for every $\mu\in F$.
\end{definition}

The following propositions are standard results. Corresponding proofs can be found for instance in \cite[Section 15.7]{kallenberg2017random}.

\begin{proposition}\label{prop:rel_cpt_vague}
 A set $F \subset \MM(\Omega)$ is relatively compact for the vague topology if and only if for every compact set $K$ included in $\Omega$, 
\[
\sup\{\mu(K),\ \mu \in F\} < \infty.
\]
\end{proposition}

\begin{proposition}[Prokhorov's theorem]\label{prop:rel_cpt_weak} A set $F \subset \MM_f(\Omega)$ is relatively compact for the weak topology if and only if $F$ is tight and $\sup_{\mu \in F} \mu(\Omega) < \infty$.
\end{proposition}

\begin{proposition}\label{prop:weak_vague} Let $\mu,\mu_1,\mu_2,\dots$ be measures in $\MM_f(\Omega)$. Then, $\mu_n \cvweak \mu$ if and only if $\mu_n(\Omega) \to \mu(\Omega)$ and $\mu_n \cvvague \mu$.
\end{proposition}

\begin{proposition}[The Portmanteau theorem]\label{prop:portemanteau} Let $\mu,\mu_1,\mu_2,\dots$ be measures in $\MM(\Omega)$. Then, $\mu_n \cvvague \mu$ if and only if one of the following propositions holds:
\begin{itemize}
\item for all open sets $U\subset \Omega$ and all \emph{bounded} closed sets $F\subset \Omega$ ,
\[\limsup_{n \to \infty} \mu_n(F) \leq \mu(F) \text{ and }\liminf_{n\to \infty} \mu_n(U) \geq \mu(U).\]
\item for all \emph{bounded} Borel sets $A$ with $\mu(\partial A)= 0$, $\displaystyle \lim_{n\to \infty} \mu_n(A) = \mu(A)$.
\end{itemize}
\end{proposition}

\begin{definition}\label{def:point_measure}
	The set of \emph{point measures} on $\Omega$ is the subset $\PD(\Omega) \subset \MM(\Omega)$ of Radon measures with discrete support and integer mass on each point, that is of the form
	\[
		\sum_{x \in X} n_x \delta_{x}
	\]
	where $n_x \in \N$ and $X \subset \Omega$ is some locally finite set.
\end{definition}

\begin{proposition}The set $\PD(\Omega)$ is closed in $\MM(\Omega)$ for the vague topology.
\label{prop:diagram_close} 
\end{proposition}

We now discuss properties of the VM topology. Recall that $E_\Omega = (\overline \Omega\times \overline \Omega)\backslash (\thediag\times \thediag)$, where $\Omega = \{(t_1,t_2)\in \R^2,\ t_2>t_1\}$.  
\begin{proposition}\label{prop:compact_VM}
    A subset $F\subset \MM(E_\Omega)$ is relatively compact for the VM topology if and only if the sets $\{\pi_1,\ \pi \in F\}$ and $\{\pi_2,\ \pi\in F\}$ are relatively compact for the vague topology, where $\pi_1$ (resp. $\pi_2$) is the first (resp. second) marginal of a measure $\pi$.
\end{proposition}

\begin{proof}
By definition of the initial topology, a subset $F\subset \MM(E_\Omega)$ is relatively compact for the VM topology if and only if $I(F)$ is relatively compact: according to Proposition \ref{prop:rel_cpt_vague}, this is equivalent to having for all compact sets $K\subset \Omega$ and $L\subset E_\Omega$
\[ \max(\sup\{\pi_1(K),\ \pi \in F\},\sup\{\pi_2(K),\ \pi \in F\}, \sup\{\pi(L),\ \pi \in F\} )< \infty. \]
Any compact set $L$ in $E_\Omega$ is included in a set of the form $(K\times \Omega)\cup (\Omega\times K)$ for some compact set $K\subset \Omega$. 
Hence, 
\begin{align*}
   \sup_{\pi\in F} \pi(L)\leq   \sup_{\pi\in F} \pi((K\times \Omega)\cup (\Omega\times K))&\leq \sup_{\pi\in F} (\pi_1(K) + \pi_2(K))\\
    &\leq  \sup_{\pi\in F} \pi_1(K) + \sup_{\pi\in F}\pi_2(K).
\end{align*}
Hence, the set $I(F)$ is relatively compact  if and only if for all compact sets $K\subset \Omega$, $ \sup_{\pi\in F} \pi_1(K) <+\infty$ and $\sup_{\pi\in F}\pi_2(K)<+\infty$, which is equivalent to the two sets of first and second marginals being relatively compact for the vague topology.
\end{proof}

\section{Delayed proofs of Section \ref{sec:structure_dgm_space}}\label{sec:proofs} 

For the sake of completeness, we present in this section proofs which either require very few adaptations from corresponding proofs in \cite{ot:figalli2010newTransportationDistance} or which are close to standard proofs in optimal transport theory.

\begin{proof}[Proofs of Proposition \ref{prop:basic_prop} and Proposition \ref{prop:basic_infty}]~
\begin{itemize}
\item For any $\pi\in \Adm(\mu,\nu)$, $\pi_1=\mu$ and $\pi_2=\nu$. Hence, Proposition \ref{prop:compact_VM} directly implies that the set $\Adm(\mu,\nu)$ is relatively compact for the VM topology. Moreover, if a sequence $(\pi_n)$ in $\Adm(\mu,\nu)$ converges for the VM topology to some $\pi\in \MM(E_\Omega)$, then, by definition, the sequence of the first (resp. second) marginals of $(\pi_n)_n$ converges vaguely to the first (resp. second) marginal of $\pi$. But this sequence is constant equal to $\mu$ (resp. $\nu$). Hence, the first marginal of $\pi$ is $\mu$ (resp. the second marginal of $\pi$ is $\nu$). We have proved that $\Adm(\mu,\nu)$ is also closed: hence, it is a compact set for the VM topology.
	
\item To prove the second point of Proposition \ref{prop:basic_prop}, consider $\pi, \pi_1, \pi_2, \dots$ such that $\pi_n \to\pi$ for the VM topology, and introduce $\pi'_n : A \mapsto \iint_A d(x,y)^p \dd \pi_n$. The sequence $(\pi'_n)_n$  converges vaguely to $\pi ' : A \mapsto \iint_{A} d(x,y)^p \dd \pi$. The Portmanteau theorem (Proposition \ref{prop:portemanteau}) applied with the open set $E_\upperdiag$ to the measures $\pi'_n \cvvague \pi'$ implies that 
\[C_p(\pi) =\pi'(E_\upperdiag)  \leq \liminf_n \pi'_n (E_\upperdiag) = \liminf_n C_p(\pi_n),\] i.e.~$C_p$ is lower semi-continuous.

\item We now prove the lower semi-continuity of $C_\infty$. Let $(\pi_n)_n$ be a sequence converging to $\pi$ for the VM topology and let $\displaystyle r >\liminf_{n\to \infty} C_\infty(\pi_n)$. The set $U_r = \{(x,y) \in E_\upperdiag, \ d(x,y) > r \}$ is open. By the Portmanteau theorem (Proposition \ref{prop:portemanteau}), we have \[0=\liminf_{n\to \infty}  \pi_n(U_r) \geq \pi(U_r).\] Therefore, $\supp(\pi) \subset U_r^c$ and $C_\infty(\pi) \leq r$. As this holds for any $\displaystyle r>\liminf_{n\to \infty}  C_\infty(\pi_n)$, we have $\displaystyle \liminf_{n\to \infty}  C_\infty(\pi_n) \geq C_\infty(\pi)$. 

\item We show that for any $1\leq p \leq \infty$, the lower semi-continuity of $C_p$ and the sequential compactness of $\Adm(\mu,\nu)$ with respect to the VM topology imply that \ref{it:opt_compact_proof}. $\Opt_p(\mu,\nu)$ is a non-empty compact set for the vague topology on $E_\upperdiag$ and that \ref{it:Dp_lsc_proof}. $\Dp$ is lower semi-continuous.
\begin{enumerate}
\item Let $(\pi_n)_n$ be a minimizing sequence of \eqref{eq:optimalCostPbm} or \eqref{eq:def_dinfty} in $\Adm(\mu,\nu)$. As $ \Adm(\mu,\nu)$ is sequentially compact, it has a subsequential limit $\pi$, and the lower semi-continuity implies that $C_p(\pi) \leq \liminf_{n\to \infty} C_p(\pi_n) = \Dp^p(\mu,\nu)$, so that $\Opt_p(\mu,\nu)$ is non-empty. Using once again the lower semi-continuity of $C_p$, if a sequence in $\Opt_p(\mu,\nu)$ converges to some limit, then the cost of the limit is less than or equal to (and thus equal to) $\Dp^p(\mu,\nu)$, i.e.\ the limit is in $\Opt_p(\mu,\nu)$. The set $\Opt_p(\mu,\nu)$ being closed in the sequentially compact set $\Adm(\mu,\nu)$, it is also sequentially compact. \label{it:opt_compact_proof}
\item \label{it:Dp_lsc_proof} Let $\mu_n \cvvague \mu$  and $\nu_n \cvvague \nu$. One has $\liminf_n \Dp(\mu_n,\nu_n)=\lim_k \Dp(\mu_{n_k},\nu_{n_k})$ for some subsequence $(n_k)_k$. For ease of notation, we will still use the index $n$ to denote this subsequence. If the limit is infinite, there is nothing to prove. Otherwise, consider $\pi_n \in \Opt_p(\mu_n,\nu_n)$, and note that  for any compact sets $K \subset \upperdiag$, one has $ \sup_n \mu_n(K) + \sup_n \nu_n(K') < \infty$. Therefore, by Proposition \ref{prop:compact_VM}, there exists a subsequence $(\pi_{n_k})_k$ which converges  to some measure $\pi \in \Adm(\mu,\nu)$ for the VM topology. Note that the first (resp.~second) marginal of $\pi$ is equal to the limit $\mu$ (resp.~$\nu$) of the first (resp.~second) marginal of $(\pi_{n_k})$, so that $\pi$ is in $\Adm(\mu,\nu)$. Therefore,
\[
\Dp^p(\mu,\nu) \leq C_p(\pi) \leq \liminf_{n\to \infty} C_p(\pi_n) =\liminf_{n\to \infty} \Dp^p(\mu_n,\nu_n). 
\]
\end{enumerate}
\item Finally, we prove that $\Dp$ is a metric on $\MM^p$. Let $\mu,\nu,\lambda\in \MM^p$. The symmetry of $\Dp$ is clear. If $\Dp(\mu,\nu) = 0$, then there exists $\pi \in \Adm(\mu,\nu)$ supported on $\{(x,x),\ x\in\Omega\}$. Therefore, for a Borel set $A \subset \Omega$, $\mu(A) = \pi(A \times \groundspace) = \pi(A \times A)=\pi(\groundspace \times A)=\nu(A)$, and $\mu = \nu$. To prove the triangle inequality, we need a variant on the gluing lemma, stated in \cite[Lemma 2.1]{ot:figalli2010newTransportationDistance}: for $\pi_{12} \in \Opt(\mu,\nu)$ and $\pi_{23} \in \Opt(\nu,\lambda)$ there exists a measure $\gamma \in \MM(\groundspace^3)$ such that the marginal corresponding to the first two entries (resp.~two last entries), when restricted to $E_\Omega$, is equal to $\pi_{12}$ (resp.~$\pi_{23}$), and induces a zero cost on $\thediag \times \thediag$. Therefore, by the triangle inequality and the Minkowski inequality,
\begin{align*}
\Dp(\mu,\lambda) &\leq \left( \int_{\groundspace^2} d(x,z)^p\dd \gamma(x,y,z) \right)^{1/p} \\
&\leq \left( \int_{\groundspace^2} d(x,y)^p\dd \gamma(x,y,z) \right)^{1/p} + \left( \int_{\groundspace^2} d(y,z)^p\dd \gamma(x,y,z) \right)^{1/p} \\
&= \left( \int_{\groundspace^2} d(x,y)^p\dd \pi_{12}(x,y) \right)^{1/p} + \left( \int_{\groundspace^2} d(y,z)^p\dd \pi_{23}(y,z) \right)^{1/p} \\
&= \Dp(\mu,\nu) + \Dp(\nu,\lambda).
\end{align*}
The proof is similar for $p= \infty$.
\end{itemize}
\end{proof}

\begin{proof}[Proof of Proposition \ref{prop:MM_p_Polish}]
We first show the separability. Consider for $k>0$ a partition of $\Omega$ into squares $(C_i^k)$ of side length $2^{-k}$, centered at points $x_i^k$. Let $F$ be the set of all measures of the form $\sum_{i\in I} q_i \delta_{x_i^k}$ for $q_i$ positive rationals, $k>0$ and $I$ a finite subset of $\N$. Our goal is to show that the countable set $F$ is dense in $\MM^p$. Fix $\epsilon > 0$, and $\mu \in \MM^p$. The proof is in three steps.
\begin{enumerate}
\item Since $\Pers_p(\mu) < \infty$, there exists a compact $K \subset \upperdiag$ such that $\Pers_p(\mu) - \Pers_p(\mu_0) < \epsilon^p$, where $\mu_0$ is the restriction of $\mu$ to $K$. By considering the transport plan between $\mu$ and $\mu_0$ induced by the identity map on $K$ and the projection onto the diagonal on $\groundspace \backslash K$, it follows that $\Dp^p(\mu,\mu_0) \leq \Pers_p(\mu) - \Pers_p(\mu_0) \leq \epsilon^p$.
\item Consider $k$ such that $2^{-k} \leq \epsilon / (\sqrt{2}\mu(K)^{1/p})$ and denote by $I$ the indices corresponding to squares $C_i^k$ intersecting $K$. Let $\mu_1 = \sum_{i\in I}^\infty \mu_0(C_i^k) \delta_{x_i^k}$. One can create a transport map between $\mu_0$ and $\mu_1$ by mapping each square $C_i^k$ to its center $x_i^k$, so that
\[ \Dp(\mu_0,\mu_1) \leq \left(\sum_{i} \mu_0(C_i^k) (\sqrt{2}\cdot 2^{-k})^p \right)^{1/p} \leq \mu(K)^{1/p} \sqrt{2}\cdot 2^{-k} \leq \epsilon.\]
\item Consider, for $i \in I$, $q_i$ a rational number satisfying $q_i \leq \mu_0(C_i^k)$ and $|\mu_0(C_i^k) - q_i| \leq \epsilon^p/\left(\sum_{i\in I} d(x_i^k,\thediag)^p \right)$. Let $\mu_2 = \sum_{i\in I} q_i\delta_{x_i^k}$. Consider the transport plan between $\mu_2$ and $\mu_1$ that fully transports $\mu_2$ onto $\mu_1$, and transport the remaining mass in $\mu_1$ onto the diagonal. Then, \[\Dp(\mu_1,\mu_2) \leq  \left(\sum_{i\in I} |\mu_0(C_i^k) - q_i|  d(x_i^k,\thediag)^p \right)^{1/p} \leq  \epsilon.\]
\end{enumerate}
 As $\mu_2 \in F$ and $\Dp(\mu,\mu_2) \leq 3 \epsilon$, the separability is proven.

To prove that the space is complete, consider a Cauchy sequence $(\mu_n)_n$. As the sequence $(\Pers_p(\mu_n))_n = (\Dp^p(\mu_n,0))_n$ is a Cauchy sequence, it is bounded. Therefore, for $K\subset \upperdiag$ a compact set, \eqref{eq:pers_compact} implies that $\sup_n\mu_n(K)<\infty$. Proposition \ref{prop:rel_cpt_vague} implies that $(\mu_n)_n$ is relatively compact for the vague topology on $\upperdiag$. Consider $(\mu_{n_k})_k$ a subsequence converging vaguely on $\upperdiag$ to some measure $\mu$. By the lower semi-continuity of $\Dp$, 
\[ \Pers_p(\mu) = \Dp^p(\mu,0) \leq \liminf_{k \to \infty} \Dp^p(\mu_{n_k},0) < \infty,\]
so that $\mu \in \MM^p$. Using once again the lower semi-continuity of $\Dp$,
\begin{align*}
\Dp(\mu_n,\mu) &\leq  \liminf_{k \to \infty} \Dp(\mu_n,\mu_{n_k}) \\
\lim_{n\to\infty} \Dp(\mu_n,\mu) &\leq \lim_{n\to\infty} \liminf_{k \to \infty} \Dp(\mu_n,\mu_{n_k})=0,
\end{align*} 
ensuring that $\Dp(\mu_n, \mu) \rightarrow 0$, that is the space is complete. 
\end{proof}

\begin{proof}[Proof of the direct implication of Theorem \ref{thm:conv_dp}]
Let $\mu,\mu_1,\mu_2,\dots$ be elements of $\MM^p$ and assume that the sequence $(\Dp(\mu_n,\mu))_n$ converges to 0. The triangle inequality implies that $\Pers_p(\mu_n)=\Dp^p(\mu_n,0)$ converges to $\Pers_p(\mu)=\Dp^p(\mu,0)$. Let $f\in C_c(\upperdiag)$, whose support is included in some compact set $K$. For any $\epsilon >0$, there exists a Lipschitz function $f_\epsilon$, with Lipschitz constant $L$ and whose support is included in $K$, with the $\infty$-norm $\|f-f_\epsilon\|_\infty$ less than or equal to $\epsilon$. The convergence of $\Pers_p(\mu_n)$ and \eqref{eq:pers_compact} imply that $\sup_k \mu_k(K) < \infty$. Let $\pi_n \in \Opt_p(\mu_n, \mu)$, we have 
\begin{align*}
|\mu_n(f)-\mu(f)| &\leq |\mu_n(f-f_\epsilon)| + |\mu(f-f_\epsilon)|  + |\mu_n(f_\epsilon)-\mu(f_\epsilon)| \\
&\leq (\mu_n(K) + \mu(K))\epsilon + |\mu_n(f_\epsilon)-\mu(f_\epsilon)| \\
&\leq (\sup_k \mu_k(K) + \mu(K))\epsilon + |\mu_n(f_\epsilon)-\mu(f_\epsilon)|.
\end{align*}
Also, 
\begin{align*}
|\mu_n(f_\epsilon)-\mu(f_\epsilon)| &\leq \iint_{\groundspace^2} |f_\epsilon(x)-f_\epsilon(y)| \dd\pi_n(x,y) \quad \text{ where } \pi_n \in \Opt(\mu_n, \mu) \\
&\leq  L \iint\limits_{(K \times \groundspace) \cup (\groundspace \times K)} d(x,y)\dd\pi_n(x,y) \\
&\leq L \pi_n((K \times \groundspace) \cup (\groundspace \times K))^{1- \frac{1}{p}}\left( \iint\limits_{(K \times \groundspace) \cup (\groundspace \times K)} d(x,y)^p \dd\pi_n(x,y)\right)^{\frac{1}{p}} \\
& \qquad \text{ by H\" older's inequality.}\\
&\leq  L \left(\sup_k \mu_k(K) + \mu(K)\right)^{1- \frac{1}{p}} \Dp(\mu_n,\mu)\xrightarrow[n\to\infty]{} 0.
\end{align*}
Therefore, taking the limsup in $n$ and then letting $\epsilon$ goes to $0$, we obtain that $\mu_n(f) \to \mu(f)$. 
\end{proof}

\section{Proofs of the technical lemmas of Section \ref{sec:barycenter}}
\label{sec:proof_barycenters}

The following proof is already found in  \cite{ot:leGouic2016existenceConsistencyWB}. We reproduce it here for the sake of completeness.
\begin{proof}[Proof of Lemma  \ref{lem:rewrite}]
Recall that $\P_n$ is a sequence in $\WW^p(\MM^p)$ such that each $\P_n$ has a $p$-Fr\'echet mean $\mu_n$ and that $W_{p,\Dp}(\P_n,\P) \to 0$  for some $\P\in  \WW^p(\MM^p)$. According to the beginning of the proof of Proposition \ref{prop:consistencyBary}, the sequence $(\mu_n)_n$ is relatively compact for the vague convergence. Let $\nu\in \MM^p$ and let $\mu$ be the vague limit of some subsequence, which, for ease of notations, will be denoted as the initial sequence. By Skorokhod's representation theorem  \cite[Theorem 6.7]{billingsley2013convergence}, as $\P_n$ converges weakly to $\P$, there exists a probabilistic space on which are defined random variables $\bm{\mu}\sim \P$ and  $\bm{\mu_n}\sim \P_n$ for $n\geq 0$, such that $\bm{\mu_n}$ converges almost surely with respect to the $\Dp$ metric towards $\bm \mu$. Using those random variables, we have
\begin{equation}\label{eq:rewrite1}
\begin{split}
\EE(\nu) &= \E \Dp^p(\nu,\bm \mu) = W_{p,\Dp}^p(\delta_{\nu},\P) \\
&= \lim_n W_{p,\Dp}^p(\delta_{\nu},\P_n) \text{ since $ W_{p,\Dp}(\P_n,\P)\to 0$}\\
&=\lim_n \E  \Dp^p(\nu,\bm{\mu_n}) \\
&\geq \lim_n \E  \Dp^p(\mu_n,\bm{\mu_n}) \text{ since $\mu_n$ is a barycenter of $\P_n$}\\
&\geq \E  \liminf_n \Dp^p(\mu_n,\bm{\mu_n}) \text{ by Fatou's lemma}\\
&\geq \E  \Dp^p(\mu,\bm{\mu}) = \EE(\mu) \text{ by lower semi-continuity of $\Dp$ (Prop. \ref{prop:basic_prop}).}
\end{split}
\end{equation}
This implies that $\mu$ is a barycenter of $\P$. We are now going to show that, almost surely, $\liminf_n \Dp(\mu_n,\bm\mu)=\Dp(\mu,\bm\mu)$. This concludes the proof by letting $n_k$ be the subsequence attaining the liminf for some fixed realization of $\bm\mu$.  By plugging in $\nu=\mu$ in \eqref{eq:rewrite1}, all the inequalities become equalities, and in particular,
\[ \lim_n W_{p,\Dp}^p(\delta_{\mu_n},\P_n)= \lim_n \E  \Dp^p(\mu_n,\bm{\mu_n})=\E  \Dp^p(\mu,\bm{\mu})=W_{p,\Dp}^p(\delta_{\mu},\P).\]
This yields
\begin{align*}
&0\leq  W_{p,\Dp}(\delta_{\mu_n},\P)-W_{p,\Dp}(\delta_{\mu},\P)\\
&\qquad\qquad \leq  W_{p,\Dp}(\delta_{\mu_n},\P_n) + W_{p,\Dp}(\P_n,\P) -W_{p,\Dp}(\delta_{\mu},\P)\to 0
\end{align*} 
as $n$ goes to $+\infty$, i.e.~$\lim_n W_{p,\Dp}(\delta_{\mu_n},\P)=W_{p,\Dp}(\delta_{\mu},\P)$. Therefore,
\begin{align*}
\E \Dp^p(\mu,\bm \mu) &= W_{p,\Dp}^p(\delta_{\mu},\P)=\lim_n W_{p,\Dp}^p(\delta_{\mu_n},\P)=\lim_n \E \Dp^p(\mu_n,\bm \mu)\\
&\geq \E  \liminf_n \Dp^p(\mu_n,\bm{\mu}) \text{ by Fatou's lemma}\\
&\geq \E  \Dp^p(\mu,\bm{\mu}) \text{ by lower semi-continuity of $\Dp$.}
\end{align*}
As $\liminf_n \Dp^p(\mu_n,\bm{\mu})\geq \Dp^p(\mu,\bm{\mu})$ and $\E  \liminf_n \Dp^p(\mu_n,\bm{\mu})=\E \Dp^p(\mu,\bm \mu)$, we actually have $\liminf_n \Dp^p(\mu_n,\bm{\mu})= \Dp^p(\mu,\bm{\mu})$, concluding the proof. 
\end{proof}

\begin{figure}[h]
\centering
  \includegraphics[width=5cm]{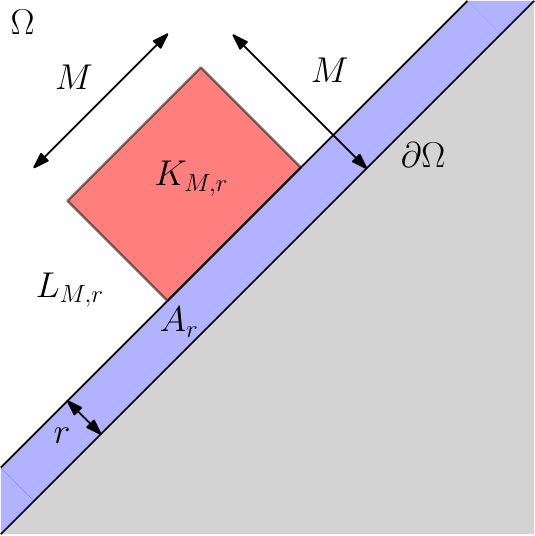}
  \caption{Partition of $\upperdiag$ used in the proof of Lemma \ref{lemma:Skorokhod}.}\label{fig:partition_groundspace}
\end{figure}
\begin{proof}[Proof of Lemma \ref{lemma:Skorokhod}]
For the direct implication, take $\nu =0$ and apply Theorem \ref{thm:conv_dp}. 

Let us prove the converse implication. Assume that $\mu_n \cvvague \mu$ and $\Dp(\mu_n,\nu) \to \Dp(\mu,\nu)$ for some $\nu \in \PD^p$. The vague convergence of $(\mu_n)_n$ implies that $\mu^{(p)}$ is the only possible accumulation point for weak convergence of the sequence $(\mu^{(p)}_n)_n$. Therefore, it is sufficient to show that the sequence $(\mu^{(p)}_n)_n$ is relatively compact for weak convergence (i.e.\ tight and bounded in total variation, see Proposition \ref{prop:rel_cpt_weak}). Indeed, this would mean that $(\mu^{(p)}_n)$ converges weakly to $\mu^{(p)}$, or equivalently by Proposition \ref{prop:weak_vague} that $\mu_n \cvvague \mu$ and $\Pers_p(\mu_n) \to \Pers_p(\mu)$. The conclusion is then obtained thanks to Theorem \ref{thm:conv_dp}.

Thus, let $(\mu_n)_n$ be any subsequence and $(\pi_n)_n$ be corresponding optimal transport plans between $\mu_n$ and $\nu$. The vague convergence of $(\mu_n)_n$ implies that $(\pi_n)_n$ is relatively compact with respect to the VM topology. Let $\pi$ be a limit of any converging subsequence of $(\pi_n)_n$, which indexes are still denoted by $n$. One can prove that $\pi \in \Opt(\mu,\nu)$ (see \cite[Prop.~2.3]{ot:figalli2010newTransportationDistance}). For $r>0$, define $A_r \defeq \{x \in \upperdiag,\ d(x,\thediag)\leq r\}$ and write $\overline{A}_r$ for $A_r \cup \thediag$.
Consider $\eta >1$. We can write
\begin{align*}
&\int_{A_r} d(x, \thediag)^p \dd\mu_n(x) = \iint\limits_{A_r \times \groundspace} d(x, \thediag)^p \dd\pi_n(x,y) \\
&= \iint\limits_{A_r \times (\upperdiag\backslash A_{\eta r})} d(x, \thediag)^p \dd\pi_n(x,y) + \iint\limits_{\overline A_r\times \overline{A}_{\eta r}} d(x,\thediag)^p \dd\pi_n(x,y) \\
&\stackrel{(*)}{\leq} \frac{1}{(\eta-1)^p} \iint\limits_{A_r \times(\upperdiag\backslash A_{\eta r})} d(x,y)^p \dd\pi_n(x,y) + \iint\limits_{\overline A_r\times \overline{A}_{\eta r}} d(x,\thediag)^p \dd\pi_n(x,y) \\
&\leq \frac{1}{(\eta-1)^p} \Dp^p(\mu_n,\nu) + 2^{p-1}\Bigg(\hspace{0.3cm}\iint\limits_{\overline{A}_r\times \overline{A}_{\eta r}} d(x,y)^p \dd\pi_n(x,y) + \iint\limits_{\overline{A}_r\times \overline{A}_{\eta r}} d(y,\thediag)^p \dd\pi_n(x,y)\Bigg)  \\
&\leq \frac{1}{(\eta-1)^p} \Dp^p(\mu_n,\nu) + 2^{p-1}\Bigg( \Dp^p(\mu_n,\nu) - \hspace{-0.6cm} \iint\limits_{E_\Omega  \backslash  (\overline{A}_r\times \overline{A}_{\eta r})} \hspace{-0.5cm} d(x,y)^p \dd\pi_n(x,y) + \int_{ A_{\eta r}} \hspace{-0.3cm} d(y,\thediag)^p \dd\nu(y)\Bigg)
\end{align*}
where $(*)$ holds because $d(x,y)\geq (\eta -1) r \geq (\eta -1) d(x,\thediag)$ for $(x,y)\in A_r \times A_{\eta r}^c$.  Therefore,
\begin{align*}
	\limsup_{n \to \infty} \int_{A_r} d(x,\thediag)^p \dd\mu_n(x) \leq & \frac{1}{(\eta-1)^p} \Dp^p(\mu,\nu)+ 2^{p-1}\Bigg(\Dp^p(\mu,\nu)  \\
	&\qquad  - \hspace{-0.3cm}\iint\limits_{E_\Omega  \backslash  (\overline{A}_r\times \overline{A}_{\eta r})} d(x,y)^p \dd\pi(x,y) + \int_{ A_{\eta r}} d(y,\thediag)^p \dd\nu(y)\Bigg)
\end{align*}
Note that at the last line, we used the Portmanteau theorem (see Proposition~\ref{prop:portemanteau}) on the sequence of measures $(d(x,y)^p \dd\pi_n(x,y))_n$ for the open set $E_\Omega  \backslash  (\overline{A}_r\times \overline{A}_{\eta r})$. Letting $r$ goes to $0$, then $\eta$ goes to infinity, one obtains
\[\lim_{r \to 0} \limsup_{n \to \infty} \int_{A_r} d(x,\thediag)^p \dd\mu_n(x) = 0.\]

The second part consists in showing that there can not be mass escaping ``at infinity'' in the subsequence $(\mu^{(p)}_n)_n$. Fix $r,M>0$. For $x \in \upperdiag$, denote $s(x)$ the projection of $x$ on $\thediag$. Pose 
\[K_{M,r} \defeq \{x \in \upperdiag\backslash A_r,\ d(x,\thediag) < M, d(s(x),0) < M\}\] 
and $L_{M,r}$ the closure of $\upperdiag \backslash (A_r\cup K_{M,r})$ (see Figure \ref{fig:partition_groundspace}). 
For $r'>0$,
\begin{align*}
&\int_{L_{M,r}} d(x,\thediag)^p \dd\mu_n(x) = \iint\limits_{L_{M,r} \times \groundspace} d(x,\thediag)^p \dd\pi_n(x,y) \\
&= \iint\limits_{L_{M,r} \times (L_{M/2,r'}\cup \overline{A}_{r'})} d(x,\thediag)^p \dd\pi_n(x,y)  + \iint\limits_{L_{M,r} \times K_{M/2,r'}}  d(x,\thediag)^p \dd\pi_n(x,y) \\
&\leq 2^{p-1} \iint\limits_{L_{M,r} \times (L_{M/2,r'} \cup \overline{A}_{r'})} d(x,y)^p \dd\pi_n(x,y) \\
&\ + 2^{p-1}  \iint\limits_{L_{M,r} \times( L_{M/2,r'}\cup \overline{A}_{r'})} d(\thediag, y)^p \dd\pi_n(x,y) \\
&\ + \iint\limits_{L_{M,r} \times K_{M/2,r'}}  d(x,\thediag)^p \dd\pi_n(x,y).
\end{align*}
We treat the three parts of the sum separately. As before, taking the $\limsup$ in $n$ and letting $M$ goes to $\infty$, the first part of the sum converges to 0 (apply the Portmanteau theorem on the open set $E_\upperdiag \backslash (L_{M,r} \times (L_{M/2,r'} \cup \overline{A}_{r'}))$. The second part is less than or equal to
\[ 2^{p-1} \int_{ L_{M/2,r'}\cup A_{r'}} d(y,\thediag)^p \dd\nu(y),\]
 which converges to $0$ as $M\to \infty$ and $r'\to 0$. For the third part, notice that if $(x,y)\in L_{M,r} \times K_{M/2,r'}$, then
\begin{align*}
d(x,\thediag) \leq d(x,s(y)) \leq d(x,y) +d(y,s(y)) \leq d(x,y) + \frac{M}{2} \leq 2d(x,y).
\end{align*}
Therefore,
\begin{align*}
\iint\limits_{L_{M,r} \times K_{M/2,r'}}  d(x,\thediag)^p \dd\pi_n(x,y) & \leq 2^p \iint\limits_{L_{M,r} \times K_{M/2,r'}}  d(x,y)^p \dd\pi_n(x,y) \\
& \leq 2^p \iint\limits_{L_{M,r} \times \groundspace}  d(x,y)^p \dd\pi_n(x,y).
\end{align*}
As before, it is shown that $\limsup_n \iint_{L_{M,r} \times \groundspace}  d(x,y)^p \dd\pi_n(x,y)$ converges to $0$ when $M$ goes to infinity by applying the Portmanteau theorem on the open set $E_\upperdiag \backslash (L_{M,r} \times \groundspace)$.

Finally, we have shown, that by taking $r$ small enough and $M$ large enough, one can find a compact set $\overline{K_{M,r}}$ such that $\int_{\upperdiag \backslash \overline{K_{M,r}}} d(x,\thediag)^p \dd\mu_n = \mu^{(p)}_n(\upperdiag \backslash \overline{K_{M,r}})$ is uniformly small: $(\mu^{(p)}_n)_n$ is tight. As we have 
\begin{align*}
 \mu^{(p)}_n(\upperdiag) &= \Pers_p(\mu_n) = \Dp^p(\mu_n, 0) \\
 &\leq (\Dp(\mu_n, \nu) + \Dp(\nu, 0))^p \to (\Dp(\mu, \nu) + \Dp(\nu,0))^p, 
\end{align*}
it is also bounded in total variation. Hence, $(\mu^{(p)}_n)_n$ is relatively compact for the weak convergence: this concludes the proof. 
\end{proof}

\begin{proof}[Proof of Lemma \ref{lemma:tum}]
Let $\P = \sum_{i=1}^N \lambda_i \delta_{a_i}$ a probability distribution with $a_i \in \PD_f$ of mass $m_i \in \N$, and define $\mtot = \sum_{i=1}^N m_i$.
By Proposition \ref{prop:equiv_minimum_functionals}, every $p$-Fr\'echet mean $a$ of $\P$ is in correspondence with a $p$-Fr\'echet mean for the Wasserstein distance $\tilde a$ of  $\tilde{\P} = \sum_{i=1}^N \lambda_i \delta_{\tilde{a}_i}$, where $\tilde{a}_i = a_i + (\mtot - m_i)\delta_{\thediag}$, with $a$ being the restriction  of $\tilde{a}$ to $\upperdiag$.

	Let thus fix $m \in \N$, and let $\tilde{a}_1,\dots,\tilde{a}_N$ be point measures of mass $m$ in $\tilde{\Omega}$. Write $\tilde{a}_i =\sum_{j=1}^{m} \delta_{x_{i,j}}$, so that $x_{i,j} \in \tilde{\Omega}$ for $1 \leq i \leq N,\ 1 \leq j \leq m$, with the $x_{i,j}$s non-necessarily distinct. Define 
	\begin{equation} 
	T:  (x_1, \dots, x_N) \in \tilde{\upperdiag}^N \mapsto \argmin\left\{ \sum_{i=1}^N \lambda_i \rho(x_i,y)^p,\ y\in \tilde{\Omega} \right\} \in \tilde{\upperdiag}.
	\end{equation}
	Since we assume $p > 1$, $T$ is well-defined and is continuous (the minimizer is unique by strict convexity). Using the localization property stated in \cite[Section 2.2]{ot:carlier2015numerical}, we know that the support of a $p$-Fr\'echet mean of $\tilde P$ is included in the finite set
	\[ S \defeq \{ T(x_{1,j_1},\dots,x_{N,j_N}),\ 1\leq j_1 ,\dots, j_N \leq m \}.\]

Let $K=m^N$ and let $z_1,\dots,z_K$ be an enumeration of the points of $S$ (with potential repetitions). Denote by $\mathrm{Gr}(z_k)$ the $N$ elements $x_1,\dots,x_N$, with $x_i \in \supp(\tilde{a}_i)$, such that $z_k = T(x_1,\dots,x_N)$. It is explained in \cite[\S 2.3]{ot:carlier2015numerical}, that finding a $p$-Fr\'echet mean of $\tilde{\P}$ is equivalent to finding a minimizer of the problem 
\begin{equation}\label{eq:bary_reformule}
\inf_{(\gamma_1,\dots,\gamma_N) \in \Pi} \sum_{i=1}^N \lambda_i \iint_{\tilde{\Omega}^2} \rho(x_i,y)^p \dd \gamma_i(x_i,y),
\end{equation} 
where $\Pi$  is the set of plans $(\gamma_i)_{i=1,\dots,N}$, with $\gamma_i$ having for first marginal $\tilde{a}_i$, and such that all $\gamma_i$s share the same (non-fixed) second marginal. Furthermore, we can assume without loss of generality that $(\gamma_1 \dots \gamma_N)$ is supported on $(\mathrm{Gr}(z_k), z_k)_k$, i.e.~a point $z_k$ in the $p$-Fr\'echet mean is necessary transported to its corresponding grouping $\mathrm{Gr}(z_k)$ by (optimal) $\gamma_1, \dots \gamma_N$ \cite[\S 2.3]{ot:carlier2015numerical}. For such a minimizer, the common second marginal is a $p$-Fr\'echet mean of $\tilde{\P}$. 

A potential minimizer of \eqref{eq:bary_reformule} is described by a vector $\gamma = (\gamma_{i,j,k})\in \R_+^{NmK}$ such that: 
\begin{equation}\label{eq:constraints}
\begin{cases}
\text{for}\ 1\leq i \leq N,\ 1\leq  j \leq m,\quad & \sum_{k=1}^K \gamma_{i,j,k} = 1 \text{ and}\\
\text{for}\ 2\leq i \leq N,\ 1\leq  k \leq K,\quad &\sum_{j=1}^m \gamma_{1,j,k} = \sum_{j=1}^m \gamma_{i,j,k}.
\end{cases}
\end{equation} 
Let $c \in \R^{NmK}$ be the vector defined by $c_{i,j,k} = \ones\{x_{i,j} \in \mathrm{Gr}(z_k)\} \lambda_i \rho(x_{i,j},z_k)^p $. Then, the problem \eqref{eq:bary_reformule}  is equivalent to 
\begin{equation}\label{eq:reformule_again}
		\minimize_{\gamma \in R_+^{NmK}} \gamma^T c \quad \text{under the constraints \eqref{eq:constraints}}.
		\end{equation}
The set of $p$-Fr\'echet means of $P$ are in bijection with the set of minimizers of this Linear Programming problem (see \cite[\S 5.15]{schrijver2003combinatorial}), which is given by a face of the polyhedron described by the equations \eqref{eq:constraints}. Hence, if we show that this polyhedron is integer (i.e.~its vertices have integer values), then it would imply that the extreme points of the set of $p$-Fr\'echet means of $P$ are point measures, concluding the proof.
 The constraints \eqref{eq:constraints} are described by a matrix $A$ of size $(Nm+(N-1)K) \times NmK$ and a vector $b= [\ones_{Nm},\mathbf{0}_{(N-1)K}]$, such that $\gamma \in \R^{NmK}$ satisfies \eqref{eq:constraints} if and only if $A\gamma=b$. A sufficient condition for the polyhedron $\{Ax\leq b\}$ to be integer is to satisfy the following property (see \cite[Section 5.17]{schrijver2003combinatorial}): for all $u \in \Z^{NmK}$, the dual problem
 \begin{equation} \label{eq:dual_polytope}
 \max\{y^T b,\ y\geq 0 \text{ and } y^TA=u\} 
 \end{equation}
has either no solution (i.e.~there is no $y \geq 0$ satisfying $y^T A = u$), or it has an integer optimal solution $y$. 

For $y$ satisfying $y^T A = u$, write $y=[y^0,y^1]$ with $y^0 \in \R^{Nm}$ and $y^1 \in \R^{(N-1)K}$, so that $y^0$ is indexed on $ 1\leq i \leq N,\ 1\leq  j \leq m$ and $y^1$ is indexed on $ 2\leq i \leq N,\ 1\leq  k \leq K$. One can check that, for $ 2\leq i \leq N,\ 1\leq  j \leq m,\ 1\leq  k \leq K$:
 \begin{equation}\label{eq:constraint_c}
  u_{1,j,k} = y^0_{1,j} + \sum_{i'=2}^N y^1_{i',k} \quad \text{ and } \quad u_{i,j,k} = y^0_{i,j} - y^1_{i,k},
\end{equation} 
so that, 
\begin{align*}
y^T b &= \sum_{i=1}^N \sum_{j=1}^m y^0_{i,j} = \sum_{j=1}^m y^0_{1,j} + \sum_{i=2}^N \sum_{j=1}^m y^0_{i,j} \\
&= \sum_{j=1}^m (u_{1,j,k} - \sum_{i=2}^N y^1_{i,k}) + \sum_{i=2}^N \sum_{j=1}^m (u_{i,j,k} + y^1_{i,k}) \\
&= \sum_{i=1}^N \sum_{j=1}^m u_{i,j,k}.
\end{align*}
Therefore, the function $y^Tb$ is constant on the set $P \defeq \{y\geq 0,\ y^T A=u\}$, and any point of the set is an argmax. We need to check that if the set $P$ is non-empty, then it contains a vector with integer coordinates: this would conclude the proof. A solution of the homogeneous equation $y^T A=0$ satisfies $y^0_{i,j}=y^1_{i,k} = \lambda_i$ for $i\geq 2$ and $y^0_{1,j} = -\sum_{i=2}^N y^1_{i,k} = -\sum_{i=2}^N \lambda_i$ and reciprocally, any choice of $\lambda_i \in \R$ gives rise to a solution of the homogeneous equation. For a given $u$, one can verify that the set of solutions of $y^TA=u$ is given, for $\lambda_i \in \R$, by
\[\begin{cases}
y^0_{1,j} = \sum_{i=1}^N u_{i,j,k} - \sum_{i=2}^N \lambda_i \\
y^0_{i,j} = \lambda_i \text{ for } i\geq 2,\\
y^1_{i,k} = -u_{i,j,k} + \lambda_i \text{ for } i\geq 2.
\end{cases}
\]
Such a solution exists if and only if for all $j$, $U_j \defeq \sum_{i=1}^N u_{i,j,k}$ does not depend on $k$ and for $i\geq 2$, $U_{i,k} \defeq u_{i,j,k}$ does not depend on $j$. For such a vector $u$, $P$ corresponds to the $\lambda_i \geq 0$ with $\lambda_i \geq \max_k  U_{i,k}$ and $U_j \geq \sum_{i=1}^N \lambda_i$. If this set is non empty, it contains as least the point corresponding to $\lambda_i = \max\{0, \max_k U_{i,k}\}$, which is an integer: this point is integer valued, concluding the proof. 
\end{proof}

\section{Technical details regarding Section \ref{subsec:stability_expectedPD}}
\label{sec:details_expectations}

Write $\MM_f$ for $\MM_f(\Omega)$ and define $\MM_\pm$ the space of finite signed measures on $\upperdiag$, i.e.\ a measure $\mu \in \MM_\pm$ is written $\mu_+ - \mu_-$ for two finite measures $\mu_+,\mu_- \in \MM_f$. The total variation distance $|\cdot |$ is a norm on $\MM_\pm$, and $(\MM_\pm,|\cdot |)$ is a Banach space. The Bochner integral \cite{bochner1933integration} is a generalization of the Lebesgue integral for functions taking their values in Banach space. We define the expected persistence measure of $\P \in \WW^p(\MM^p)$ as the Bochner integral of some pushforward  of $\P$. More precisely, recall the definition \eqref{eq:def_mu^p} of $\mu^{(p)}$ and define
\begin{align*}
F: (\MM^p,\Dp) & \to  (\MM_\pm,  |\cdot |) \\
 \mu & \mapsto \mu^{(p)}.
\end{align*}
Note that $F$ has an inverse $G$ on $\MM_f$, defined by $G(\nu)(A) \defeq \int_A d(x,\upperdiag)^{-p}\dd \nu(x)$ for $A\subset \upperdiag$ a Borel set. Theorem \ref{thm:conv_dp} implies that $G$ is a continuous function from $(\MM_f, | \cdot |)$ to $(\MM^p, \Dp)$. In particular, as $\MM_f$  and $\MM^p$ are Polish spaces and $G$ is injective, the map $F$ is measurable (see \cite[Theorem 15.1]{kechris1995classical}). For $\P \in \WW^p(\MM^p(\upperdiag))$, define for $\bm{\mu} \sim \P$, $\E[\bm{\mu}]$ the \emph{linear} expectation of $\P$ by
\begin{equation}\label{eq:epd}
\E[\bm{\mu}] \defeq G\left( \int \nu \dd(F_\#\P)(\nu) \right) \in \MM^p,
\end{equation}
where the integral is the Bochner integral on the Banach space $(\MM_\pm,|\cdot|)$ and $F_\#\P$ is the pushforward of $\P$ by $F$. It is straightforward to check that $\E[\bm{\mu}]$ defined in that way satisfies the relation
\[
\forall K \subset \upperdiag\ \text{compact},\ \E[\bm{\mu}](K) = \E[\bm{\mu}(K)].
\]

The proof of Proposition \ref{prop:Dp_jensen} consists in applying Jensen's inequality in an infinite-dimensional setting. We first show that the function $\Dp^p$ is convex.

\begin{lemma}\label{lem:convex} For $1\leq p < \infty$, the function $\Dp^p : \MM^p \times \MM^p \to \R$ is convex.
\end{lemma}
\begin{proof} Fix $\mu_1,\mu_2,\nu_1,\nu_2 \in \MM^p$ and $t\in [0,1]$. Our goal is to show that
\[ \Dp^p(t\mu_1+(1-t)\mu_2,t\nu_1 + (1-t)\nu_2) \leq t\Dp^p(\mu_1,\nu_1)+(1-t)\Dp^p(\mu_2,\nu_2).\]
Let $\pi_{11} \in \Opt_p(\mu_1,\nu_1)$ and $\pi_{22}\in \Opt_p(\mu_2,\nu_2)$. It is straightforward to check that $\pi \defeq t\pi_{11} +(1-t)\pi_{22}$ is an admissible plan between $t\mu_1+(1-t)\mu_2$ and $t\nu_1 + (1-t)\nu_2$. The cost of this admissible plan is $t\Dp^p(\mu_1,\nu_1) +(1-t)\Dp^p(\mu_2,\nu_2)$, which is therefore larger than $\Dp^p(t\mu_1+(1-t)\mu_2,t\nu_1 + (1-t)\nu_2)$. 
\end{proof}

We then use the following result, which is a particular case of \cite[Theorem 3.10]{perlman1974jensen}.
\begin{proposition}\label{prop:jensen}
Let $\XX$ be a Hausdorff locally convex topological vector space and $C\subset \XX$ a closed convex set. Let $Q$ be a probability measure on $\XX$ endowed with its borelian $\sigma$-algebra, which is supported on $C$. Assume that $\int \|x\| \dd Q(x) <\infty$. Let $f:C\to [0,\infty)$ be a continuous convex function with $\int f(x) \dd Q(x) < \infty$. Then
\[ f\left(\int x \dd Q(x) \right) \leq \int f(x) \dd Q(x).\]
\end{proposition}

Let $\XX=\MM_\pm\times \MM_\pm$ which is a Banach space (endowed with the product norm), and thus in particular a Hausdorff locally convex topological vector space. Let $C = \MM_f \times \MM_f$, which is convex and closed (closedness follows immediately from the definition of the total variation $|\cdot|$) and let $f=\Dp^p \circ (G,G) : \XX \to \R$. The continuity of $G$ implies that $f$ is continuous and Lemma \ref{lem:convex} implies the convexity of $f$. Let $P$, $P'$ be two probability measures in $\WW^p(\MM^p)$ and $\gamma$ be an optimal coupling between $P$ and $P'$. We let $Q$ be the image measure of $\gamma$ by $(F,F)$, so that
\begin{align*}
 \int_{x\in \XX} \|x\| \dd Q(x) &= \int_{\mu,\mu'\in \MM^p} \max(|\mu^{(p)}|,|(\mu')^{(p)}|)\dd \gamma(\mu,\mu') \\
 &\leq  \int_{\mu}\Pers_p(\mu)\dd P(\mu) + \int_{\mu'}\Pers_p(\mu')\dd P'(\mu')<\infty
\end{align*}
and that
\[ \int_{x\in \XX} f(x)\dd Q(x) = \int_{\mu,\mu'\in \MM^p} \Dp^p(\mu,\mu')\dd \gamma(\mu,\mu') =  W_{p,\Dp}^p(P,P') <\infty. \]
Also, we have 
\begin{align*}
\int x \dd Q(x)  &= \int_{\nu,\nu'\in \MM^p} (\nu,\nu')\dd (F,F)_{\#}\gamma(\nu,\nu')\\
&=  \left( \int_{\nu \in \MM^p} \nu\dd F_{\#} P(\nu), \int_{\nu'\in \MM^p} \nu'\dd F_{\#}P'(\nu')\right),
\end{align*}
so that by \eqref{eq:epd}, $f\left(\int x \dd Q(x) \right) = \Dp^p(\E[\bm{\mu}],\E[\bm{\mu'}])$, where $\bm\mu\sim P$ and $\bm \mu' \sim P'$. Proposition \ref{prop:Dp_jensen} yields the conclusion.

\end{document}